\pdfoutput=1
\documentclass[11pt]{article}
\usepackage{setspace}
\usepackage[nottoc,numbib]{tocbibind}
\usepackage{cite}

\usepackage{booktabs}
\usepackage{array}

\usepackage{amssymb}
\usepackage{amsthm}
\usepackage{amsfonts}
\usepackage{caption}
\usepackage{enumerate}
\usepackage{mathtools}
\usepackage[colorlinks=true,linkcolor=blue]{hyperref}
\usepackage{amsmath}
\usepackage{graphicx}
\usepackage{float}
\usepackage{multirow}
\usepackage{blindtext}
\usepackage{dsfont}
\usepackage[title]{appendix}
\usepackage[dvipsnames]{xcolor}
\usepackage{algorithm}
\usepackage[noend]{algpseudocode}
\makeatletter
\def\BState{\State\hskip-\ALG@thistlm}
\makeatother

\usepackage[margin=1in]{geometry}

\newcommand{\R}{{\mathbb{R}}}

\newcommand{\N}{{\mathbb{N}}}

\newcommand{\statespace}{\mathcal{X}}
\newcommand{\states}{\textbf{x}}

\newcommand{\pis}{\boldsymbol{\pi}}
\newcommand{\rstates}{\textbf{X}}

\newcommand\Perm{\textrm{Perm}(N)}

\newcommand{\Ps}{\textbf{P}}

\newcommand{\K}{{\textbf{K}}}
\newcommand{\Kcomm}{\K^{\textrm{comm}}_n}
\newcommand{\Kexpl}{\K^{\textrm{expl}}}
\newcommand{\KSEO}{\K^{\textrm{SEO}}}
\newcommand{\KDEO}{\K^{\textrm{DEO}}_n}
\newcommand{\Kodd}{\K^{\textrm{odd}}}
\newcommand{\Keven}{\K^{\textrm{even}}}
\newcommand{\Kswap}{\K^{(i,i+1)}}
\newcommand{\KPT}{\K^{\textrm{PT}}_n}

\newcommand{\Rn}{\mathcal{R}_n}
\newcommand{\Tn}{\mathcal{T}_n}

\newcommand{\partition}{\mathcal{P}_N}

\newcommand{\PSEO}{\mathbb{P}_{\textrm{SEO}}}
\newcommand{\PDEO}{\mathbb{P}_{\textrm{DEO}}}
\newcommand{\ESEO}{\mathbb{E}_{\textrm{SEO}}}
\newcommand{\EDEO}{\mathbb{E}_{\textrm{DEO}}}

\newcommand{\nexpl}{n_\textrm{expl}} 
 
\newcommand{\nscan}{n_\textrm{scan}} 
\newcommand{\ntune}{n_\textrm{tune}}
\newcommand{\nsample}{n_\textrm{sample}}  
\newcommand{\ntotal}{n_\textrm{total}}

\newcommand{\tauSEO}{\tau_{\textrm{SEO}}}
\newcommand{\tauDEO}{\tau_{\textrm{DEO}}}

\newcommand{\popt}{\mathcal{P}_N^*}

\newcommand{\Var}{\mathrm{Var}}
\newcommand{\Bern}{\mathrm{Bern}}
\renewcommand{\P}{\mathbb{P}}
\newcommand{\E}{{\mathbb{E}}}

\newcommand{\eps}{\varepsilon}
\newcommand{\ud}{\textrm{d}}

\newtheorem{theorem}{Theorem}
\newtheorem{corollary}{Corollary}
\newtheorem{lemma}{Lemma}
\newtheorem{proposition}{Proposition}

\begin{document}
\title{Non-Reversible Parallel Tempering: a Scalable Highly\newline Parallel MCMC Scheme}
\author{Saifuddin Syed\thanks{Department of Statistics, University of British Columbia, Canada.}, Alexandre Bouchard-C\^{o}t\'{e}$^{*}$, George Deligiannidis\thanks{Department of Statistics, University of Oxford, UK.}, Arnaud Doucet$^{\dagger}$}
\maketitle
\begin{abstract}
Parallel tempering (PT) methods are a popular class of Markov chain Monte Carlo schemes used to sample complex high-dimensional probability distributions. They rely on a collection of $N$ interacting auxiliary chains targeting tempered versions of the target distribution to improve the exploration of the state-space. We provide here a new perspective on these highly parallel algorithms and their tuning by identifying and formalizing a sharp divide in the behaviour and performance of reversible versus non-reversible PT schemes. We show theoretically and empirically that a class of non-reversible PT methods dominates its reversible counterparts and identify distinct scaling limits for the non-reversible and reversible schemes, the former being a piecewise-deterministic Markov process and the latter a diffusion. These results are exploited to identify the optimal annealing schedule for non-reversible PT and to develop an iterative scheme approximating this schedule. We provide a wide range of numerical examples supporting our theoretical and methodological contributions. The proposed methodology is applicable to sample from a distribution $\pi$ with a density $L$ with respect to a reference distribution $\pi_0$ and compute the normalizing constant $\int L\ \ud \pi_0$. A typical use case is when $\pi_0$ is a prior distribution, $L$ a likelihood function and $\pi$ the corresponding posterior distribution.
\end{abstract}
%\tableofcontents 

\section{Introduction}\label{sec_introduction}
Markov Chain Monte Carlo (MCMC) methods are widely used to approximate expectations with respect to a probability distribution with density $\pi(x)$ known up to a normalizing constant, i.e., $\pi(x) = \gamma(x) / \mathcal{Z}$ where $\gamma$ can be evaluated pointwise but the normalizing constant $\mathcal{Z}$ is unknown. When $\pi$ has multiple well-separated modes, highly varying curvature or when one is interested in sampling over combinatorial spaces, standard MCMC algorithms can perform very poorly. This work is motivated by the need for practical methods for these difficult sampling problems. A natural direction to address them is to use multiple cores and to distribute the computation.

\begin{figure}[bt]
	\centering 
	\includegraphics[width=0.9\linewidth]{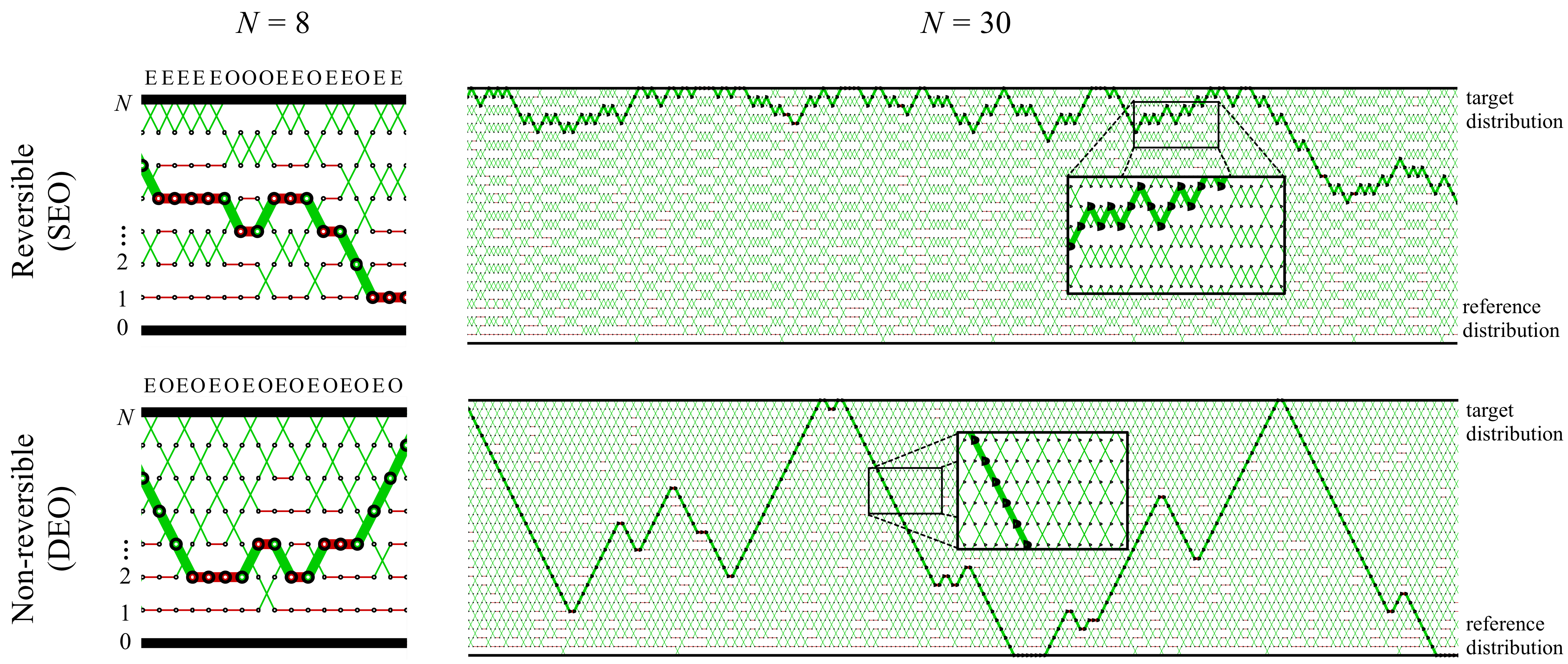}
 	\caption{Reversible (top) and non-reversible (bottom) PT for $N=8$ (left) and $N=30$ auxiliary chains (right) using equally spaced annealing parameters on a Bayesian change-point detection model \cite{davidson-pilon_bayesian_2015} where $\pi_0$ is the prior, $\pi$ the posterior. The sequence of swap moves forms $N+1$ index process trajectories (paths formed by the red and green edges). We show one such path in bold. The reversible and non-reversible PT clearly exhibit different scaling behaviour which we formalize in Section \ref{sec:scalinglim}.
}
	\label{fig:swaps}
\end{figure}

\subsection{Parallel Tempering}
One popular approach for multi-core and distributed exploration of complex distributions is Parallel Tempering (PT) which was introduced independently in statistics \cite{geyer1991markov} and physics \cite{hukushima1996exchange}; see also \cite{swendsen1986replica} for an earlier related proposal. Since its inception, PT remains to this day a very popular MCMC method to sample from complex multimodal target distributions arising in physics, chemistry, biology, statistics, and machine learning; see, e.g., \cite{issaoun2021persistent,ballnus2017comprehensive,chandra2019langevin,cho2010parallel,desjardins2014deep,diaz2020we,dorri_efficient_2020,kamberaj2020molecular,mueller2020adaptive}.

To sample from the target distribution $\pi$, PT introduces a sequence of auxiliary \emph{tempered} or \emph{annealed} probability distributions with densities $\pi^{(\beta_{i})}(x) \propto L(x)^{\beta_{i}} \pi_0(x)$ for $i=0,1,...,N$, where $\pi_0$ is an easy-to-sample reference distribution, $L(x)= \pi(x)/\pi_0(x)$ and the sequence $0=\beta_0<\beta_1<\cdots<\beta_N=1$ defines the \emph{annealing schedule}.
This bridge of auxiliary distributions is used to progressively transform samples from the \emph{reference distribution} ($\beta = 0$) into samples from the \emph{target distribution} ($\beta = 1$), for which only poorly mixing MCMC kernels may be available. For example, in the Bayesian setting where the target distribution is the posterior, we can choose the reference distribution as the prior, from which we can often obtain independent samples.

More precisely, PT algorithms are based on Markov chains in which the states are $(N+1)$-tuples, $\states = (x^0, x^1, x^2, \dots, x^N) \in \statespace^{N+1}$, and whose stationary distribution is given by $\pis(\states) = \prod_{i=0}^N \pi^{(\beta_i)}(x^i)$ \cite{geyer1991markov}.
At each iteration, PT proceeds by applying in parallel $N+1$ MCMC kernels targeting $\pi^{(\beta_i)}$ for $i=0,...,N$. We call these model-specific kernels the \emph{local exploration kernels}.  
The chains closer to the reference chain (i.e.\ those with annealing parameter $\beta_{i}$ close to zero) can typically traverse regions of low probability mass under $\pi$ while the chain at $\beta = 1$ ensures that asymptotically we obtain samples from the target distribution. Frequent communication between the chains at the two ends of the spectrum is therefore critical for good performance, and achieved by proposing to swap the states of chains at adjacent annealing parameters. \emph{Even swap} moves (rows labelled `E'), respectively \emph{Odd swap} moves (labelled `O'), propose to exchange states at chains with an even index $i$, respectively odd index $i$, and $i+1$. These proposals are accepted or rejected according to a Metropolis mechanism.

\subsection{Deterministic and stochastic even-odd schemes} \label{intro_deo}

A key notion used to analyze the behaviour of PT is the \emph{index process}. To provide intuition on this process, it is helpful to discuss briefly how PT is distributed over several machines. An important point is that instead of having pairs of machines exchanging high-dimensional states when a swap is accepted (which could be detrimental due to network latency), the machines should just exchange the annealing parameters. Suppose now we initialize machine $j$ with annealing parameter $\beta_j$. Then after $n$ scans, the annealing parameters are permuted among the $N+1$ machines according to the permutation $\textbf{I}_n=(I_n^0,\dots,I_n^N)$ of $\{0,\dots,N\}$, so that machine $j$ has annealing parameter $\beta_{I_n^j}$ at iteration $n$ of PT. Each \emph{index process} $I_n^j$, formally introduced in Section \ref{sec_trajectories}, is initialized using $I^j_0=j$ and tracks how the state of the corresponding chain evolves over the annealing schedule thanks to the swap moves; see Figure \ref{fig:swaps}. The index process $I_n^j$ thus monitors the information transfer between the reference and target on machine $j$ and as such determines partly the effectiveness of PT.

There have been many proposals made to improve this information transfer by adjusting the annealing schedule; see, e.g., \cite{kone_selection_2005,atchade_towards_2011,miasojedow2013adaptive}. These proposals are useful but do not address a crucial limitation of standard PT algorithms. In a distributed context, one can select randomly at each iteration whether to apply Even or Odd swap moves in parallel. The resulting stochastic even-odd swap (SEO) scheme, henceforth referred to as \emph{reversible PT} as it admits a reversible scaling limit (see Section \ref{sec:scalinglim}), yields index processes exhibiting a diffusive behaviour; see top row of Figure \ref{fig:swaps}. Hence we can expect that when $N$ is large it takes roughly $O(N^2)$ swap attempts for a state at $\beta_0=0$ to reach $\beta_N=1$ \cite{diaconis2000analysis}. The user thus faces a trade-off. If $N$ is too large, the acceptance probabilities of the swap moves are high but it takes a time of order $O(N^2)$ for a state at $\beta=0$ to reach $\beta=1$. If $N$ is too low, the acceptance probabilities of swap moves deteriorate resulting in poor mixing between the different chains. Informally, even in a multi-core or distributed setting, for $N$ large, the $O(N)$ gains in being able to harness more cores do not offset the $O(N^2)$ cost of the diffusion (see Section \ref{sec_domination} where we formalize this argument). As a consequence, the general consensus is that the temperatures should be chosen to allow for about a 20--40\% acceptance rate to maximize the squared jump distance travelled per swap in the space of annealing parameters $[0, 1]$ \cite{rathore2005optimal, kone_selection_2005, lingenheil2009efficiency, atchade_towards_2011}. Adding more chains past this threshold actually deteriorates the performance of reversible PT and there have even been attempts to adaptively reduce the number of additional chains \cite{lkacki2016state}. This is a lost opportunity, since PT is otherwise particularly suitable to implementation on multi-core or distributed architectures. 

An alternative to the SEO scheme is the deterministic even-odd swap (DEO) scheme introduced in \cite{okabe2001replica} where one deterministically alternates Even and Odd swap moves. We refer to DEO as \emph{non-reversible PT} as it admits a non-reversible scaling limit (see Section \ref{sec:scalinglim}). In particular, the resulting index processes do not appear to exhibit a diffusive, i.e.\ random walk type, behaviour, illustrated by the bottom row of Figure \ref{fig:swaps}. This non-diffusive behaviour of non-reversible PT explains its excellent empirical performance when compared to alternative reversible PT schemes observed in practice \cite{lingenheil2009efficiency} for any given schedule. However non-reversible PT, like reversible PT, is sensitive to the choice of the schedule. All the aforementioned tuning strategies developed for PT implicitly assume a reversible framework and do not apply to non-reversible PT (as empirically verified in \cite{lingenheil2009efficiency}). The main contribution of this paper is to identify some of the theoretical properties of non-reversible PT so as to establish optimal tuning guidelines for this algorithm and propose a novel schedule optimization scheme to implement them. In all our experiments, the resulting iteratively optimized  non-reversible PT scheme markedly outperforms the reversible and non-reversible PT schemes currently available (see Section \ref{sec_examples}).

\subsection{Overview of our contributions}\label{sec:overview-of-our-contributions}

After introducing formally the SEO and DEO schemes in Section \ref{sec_setup}, our first contribution is a non-asymptotic result showing that the non-reversible DEO scheme is guaranteed to outperform its reversible SEO counterpart. The notion of optimality we analyze is the \emph{round trip rate}, which quantifies how often information from the reference distribution percolates to the target; see Section \ref{sec_non_asymptotic_analysis}. The theoretical analysis is based on a simplifying assumption called Efficient Local Exploration (ELE), which is not expected to hold exactly in real scenarios. However we show empirically that there are practical methods to approximate ELE and that even when ELE is violated the key predictions made by the theory closely match empirical behaviour. In this sense ELE can be thought of as a useful \emph{model} for understanding PT algorithms.

In Section \ref{sec_communication_barrier} we introduce the local and global communication barrier $\lambda(\beta), \Lambda$ which encode the local and global efficiency of PT.  We then show that for non-reversible PT the round trip rate converges to $(2+2\Lambda)^{-1}$, in contrast to the reversible counterpart for which it decays to zero. To provide some intuition on how the small algorithmic difference between SEO and DEO can have such a profound impact, consider the scenario where PT would use the same distributions at the two end-points, $\pi = \pi_0$, as well as for all intermediate distributions so that $\Lambda=0$. Clearly this is not a realistic scenario, but in this simple context the index processes of DEO and SEO are easy to describe and contrast. In both DEO and SEO, $\pi = \pi_0$ implies that all proposed swaps will be accepted. For DEO, this makes the index process fully deterministic, performing direct trips from index $0$ to $N$ and back. Such a process could be compared to a ``conveyor belt'' with the property that no matter what is the value of $N$, one novel trip from chain $0$ reaches chain $N$ every two iterations, i.e.\ the round trip rate is $(2+2\Lambda)^{-1}=1/2$ as $\Lambda=0$. For SEO, even when $\pi = \pi_0$ the index process is still random and can be readily seen to be a simple discrete random walk. This implies that as $N$ increases, the round trip rate decreases to zero.

In practice, achieving high round trip rates requires careful tuning of the annealing schedule. %$\beta_0, \beta_1, \dots, \beta_N$.  
In Section \ref{sec_computation} we combine the analysis from Section \ref{sec_non_asymptotic_analysis} and Section \ref{sec_communication_barrier} to develop a novel methodology to optimize the annealing parameters. The optimal tuning guidelines we provide are different from existing (reversible) PT guidelines and the novel methodology is highly parallel as its performance does not collapse when a very large number of chains is used. However using a large number of chains does have a diminishing return, therefore we propose a mechanism to determine the optimal trade-off between the number of chains and the number of independent PT algorithms one should use.

In Section \ref{sec:scalinglim} we identify the scaling limit of the index processes for both reversible and non-reversible PT as the number of parallel chains goes to infinity. We show that this scaling limit is a piecewise-deterministic Markov process for non-reversible PT whereas it is a diffusion for reversible PT as suggested by the dynamics of the bold paths in Figure \ref{fig:swaps}. 

Finally in Section \ref{sec_examples}, we present a variety of experiments validating our theoretical analysis and novel methodology. The method is implemented in an open source Bayesian modelling language available at \url{https://github.com/UBC-Stat-ML/blangSDK}. Our software implementation allows the user to specify the model in BUGS-like language \cite{Lunn2000}. From this model declaration, a suitable sequence of annealed distributions is instantiated and a schedule optimized using our iterative method.

\section{Setup and notation}\label{sec_setup}

\subsection{Parallel tempering}\label{sec_Tempered}

Henceforth we will assume that the \emph{target} and \emph{reference} probability distributions $\pi$ and $\pi_0$ on $\statespace$ admit strictly positive densities with respect to a common dominating measure $\mathrm{d}x$. We will also denote these densities somewhat abusively by $\pi$ and $\pi_0$. 
It will be useful to define $V_0(x)=-\log\pi_0(x)$ and $V(x)=-\log L(x)$, where $L(x)=\pi(x)/\pi_0(x)$ is assumed finite for all $x\in\statespace$. Using this notation, the \emph{annealed distribution} at an annealing parameter $\beta$ is given by 
\begin{align}
\pi^{(\beta)}(x)= \frac{L(x)^\beta\pi_0(x)}{\mathcal{Z}(\beta)}=\frac{e^{-\beta V(x)-V_0(x)}}{\mathcal{Z}(\beta)},
\end{align}
where $\mathcal{Z}(\beta)= \int_\mathcal{X} L(x)^\beta\pi_0(x)\mathrm{d}x$ is the corresponding normalizing constant. We denote the \emph{annealing schedule} by $0=\beta_0<\beta_1<\cdots<\beta_N=1$. In our analysis we will view it as a partition $\partition=\{\beta_0,\dots,\beta_N\}$ of $[0,1]$ with mesh-size $\|\partition\|= \sup_i~\{\beta_i-\beta_{i-1}\}$. 

We define the PT state $\bar{\states}=(\states,\textbf{i})$ where $\states=(x^0,\dots,x^N)\in\statespace^{N+1}$ and $\textbf{i}=(i^0,\dots,i^N)\in\Perm$ the group of permutations on $\{0,\dots,N\}$. PT involves constructing a Markov Chain $\bar{\textbf{X}}_n=(\textbf{X}_n,\textbf{I}_n)$ over $\statespace^{N+1}\times\Perm$ invariant with respect to $\bar\pis(\bar\states) = (N+1)!^{-1}\prod_{i=0}^N \pi^{(\beta_i)}(x^i)$. The $i$-th component of $\textbf{X}_n=(X_n^0,\dots,X_n^N)$ tracks the states associated with annealing parameter $\beta_i$ at iteration $n$ while the permutation $\textbf{I}_n=(I_n^0,\dots,I_n^N)$ tracks how the annealing parameters are shuffled among machines. In particular at iteration $n$, machine $j$ stores the $I^j_n$-th component of $\textbf{X}_n$ and annealing parameter $\beta_{I_n^j}$. The sequence of states associated with $\beta_i$ is called the $i$-th \emph{chain} and the state on machine $j$ is the $j$-th \emph{replica}.

For both SEO and DEO, the overall $\bar\pis$-invariant Markov kernel $\KPT$ describing the algorithm is obtained by the composition of a $\bar\pis$-invariant local exploration kernel $\Kexpl$ and communication kernel $\Kcomm$,
\begin{equation}\label{def_PT_kernel}
    \KPT(\bar\states, A)=\Kcomm \Kexpl(\bar\states, A): = \int \Kcomm(\bar\states, \ud \bar\states') \Kexpl(\bar\states', A).
\end{equation} 
The difference between SEO and DEO is in the communication phase, namely $\Kcomm = \KSEO$ in the former case and $\Kcomm = \KDEO$ in the latter. Markov kernels corresponding to the reversible (SEO) and non-reversible (DEO) PT algorithms are described informally in the introduction and illustrated in Figure~\ref{fig:swaps}. 

 \subsection{Local exploration kernels}

The local exploration kernels are defined in the same way for SEO and DEO. They are also model specific, so we assume we are given one $\pi^{(\beta_i)}$-invariant kernel $K^{(\beta_i)}$ for each annealing parameter $\beta_i\in\partition$. These can be based on $\nexpl$ steps of Hamiltonian Monte Carlo, Metropolis--Hastings, Gibbs Sampling, Slice Sampling, etc. We construct the overall local exploration kernel by applying the annealing parameter specific kernels to each component independently from each other: 
\begin{equation}\label{def_exploration_kernel}
    \Kexpl((\states,\textbf{i}), A_0 \times A_1 \times \dots A_N\times\{\textbf{i}'\}) =\prod_{i=0}^N K^{(\beta_i)}(x^i, A_i)\delta_{\textbf{i}}(\textbf{i}'),
\end{equation}
where $\delta_{\textbf{i}}$ denotes the Dirac delta.

In our computational model, we implicitly assume that the local exploration kernel at $\beta = 0$ is special in that it can provide independent exact samples from $\pi_0$. Mathematically, $K^{(0)}(x, A_0) = \pi_0(A_0)$. This assumption is satisfied in most  Bayesian models equipped with proper prior distributions, but  also in other situations such as Markov random fields (see Appendix~\ref{sec:selection-of-pi0}).

\subsection{Communication kernels.} 

Before defining the communication scheme, we first construct its fundamental building block, a \emph{swap}. A swap is a Metropolis--Hastings move with a deterministic proposal which consists of swapping $\beta_i$ and $\beta_{i+1}$ across machines. From a current state $\bar\states=(\states, \textbf{i})$, a proposed swap state is denoted $\bar{\states}^{(i,i+1)}=(\states^{(i,i+1)},\textbf{i}^{(i,i+1)})$ where 
\begin{align}\label{swap_def}
\states^{(i,i+1)} &= (x^0,\dots,x^{i-1},x^{i+1},x^i,x^{i+2}, \dots,x^{N}),
\end{align}
and $\textbf{i}^{(i,i+1)}\in\Perm$ is the permutation obtained by swapping $i$ and $i+1$ in $\textbf{i}$.  The Metropolis--Hastings kernel $\textbf{K}^{(i,i+1)}$ corresponding to this update is given by
\begin{equation}
\textbf{K}^{(i,i+1)}(\bar{\states}, \cdot) = (1-\alpha^{(i,i+1)}(\bar{\states})) \delta_{\bar\states}(\cdot)+ \alpha^{(i,i+1)}(\bar\states) \delta_{\bar\states^{(i,i+1)}}(\cdot).
\end{equation}
The function  $\alpha^{(i,i+1)}(\bar\states)$ is the corresponding acceptance probability equal to 
\begin{align}
\alpha^{(i,i+1)}(\bar\states)
&=\min\left\{1,\frac{\bar\pis\left(\bar\states^{(i,i+1)}\right)}{\bar\pis(\bar\states)}\right\}=\exp \left( \min\{0,(\beta_{i+1}-\beta_i)(V(x^{i+1})-V(x^i))\}\right)\label{accept_ratio}
\end{align}
and the maximal collection of adjacent swap kernels that can be proposed in parallel without interference are
\begin{equation}
\Keven = \prod_{i \text{ even}} \Kswap,\qquad
\Kodd = \prod_{i \text{ odd}} \Kswap,
\end{equation}
which we call the \emph{even} and \emph{odd kernels} respectively.

For SEO, the kernel $\Kcomm=\KSEO$ is given by a mixture of the even and odd kernels in equal proportion while for DEO the kernel $\Kcomm=\KDEO$ is given by a deterministic alternation between even and odd kernels, that is
\begin{equation}
  \KSEO = \frac{1}{2} \Keven + \frac{1}{2} \Kodd,\quad\quad  \KDEO
=\begin{cases}
\Keven & \text{if $n$ is even,}\\
\Kodd  & \text{if $n$ is odd.}
\end{cases}
\end{equation} 

We provide pseudo-code for the DEO scheme in Algorithm \ref{alg_deo}. The pseudo-code also estimates the average rejection probabilities $r^{(i,i+1)}$ of swap moves between chains $i$ and $i+1$ which are used to optimize the annealing schedule in Section \ref{sec_computation}. When the schedule is fixed, lines 1, 12, 17 can be omitted, and one should use ``for $i\in P$" on line 10. For simplicity, the swap in line 15 is shown for a \emph{parallel computing} context, where several cores have a shared memory, and hence line 15 is simply an exchange of pointers in an array, which is efficient thanks to memory sharing. For a \emph{distributed computing} implementation, where several machines do not share memory and instead need to communicate over the network, it becomes advantageous to swap annealing parameters instead of states. Refer to Appendix~\ref{app:distributed} for pseudo-code for the distributed computing implementation. 

\begin{algorithm}[ht]
	\caption{DEO(number of scans $\nscan$, annealing schedule $\partition$)}\label{alg_deo} 
	\begin{algorithmic}[1]
		\State $ r^{(i,i+1)} \gets 0$ for all $i \in \{0, 1, \dots, N-1\}$ \Comment{Swap rejection statistics used in Section \ref{sec_adaptive_algo} to optimize the schedule} 
		\State $\states\gets \states_0$ \Comment{Initialize chain} 
		\For{$n$ {\bf in} 1, 2, \dots, $\nscan$}
		\If{$n$ is even} \Comment{Non-reversibility inducing alternation}
		\State $P \gets \{i: 0 \le i < N, i\text{ is even} \}$ \Comment{Even subset of $\{0,\dots,N-1\}$}
		\Else
		\State $P \gets \{i: 0 \le i < N, i\text{ is odd} \}$ \Comment{Odd subset of $\{0,\dots,N-1\}$}
		\EndIf
		\For{$i$ \textbf{in} $0,\dots,N$}  \Comment{Local exploration phase (parallelizable)}
		\State $x^i \sim K^{(\beta_i)}(x^i_{n-1},\cdot)$
		\EndFor
		\For{$i$ \textbf{in} $0,\dots,N-1$} \Comment{Communication phase (parallelizable)}
		\State $\alpha \gets \alpha^{(i,i+1)}$  \Comment{Equation (\ref{accept_ratio}).}
		\State $r^{(i,i+1)} \gets r^{(i,i+1)} + (1-\alpha)$
		\State $A \sim \Bern(\alpha)$
		\If{$i\in P$ \textbf{and} $A = 1$}
		\State $(x^i, x^{i+1}) \gets (x^{i+1}, x^i)$ \Comment{ See also Appendix \ref{app:distributed} for a distributed computing implementation, where annealing parameters are swapped instead of states.}
		\EndIf
		\EndFor
		\State $\states_n\gets \states$ 
		\EndFor
		\State $r^{(i,i+1)} \gets  r^{(i,i+1)}/\nscan$ for all $i \in \{0, 1, \dots, N-1\}$\Comment{Equation \eqref{def_rhat}}
		\State 
		\Return $(\states_1,\dots,\states_{\nscan}),(r^{(0,1)},\dots,r^{(N-1,N)})$
	\end{algorithmic}
\end{algorithm}

\subsection{The index process}\label{sec_trajectories}

As discussed in Section \ref{intro_deo}, the $j$-th component of permutation $\textbf{I}_n=(I^0_n,\dots,I^N_n)$ encodes the flow of information between the reference and target on machine $j$. We define the \emph{index process} for machine $j$ as $(I_n^j,\eps_n^j)\in\{0,\dots,N\}\times\{-1,1\}$, where $\eps_n^j=1$ if the annealing parameter on machine $j$ is proposed an increase after scan $n$, and $-1$ otherwise. The concept is best understood visually: refer to the bold piecewise linear paths illustrating $I^{N-2}_{n}$ for $N=8$ and $I^{N-1}_{n}$ for $N=30$ in Figure~\ref{fig:swaps}.

Let $P^{(i,i+1)} \in\{0,1\}$ denote an indicator that a swap is proposed between chains $i$ and $i+1$ at iteration $n$. The realized swaps are then defined from the proposal indicators as $S^{(i,i+1)} = P^{(i,i+1)}A^{(i,i+1)}$, where $A^{(i,i+1)}|\bar{\textbf{X}} \sim \mathrm{Bern}(\alpha^{(i,i+1)}(\bar{\textbf{X}} ))$ are acceptance indicator variables (Figure \ref{fig:notation} in Appendix~\ref{app_algos}). The index process satisfies the following recursive relation: initialize $I^j_0=j$ and $\eps_0^j=1$ if $P_0^{(j,j+1)}=1$ and $-1$ otherwise. For $n > 0$, we have
\begin{align}
I_{n+1}^j &= 
\begin{cases}
I_n^j+\eps_n^j &\text{if } S_n^{(I_n^j,I_n^j+\eps_n^j)}=1, \\
I_n^j &\text{otherwise},
\end{cases},
\qquad
  \eps_{n+1}^j &= 
\begin{cases}
1 &\text{if } P_n^{(I_{n+1}^j,I_{n+1}^j+1)}=1, \\
-1 &\text{otherwise.}
\end{cases} 
\end{align}

\begin{figure}
	\centering
	\includegraphics[width=0.9\linewidth]{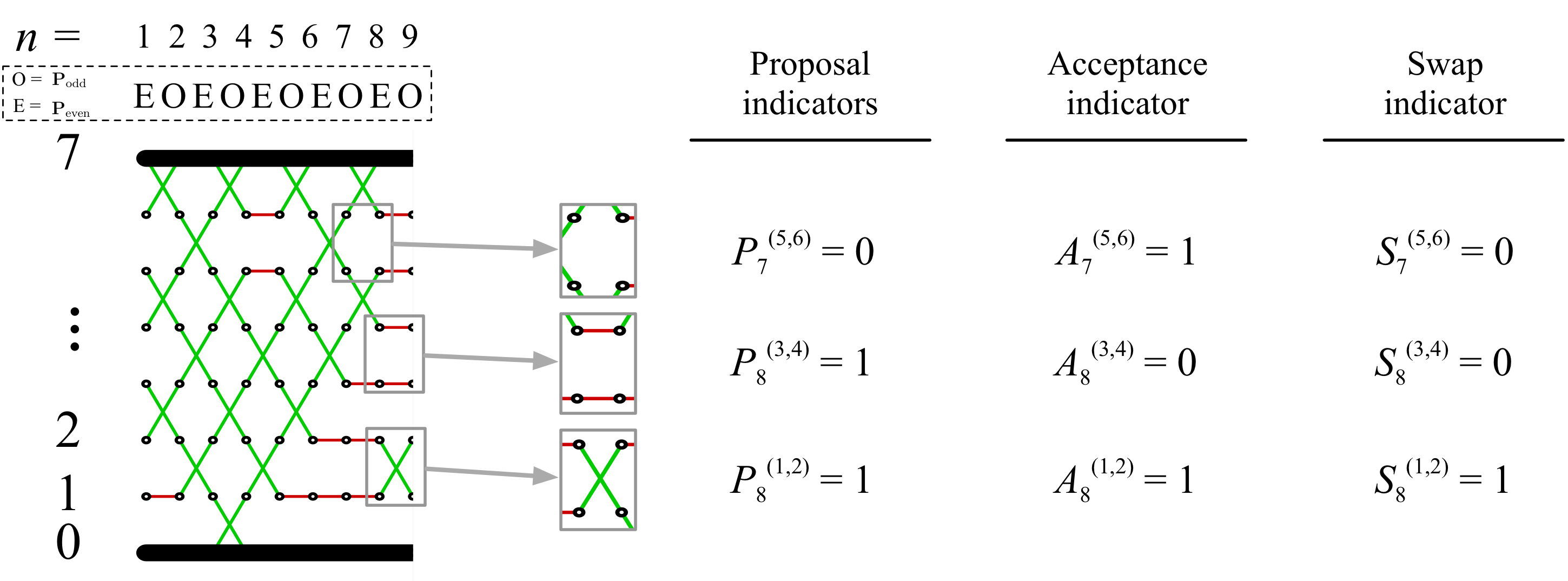}
	\caption{Illustration of the proposal, acceptance and swap indicators.}
	\label{fig:notation}
\end{figure}

We will use the dynamics of the index process to explain the differences between SEO and DEO communication. The only difference between the two is in the proposal indicators. Define $\Ps_n = (P_n^{(0,1)}, P_n^{(1,2)}, \dots, P_n^{(N-1,N)})$, $\Ps_n$ is deterministic for DEO, i.e.\ $\Ps_n = \Ps_{\text{even}} = (1,0,1,\dots)$ for even $n$ and $\Ps_n = \Ps_{\text{odd}} = (0,1,0,\dots)$ for odd $n$. In SEO, we have $\Ps_n \sim \text{Unif}\{\Ps_{\text{even}}, \Ps_{\text{odd}}\}$. 

For SEO, the variables $\eps_n\sim\mathrm{Unif}\{-1,1\}$ are i.i.d, and consequently the index process exhibits a random walk behaviour. In contrast for DEO, we have $\eps^{j}_{n+1}=\eps^{j}_n$ so long as $I^{j}_{n+1}=I^{j}_n+\eps^{j}_n$ and $\eps^{j}_{n+1}=-\eps^{j}_n$ otherwise. Therefore the index process for DEO performs a more systematic exploration of the space as the direction $\eps^{j}_{n+1}$ is only reversed when a swap involving machine $j$ is rejected or if the boundary is reached. The qualitative differences between the two regimes can be seen in Figure \ref{fig:swaps} (see also Figure \ref{fig_trajectory} in the Supplementary Material). In particular the index process for DEO in these figures behaves very differently as $N$ increases, this will be explored formally in Section \ref{sec:scalinglim}.

As mentioned in the introduction, we refer to the PT algorithm with SEO and DEO communication as \emph{reversible}  PT and \emph{non-reversible} PT respectively. Our terminology is somewhat abusive but is justified by the analysis in Section \ref{sec:lifted_property} and Section \ref{sec:scalinglim}, where it is shown that, under certain assumptions, the index process is reversible for SEO while it is non-reversible for DEO.

\section{Non-asymptotic analysis of PT algorithms}\label{sec_non_asymptotic_analysis}

\subsection{Model of compute time}
We start with a definition of what we model as one unit of compute time: throughout the paper, we assume a massively parallel or distributed computational setup and sampling once from each of local exploration kernel $K^{(\beta)}$ has cost $O(\nexpl)$ which dominates the cost of the swap kernel $\K^{(i,i+1)}$. Consequently a scan of PT has cost $O(\nexpl)$ and for a fixed computational budget $\ntotal$, the total number of scans is $\nscan=O(\ntotal/\nexpl)$.

The assumption that the per-iteration cost of PT is independent of the number of chains is reasonable in GPU and parallel computing scenarios, since the communication cost for each swap does not increase with the dimension of the problem (by swapping annealing parameters instead of states). We also assume that the number of PT scans will  dominate the number of parallel cores available, i.e.\ $\nscan \gg N$. This is reasonable when addressing challenging sampling problems. Although there are numerous empirical studies on multi-core and distributed implementation of PT \cite{altekar_parallel_2004,mingas_parallel_2012,fang_parallel_2014}, we are not aware of previous theoretical work investigating such a computational model.

\subsection{Performance metrics for PT methods}

The standard notion of computational efficiency of MCMC schemes is the effective sample size (ESS) per compute time. However, for PT methods, since the ESS per compute time depends on the details of the problem specific local exploration kernels $\Kexpl$, alternatives have been developed in the literature to assess the performance of $\Kcomm$ which are independent of $\Kexpl$ \cite{katzgraber2006feedback,lingenheil2009efficiency}.

We are motivated by the Bayesian context where it is typically possible to obtain one independent sample from the reference distribution $\pi_0$ (i.e.\ from the prior distribution) at each iteration. We say that an \emph{annealed restart} has occurred on machine $j$ when $I_n^j$ goes from $0$ to $N$ (i.e.\ $\beta$ goes from $0$ to $1$), which corresponds to a sample generated from $\pi_0$ propagating to the target $\pi$. Informally an annealed restart can be thought of as a sampling equivalent to what is known in optimization as a random restart. We say a \emph{round trip} has occurred on machine $j$ when $I_n^j$ goes from $0$ to $N$ and then goes back to $0$ (i.e. $\beta$ goes from $0$ to $1$ to $0$).

Formally, we recursively define $T_{\downarrow,0}^j=\inf\{n: (I_n^j,\eps_n^j)=(0,-1)\}$ and for $k\geq 1$,
\begin{align}
T_{\uparrow,k}^j &= \inf\{n > T_{\downarrow,k-1}^j : (I_n^j,\eps_n^j) = (N,1)\}, \\
T_{\downarrow,k}^j &= \inf\{n > T_{\uparrow,k}^j: (I_n^j,\eps_n^j) = (0,-1)\}.
\end{align}
The $k$-th annealed restart and round trip for machine $j$ occurs at scan $T_{\uparrow,k}^j$ and $T_{\downarrow,k}^j$ respectively. Let $\Tn$ and $\Rn$ be the total number of annealed restarts and round trips respectively during the first $n$ scans of the DEO algorithm.

We wish to optimize for the percentage of iterations that result in an annealed restart, i.e.\ $\tau = \lim_{n\to\infty} \E[\Tn] / n$, where we use abusively the same random variables for SEO and DEO but differentiate these schemes by using the probability measures $\PSEO$ and $\PDEO$ with associated expectation operators $\ESEO$ and $\EDEO$. We use $\P$ and $\E$ for statements that hold for both algorithms. If $\Tn^j$ and $\Rn^j$ are the total number of annealed restarts and round trips during the first $n$ iterations on machine $j$ respectively, then we have $\Rn^j\leq \Tn^j\leq \Rn^j+1$. Consequently, $\Rn\leq \Tn\leq \Rn+N+1$ and thus $\tau = \lim_{n\to\infty} \E[\Rn] / n$. 

 In the PT literature, $\tau$ is commonly referred to as the \emph{round trip rate} and has been used to compare the effectiveness of various PT algorithms \cite{katzgraber2006feedback,lingenheil2009efficiency}. Empirically we observed that round trips per unit cost strongly correlate with ESS per unit cost as seen in Figure \ref{fig_ESSvsRoundTrip}, making the round trip rate a natural objective function to compare and tune parallel tempering algorithms. 
 
 \begin{figure}
\begin{center}
\includegraphics[width=.42\linewidth]{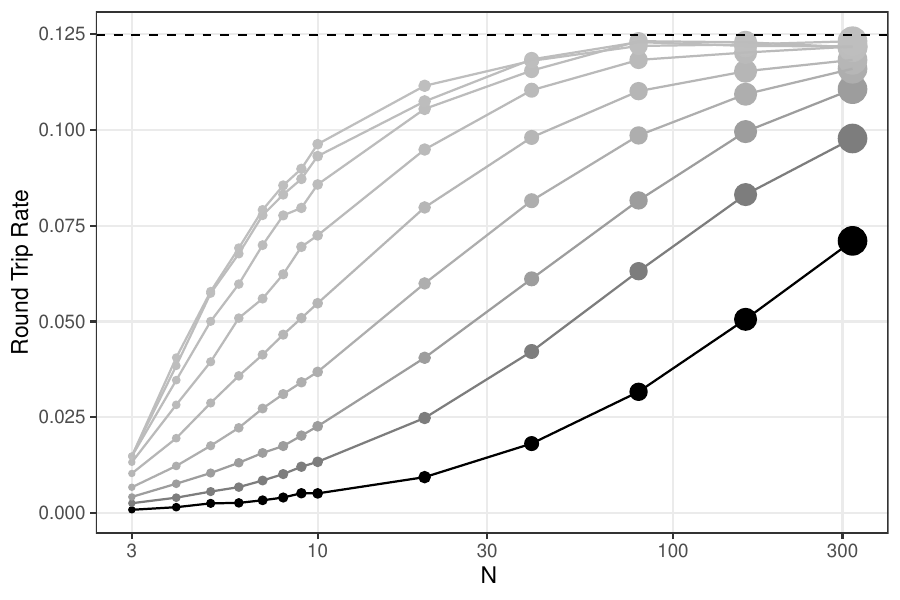}
\includegraphics[width = 0.56\linewidth]{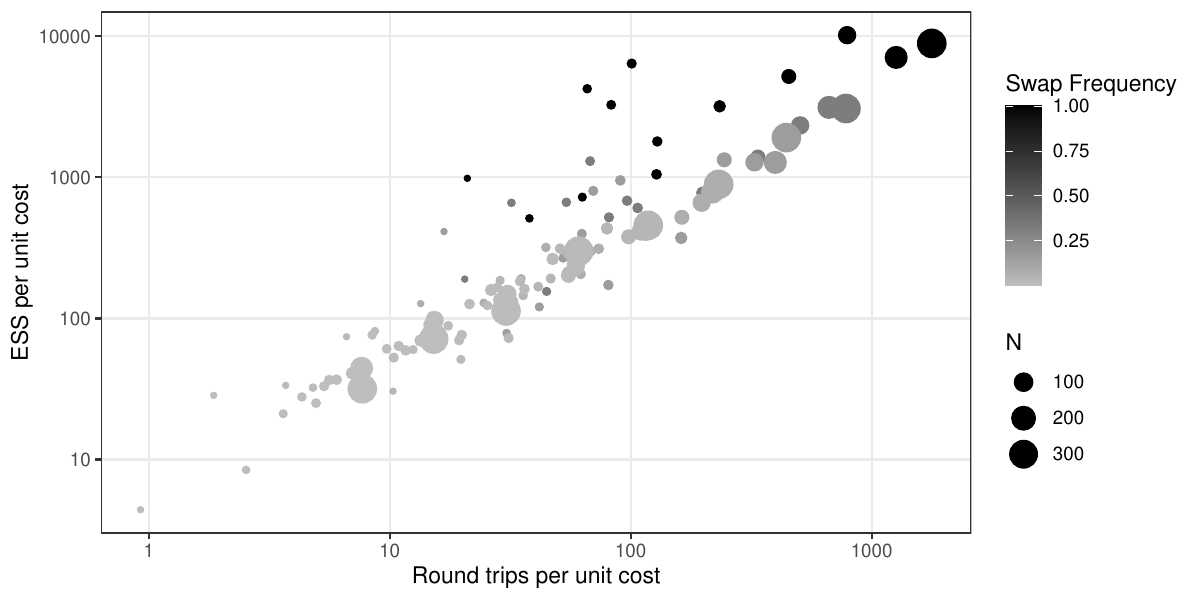}
\end{center}
\caption{The round trip rate and ESS for a $5\times 5$ Ising Model with a magnetic moment $0.1$ and $\nexpl$ 1-bit flips between scans ranging from $1$ to $400$, and $N$ ranging from $3$ to $320$. The schedule is tuned using Algorithm \ref{alg_adaptive} with $\ntune=25000$ scans for tuning, $\nsample=25000$ scans for sampling.  (Left) Round trip rate versus $N$ for different swap attempt frequency ($1/\nexpl)$. The dotted line is the optimal round trip rate $\bar{\tau}$ predicted by Theorem \ref{thm_efficiency_convergence}. (Right) The ESS per unit cost versus round trip rate per unit cost for each run, with a correlation coefficient of 0.81.}
\label{fig_ESSvsRoundTrip}
\end{figure}

Another performance metric commonly used in the PT literature is the \emph{expected square jump distance} (ESJD) \cite{atchade_towards_2011, kone_selection_2005}. 
While this criterion is useful within the context of reversible PT for selecting the optimal number of parallel chains, the ESJD is too coarse to compare reversible to non-reversible PT methods as, for any given annealing schedule, the ESJD is identical in both cases.

\subsection{Model assumptions}\label{sec_assumptions}
The analysis of the round trip times is in general intractable because the index process is not Markovian. Indeed, simulating a transition depends on the swap indicators $S_n^{(i,i+1)}$ (see Section \ref{sec_trajectories}), the distributions of which themselves depend on the state configuration $\bar\rstates$. To simplify the analysis, we will make in the remainder of the paper the following simplifying assumptions:
\begin{enumerate}
    \item [(A1)] \emph{Stationarity}: $\bar\rstates_0 \sim \bar\pis$ and thus $\bar\rstates_n \sim \bar\pis$ for all $n$ as the kernel $\KPT$ is $\bar\pis$-invariant.
    \item [(A2)] \emph{Efficient Local Exploration (ELE)}: For $X\sim \pi^{(\beta)}$ and $X' | X \sim K^{(\beta)}(X, \cdot)$, the random variables $V(X)$ and $V(X')$ are independent.
    \item [(A3)] \emph{Integrability}: $V^3$ is integrable with respect to $\pi_0$ and $\pi$. 
\end{enumerate}

It follows from Assumptions (A1)--(A2) and \eqref{accept_ratio} that the behaviour of the communication scheme only depends on the distribution of the state $\bar\rstates_n$ via the $N+1$ univariate distributions of the chain-specific energies $V^{(i)} = V\left(X^{(i)}\right)$, $i \in \{0, 1, 2, \dots, N\}$.  This allows us to build a theoretical analysis which makes no structural assumption on the state space $\statespace$ or the target $\pi$ as typically done in the literature: for example,  \cite{atchade_towards_2011} assume a product space $\statespace = \statespace_0^d$ for large $d$, and \cite{predescu2004incomplete} assume $\pi^{(\beta)}$ satisfies a constant heat capacity.

Admittedly the ELE assumption (A2) does not hold in practical applications. ELE can be approximated by increasing the number of local exploration kernels applied between consecutive swap ($\nexpl$). However one may worry that to achieve a good approximation in challenging problems, $\nexpl$ would have to be set to a value so large as to defy the practicality of our analysis. Surprisingly, we have observed empirically that this was not the case in the multimodal problems we considered. Figure~\ref{fig:mixture-V_vs_X} displays results in four models where a local exploration kernel alone induces good mixing of the energy chain $V(X_n)$ (hence ELE can be approximated) yet the local exploration kernel alone is insufficient to achieve good mixing on the full state space, $\bar\rstates_n$ (so that PT is justified and indeed yields efficient exploration of the configuration space). This gap is possible since $V(X)$ is 1-dimensional and potentially unimodal even when $X$ is not. This is the motivation for ELE since assuming the independence of $V(X)$ and $V(X')$ is weaker than assuming the independence of $X$ and $X'$ (as hypothesized e.g.\ in Section 5.1 of \cite{atchade_towards_2011}). Obviously ELE is still expected to be a somewhat crude simplifying assumption in very complex problems; e.g.\ for the highly challenging high-dimensional copy number inference problem illustrated in Figure~\ref{fig:mixture-V_vs_X}.

In Section~\ref{sec_ELE_violation}, we describe additional empirical results supporting that ELE is a useful model for the purpose of designing and analyzing PT algorithms. In the severe ELE violation regime, we show that the key quantities used in our analysis are either well approximated, or approached as $N$ increases.

\begin{figure}
	\begin{center}
		\begin{tabular}{cc} Bayesian mixture model & ODE parameters \\
			 \includegraphics[width=0.4\linewidth]{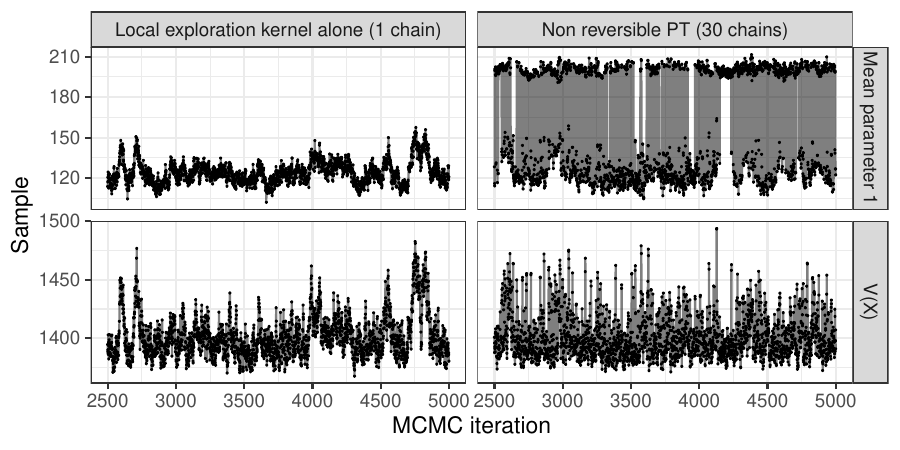} &  \includegraphics[width=0.4\linewidth]{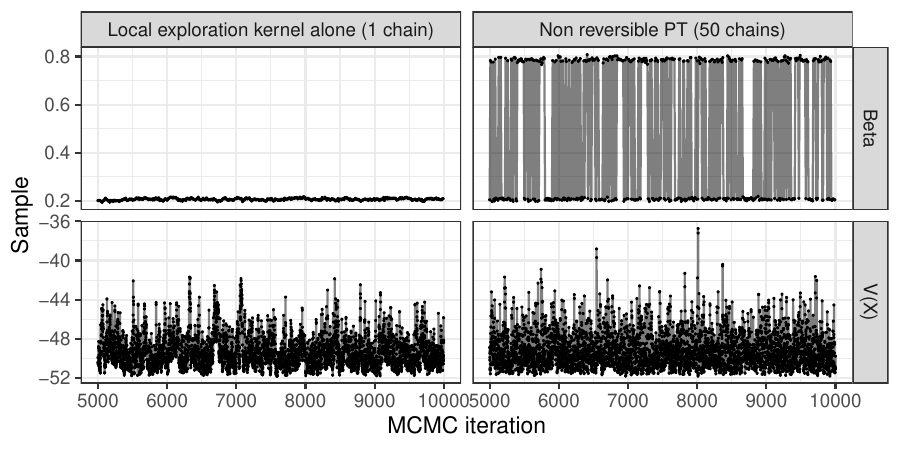}\\
			 Ising model & Copy number inference \\
			\includegraphics[width=0.4\linewidth]{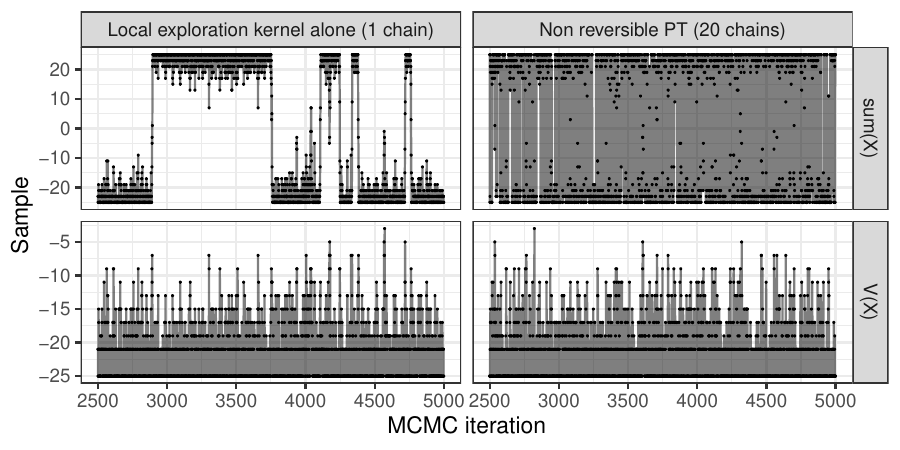}   &  \includegraphics[width=0.4\linewidth]{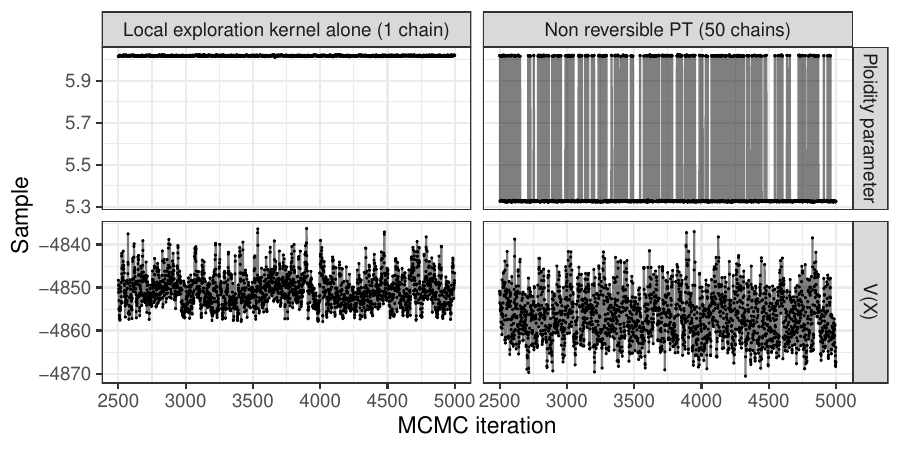}  
		\end{tabular} 
	\end{center}
	\caption{Four multimodal examples (described in Section~\ref{sec:empirical} and Appendix~\ref{app_experiments}) where a local exploration kernel provides a reasonable approximation of the ELE assumption. For each inference problem, we show  trace plots for MCMC based on single chain (i.e.\ the local exploration kernel alone; left facet), and for a non-reversible PT algorithm based on the same local exploration kernel (right facet). The top facets each show a component of $X$, and the bottom facets, $V(X)$.}
	\label{fig:mixture-V_vs_X}
\end{figure}

Assumptions (A1)-(A2) allow us to express the swap indicators as \emph{independent} Bernoulli random variables $S_n^{(i,i+1)} \sim \text{Bern}(s^{(i,i+1)})$ where $s^{(i,i+1)}$ is given by the expectation of Equation \eqref{accept_ratio},
\begin{equation}\label{eq_exp_accept} 
s^{(i,i+1)} = \E\left[\alpha^{(i,i+1)}(\bar\rstates)\right] = \E\left[\exp \left( \min\left\{0,\left(\beta_{i+1}-\beta_i\right)\left(V^{(i+1)}-V^{(i)}\right)\right\}\right)\right],
\end{equation}
the expectation being over two independent random variables $V^{(i)},V^{(i+1)}$, satisfying $V^{(i)} \stackrel{d}{=} V(X^{(\beta_i)})$ for $X^{(\beta_{i})}\sim\pi^{(\beta_i)}$.

\subsection{Reversibility and non-reversibility of the index process}\label{sec:lifted_property}
Under assumptions (A1)-(A2), each index process $(I^{j}_n,\eps^{j}_n)$ is Markovian for $j=0,\dots,N$ with transition kernel $K^{\mathrm{SEO}}$ and $K^{\mathrm{DEO}}$ for reversible and non-reversible PT respectively. See Appendix \ref{app_markov_index_process} for an explicit representation of these kernels. We drop the superscript $j$ when the particular machine is not relevant.

The kernel $K^{\mathrm{SEO}}$ defines a reversible Markov chain on $\{0,\dots,N\}\times\{-1,1\}$ with uniform stationary distribution while $K^{\mathrm{DEO}}$  
satisfies the skew-detailed balance condition with respect to the same distribution,
\begin{equation}\label{DEO_skew_detail_balance}
    K^{\mathrm{DEO}}((i,\eps),(i',\eps'))=K^{\mathrm{DEO}}((i',-\eps'),(i,-\eps)),
\end{equation}
and is thus non-reversible. It falls within the generalized Metropolis--Hastings framework, see, e.g., \cite{stoltz2010free}.

Reversibility necessitates that the Markov chain must be allowed to backtrack its movements. This leads to inefficient exploration of the state space. As a consequence,  non-reversibility is typically a favourable property for MCMC chains. A common recipe to design non-reversible sampling algorithms consists of expanding the state space to include a ``lifting'' parameter that allows for a more systematic exploration of the state space \cite{chen1999lifting,diaconis2000analysis,turitsyn2011irreversible,vucelja2016lifting}. 

The index process $(I_n,\eps_n)$ for non-reversible PT can be interpreted as a ``lifted'' version of the index process for reversible PT with lifting parameter $\eps_n$. Under DEO communication, $I_n$ travels in the direction $\eps_n$ and only reverses direction when $I_n$ reaches a boundary or when a swap rejection occurs. This ``lifting'' construction helps explain the qualitatively different behaviour between reversible and non-reversible PT and will be further explored when identifying the scaling limit of $(I_n,\eps_n)$ in Section \ref{sec:scalinglim}. The lifted PT of \cite{wu_irreversible_2017} exploits a similar construction but only one of the $N+1$ index processes is lifted instead of all of them for DEO. A lifted version of simulated tempering has also been proposed by \cite{sakai_irreversible_2016}.

\subsection{Non-asymptotic domination of non-reversible PT} \label{sec_domination}
Assumptions (A1)-(A2) ensure that for each $j=0,\dots,N$, $\mathcal{R}_n^j$ is a delayed renewal process with round trip times $T_k^j=T^j_{\downarrow,k}-T^j_{\downarrow,k-1}$ for $k\geq 1$ and $j=0,\dots,N$. In particular, $T_k^j$ are independent and identically distributed and, for convenience, we introduce the random variable $T\stackrel{d}{=}T_k^j$. By the key renewal theorem, we have
\begin{align}\label{round_trip_key_renewal}
    	\tau = \sum_{j=0}^N\lim_{n\to\infty}\frac{\E[\mathcal{R}^j_n]}{n}=\frac{N+1}{\E[T]}.
\end{align}
An analytical expression for $\E[T]$ for reversible PT was first derived by \cite{nadler2007optimizing}. 
We extend this result to non-reversible PT in Theorem \ref{prop_round_trip}.
\begin{theorem}\label{prop_round_trip}
For any annealing schedule $\partition=\{\beta_0,\dots,\beta_N\}$, 
    \begin{align}
        \E_{\mathrm{SEO}}[T]&=2(N+1)N+2(N+1)E(\partition),\label{round_trip_SEO}\\
        \E_{\mathrm{DEO}}[T]&=2(N+1)+2(N+1)E(\partition),\label{round_trip_DEO}
    \end{align}
where $E(\partition)=\sum_{i=1}^Nr^{(i-1,i)}/s^{(i-1,i)}$, and $r^{(i-1,i)}=1-s^{(i-1,i)}$ is the probability of rejecting a swap between chains $i$ and $i+1$.
\end{theorem}
The proof can be found in Appendix \ref{app_round_trip}.
 
Intuitively, Theorem \ref{prop_round_trip} implies $\E[T]$ can be decomposed as the independent influence of communication scheme $\Kcomm$ and schedule $\partition$ respectively. When all proposed swaps are accepted (i.e.\ $\pi=\pi_0$), the index process for reversible PT reduces to a simple random walk on $\{0,\dots,N\}$, whereas for non-reversible PT, the index processes takes a direct path from $0$ to $N$ and back. Therefore, the first term in \eqref{round_trip_SEO} and \eqref{round_trip_DEO} represents the expected time for a round trip to occur in this idealized, rejection-free setting. The second term of \eqref{round_trip_SEO} and \eqref{round_trip_DEO} are identical and represent the additional time required to account for rejected swaps under schedule $\partition$. Motivated by Theorem \ref{prop_round_trip}, we will refer to $E(\partition)$ as the \emph{schedule inefficiency}.

By applying Theorem \ref{prop_round_trip} to Equation \eqref{round_trip_key_renewal}, we get a non-asymptotic formula for the round trip rate in terms of $E(\partition)$.

\begin{corollary}\label{cor_rtr}
For any annealing schedule $\partition$ we have 
\begin{align}
    \tau_{\mathrm{SEO}}(\partition)&=\frac{N+1}{\E_{\mathrm{SEO}}[T]} =\frac{1}{2N+2E(\partition)},\\
    \tau_{\mathrm{DEO}}(\partition)&= \frac{N+1}{\E_{\mathrm{DEO}}[T]} =\frac{1}{2+2E(\partition)}. 
\end{align}
Consequently, $\tau_{\mathrm{DEO}}(\partition)> \tau_{\mathrm{SEO}}(\partition)$ for $N>1$.
\end{corollary}

Corollary \ref{cor_rtr} implies that non-reversible PT dominates reversible PT for any $N>1$ and any annealing schedule $\partition$. 

\section{Asymptotic analysis of PT algorithms}\label{sec_communication_barrier}

\subsection{The communication barrier}\label{section_swap_acceptance}
We begin by analyzing the behaviour of the PT swaps as $\|\partition\|=\max_i |\beta_i-\beta_{i-1}|$ goes to zero. We define the swap and rejection functions $s,r:[0,1]^2\to[0,1]$ respectively as, 
\begin{align}  
s(\beta,\beta')&= \E\left[ \exp \left(\min\{0,(\beta'-\beta)(V^{(\beta')}-V^{(\beta)})\} \right)\right],\label{def_s_function}\\
r(\beta,\beta')&= 1-s(\beta,\beta'),\label{def_r_function}\
\end{align}
where $V^{(\beta)}\stackrel{d}{=} V(X^{(\beta)})$ for $X^{(\beta)}\sim \pi^{(\beta)}$ and $V^{(\beta)},V^{(\beta')}$ are independent. The quantities $s(\beta,\beta')$ and $r(\beta,\beta')$ are symmetric in their arguments and represent the probability of swap and rejection occurring between $\beta$ and $\beta'$ respectively under the ELE assumption (A2). Note that $s^{(i-1,i)}=s(\beta_{i-1},\beta_i)$.

To take the limit as $\|\partition\|\to 0$, it will be useful to understand the behaviour of $r(\beta,\beta')$ when $\beta\approx\beta'$. The key quantity that drives this asymptotic regime is given by a function $\lambda:[0,1]\to [0,\infty)$ defined as the instantaneous rate of rejection of a proposed swap at annealing parameter $\beta$,
\begin{align}\label{lambda_limit_def}
\lambda(\beta)=\lim_{\delta \to 0} \frac{r(\beta,\beta+\delta)}{|\delta|}.
\end{align}
See Figure~\ref{fig:all-models} (center) for examples of estimated $\lambda$ from various models. We define the integral of $\lambda$ by $\Lambda(\beta)=\int_0^\beta\lambda(\beta')\mathrm{d}\beta'$ and denote $\Lambda=\Lambda(1)$. 
Extending Proposition 1 in \cite{predescu2004incomplete} provides the following result. 
\begin{theorem}\label{theorem_rate_est}
The instantaneous rejection rate $\lambda$ is twice continuously differentiable and is equal to
\begin{equation}\label{def_lambda}
\lambda(\beta)=\frac{1}{2}\E\left[|V^{(\beta)}_1-V_2^{(\beta)}|\right],
\end{equation}
where $V_1^{(\beta)},V_2^{(\beta)}$ are independent random variables with law identical to the law of $V^{(\beta)}$. Moreover, we have
\begin{align}\label{estimate_rate_integral_r}
r(\beta,\beta')&=|\Lambda(\beta')-\Lambda(\beta)|+O(|\beta'-\beta|^3).
\end{align}
\end{theorem}
See Appendix \ref{app_regularity} for a proof and general smoothness properties of $\lambda$. In particular, the existence and smoothness of $\lambda$ is guaranteed by Assumption (A3). 

Theorem \ref{theorem_rate_est} shows that $\lambda$ encodes up to second order the behaviour of $r$ as the annealing parameter difference between the chains goes to 0. When $\lambda(\beta)$ is high, swaps are much more likely to be rejected, implying $\lambda(\beta)$ measures the difficulty of local communication for a chain with annealing parameter $\beta$. 

Notice that $\Lambda\geq 0$ with equality if and only if $\lambda(\beta)=0$ for all $\beta\in[0,1]$. It can be easily verified from \eqref{def_lambda} that $\lambda=0$ if and only if $V^{(\beta)}$ is constant $\pi^{(\beta)}$-a.s.\ for all $\beta\in[0,1]$ which happens precisely when $\pi_0=\pi$. So $ \Lambda$ defines a natural symmetric divergence measuring the difficulty of communication between $\pi_0$ and $\pi$. In particular, for any schedule $\partition$, the sum of the rejection rates is approximately constant and equal to $\Lambda$ as formalized in Corollary \ref{cor_invariant}.

\begin{corollary}\label{cor_invariant}
For any schedule $\partition$, we have
\begin{align}
    \sum_{i=1}^Nr(\beta_{i-1},\beta_i)=\Lambda+O(N\|\partition\|^3).
\end{align}
\end{corollary}

Motivated by Theorem \ref{theorem_rate_est} and Corollary \ref{cor_invariant} we will henceforth refer to $\lambda$, $\Lambda$ and $\Lambda(\cdot)$ as the \emph{local}, \emph{global} and \emph{cumulative communication barriers} respectively.

\subsection{Asymptotic analysis of round trip rate}\label{sec_asymptotic_rtr}

The communication barrier allows us to study the asymptotic performance of PT when the number of parallel chains is large. When $N$ increases, the round trip rate for reversible PT decays to $0$ while it converges to a positive constant $\bar\tau$ for non-reversible PT as seen in left plot in Figure \ref{fig_ESSvsRoundTrip}. We formally show this in Theorem \ref{thm_efficiency_convergence}.

 \begin{theorem}\label{thm_efficiency_convergence}
As $\|\partition\|\to 0$ we have:
    \begin{enumerate}
        \item[(a)] The communication inefficiency satisfies 
        \begin{align}
             \Lambda \leq E(\mathcal{P}_N)= \Lambda + O(\|\mathcal{P}_N\|).
        \end{align}
	    \item [(b)] The round trip rate for reversible PT, $\tau_{SEO}$, goes to zero:
	\begin{align}
	\tau_{\mathrm{SEO}}(\partition)&\sim \frac{1}{2N+2\Lambda} \longrightarrow 0.
	\end{align}
	    \item [(c)] The round trip rate for non-reversible PT, $\tau_{DEO}$ satisfies
	\begin{align}\label{optimal_round_trip_def}
	\tau_{\mathrm{DEO}}(\partition) \longrightarrow \bar \tau = \frac{1}{2+2\Lambda} > 0,
	\end{align}
	where the converges occurs with rate $O(\|\partition\|)$.
	\end{enumerate}
\end{theorem}
See Appendix \ref{app_efficiency_convergence} for a proof.

In general, $\Lambda$ is large when $\pi_0$ deviates significantly from $\pi$. Since $\Lambda$ is problem specific, this identifies a limitation of PT present even in its non-reversible flavour, namely that adding more cores to the task will never be harmful, but does have a diminishing return. The bound $(2+2\Lambda)^{-1}$ could indeed be very small for complex problems. Moreover, it is independent of the choice of annealing schedule, hence it cannot be improved by the schedule optimization procedure described in Section \ref{sec_computation}. 

When $N$ is fixed, we will see in Section \ref{sec_optimal_schedule} that $\tauDEO,\tauSEO$ are both maximized when $r^{(i-1,i)}=r^*$ for all $i$. In this case,  $E(\partition)=N (1-r^*)^{-1}~r^*$  and one has 
\begin{equation}\label{eq:rtr_bound_DEO_SEO}
\frac{N}{1+Nr^\ast}\leq\frac{\tauDEO(\partition)}{\tauSEO(\partition)}=\frac{N}{1+(N-1)r^*}\leq  \frac{1}{r^*}.
\end{equation}
This implies that for an optimally chosen schedule, non-reversible PT improves the round trip rate over reversible PT by at most a factor of $1/r^*$. Corollary \ref{cor_invariant} implies $Nr^*\approx\Lambda$, allowing us to quantify the looseness of \eqref{eq:rtr_bound_DEO_SEO} directly in terms of $\Lambda$.
\begin{equation}
 \frac{1}{r^*}-\frac{N}{1+Nr^*}
= \frac{1}{r^*}\left(\frac{1}{1+Nr^*}\right)
\approx\frac{1}{r^*}\left(\frac{1}{1+\Lambda}\right).
\end{equation}

We can also investigate the performance of PT for multi-modal or high-dimensional targets by studying the behaviour of $\Lambda$. Assuming the target distribution factorizes \cite{woodard2009conditions, atchade_towards_2011,roberts2014minimising}, we can show that $\Lambda$ is robust to the number of modes and grows at rate $O(\sqrt{d})$ as the dimension $d$ of the target increases. See Appendices \ref{sec_multimodal} and \ref{sec_high_dim} for details.

\section{Tuning non-reversible PT}\label{sec_computation}

\begin{figure}
	\centering
	\includegraphics[width=0.24\linewidth]{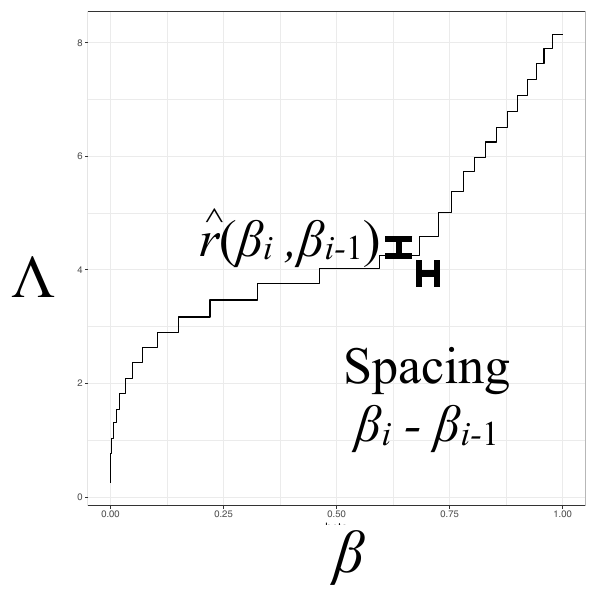}
	\includegraphics[width=0.24\linewidth]{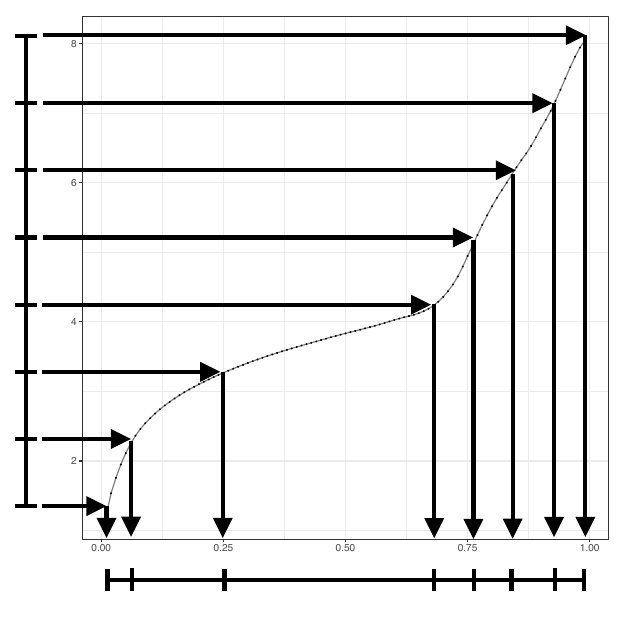}
	\includegraphics[width=0.22\linewidth]{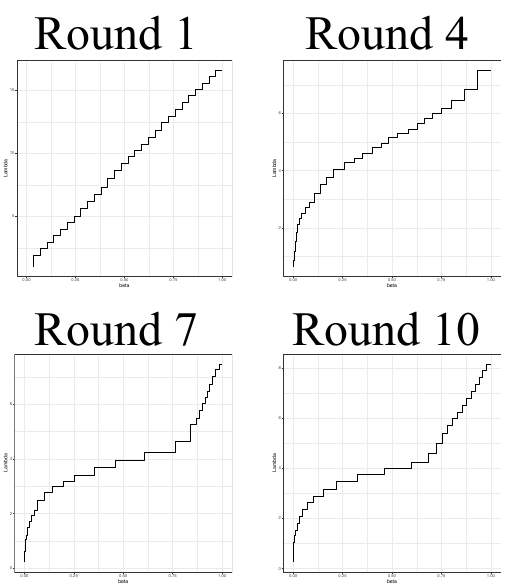}
	\includegraphics[width=0.25\linewidth]{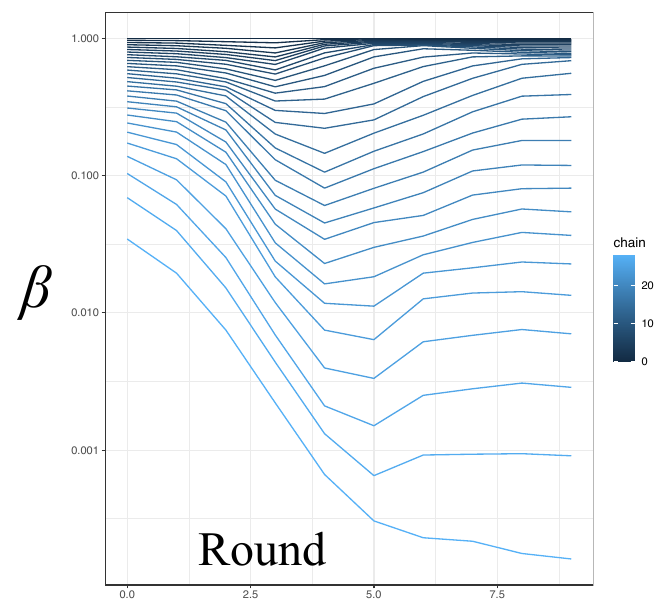}
	\caption{Proposed annealing schedule optimization method: example on the Bayesian mixture model of Section~\ref{sec_examples}. (1) The cumulative barrier $\Lambda(\cdot)$ is estimated using \eqref{eq_Lambda_estimate} at each point $\beta_k$ of an initial partition. (2) The cumulative barrier is interpolated using monotonic cubic interpolation, and a new schedule (shown beside the abscissa axis) is obtained by computing the inverse under $\Lambda(\cdot)$ of a regular grid (shown beside the ordinate). (3) The process 1--2 is repeated iteratively. (4) The sequence of annealing schedules obtained as a function of the round (colours represent different grid points in the schedule).}
	\label{fig:tuning}
\end{figure}

\subsection{Optimal annealing schedule}\label{sec_optimal_schedule}

We first discuss how to optimize the annealing schedule $\partition$ to maximize $\tauDEO$ when $N$ is fixed, which is equivalent to minimizing the schedule inefficiency, $E(\partition)=\sum_{i=1}^Nr_i/(1-r_i)$ where $r_i = r(\beta_{i-1},\beta_i)$. 
To get a tractable approximate characterization of the feasible region of $r_1, r_2, \dots, r_N$, we use Corollary \ref{cor_invariant}, which implies $\sum_{i=1}^N r_i\approx \Lambda$ for all schedules $\partition$. Therefore assuming $\|\partition\|$ is small enough to ignore the error term, finding the optimal schedule $\popt$ is approximately equivalent to minimizing $\sum_{i=1}^Nr_i/(1-r_i)$, subject to the constraint $\sum_{i=1}^Nr_i=\Lambda$ and $r_i>0$. Using Lagrange multipliers this leads to a solution where the rejection probabilities $r_i^*$ are constant in $i$, with a shared value denoted by $r^*$. 

Consequently the optimal schedule $\popt= \{\beta_0^*, \dots, \beta_N^*\}$ satisfies  $r(\beta^*_{i-1},\beta^*_i)=r^*$ for all $i$. Theorem \ref{theorem_rate_est} and Corollary \ref{cor_invariant} imply that $r^*$ must satisfy, $r^*\approx \Lambda(\beta_i^*)-\Lambda(\beta_{i-1}^*)$ for all $i$ and $r^* \approx \Lambda/N$ with $O(N^{-3})$ error. By equating these two estimates for $r^*$ and summing from $i=0,\dots,k$ we get
\begin{align}\label{beta_opt_constraint}
\Lambda(\beta_k^*) \approx \frac{k}{N} \Lambda,
\end{align}
with an error of $O(N^{-2})$. If we ignore error terms, \eqref{beta_opt_constraint} implies that $\beta_k^*\approx G(k/N)$ where $G=F^{-1}$ and $F(\beta)=\Lambda(\beta) / \Lambda$ (see Figure~\ref{fig:tuning}(2)). See Appendix~\ref{app_normal} for an example where an analytic result is available, however in general $\Lambda(\cdot)$ is not known but instead estimated from the MCMC output.

The ``equi-acceptance'' result in \eqref{beta_opt_constraint} is consistent with other theoretical frameworks and notions of efficiency \cite{atchade_towards_2011,lingenheil2009efficiency,kofke2002acceptance,predescu2004incomplete}. However implementing this equi-acceptance recommendation in practice is non-trivial. Previous work relied on Robbins-Monro schemes \cite{atchade_towards_2011,miasojedow2013adaptive}, which introduce sensitive tuning parameters. In contrast, provided $\Lambda(\cdot)$ can be estimated reliably (which we establish with Algorithm~\ref{alg_lambda} in the next section), \eqref{beta_opt_constraint} provides the basis for a straightforward and effective schedule optimization scheme, described in Algorithm~\ref{alg_update}.

\begin{algorithm}
	\caption{UpdateSchedule(communication barrier $\Lambda(\cdot)$, schedule size $N$)}\label{alg_update} 
	\begin{algorithmic}[1]
		\State $ \Lambda \gets  \Lambda(1)$
		\For{$k$ in $0,1, 2,..., N$}
			\State Find $\beta^*_k$ such that $ \Lambda(\beta^*_k) =  \frac{k}{N} \Lambda $ using e.g.\ bisection. 
		\EndFor
		\State \Return $\popt = (\beta^*_0, \beta^*_1, \beta^*_2, \dots, \beta^*_{N})$
	\end{algorithmic}
\end{algorithm}

\subsection{Estimation of the communication barrier}\label{sec_estimation_communication}

Assume we have access to a collection of samples $(\bar\rstates_1, \bar\rstates_2, \dots, \bar\rstates_n)$ from a non-reversible PT scheme based on an arbitrary annealing schedule $\partition$. These samples may come from a short pilot run, or, as described in the next section, from the previous round of an iterative scheme. For a given schedule $\partition$, when the central limit theorem for Markov chains holds, the Monte Carlo estimates for the rejection rates satisfy
\begin{equation}\label{def_rhat}
\hat{r}^{(i-1,i)} = \frac{1}{n}\sum_{k=1}^n\alpha^{(i-1,i)}(\bar\rstates_k) = r^{(i-1,i)} + O_p(n^{-1/2}).
\end{equation}
Next, using Theorem \ref{theorem_rate_est} we obtain $\sum_{i=1}^k r^{(i-1,i)} =  \Lambda(\beta_k) + O(N \|\partition\|^3)$. This motivates the following approximation for $\Lambda(\beta_k)$,
\begin{equation}\label{eq_Lambda_estimate}
\hat \Lambda(\beta_k) = \sum_{i=1}^k \hat r^{(i-1,i)},
\end{equation}
which has an error of order $O_p(\sqrt{N/n} + N \|\partition\|^3)$  (see Figure~\ref{fig:tuning} (left)). 

Given $\hat \Lambda(\beta_0),\dots, \hat \Lambda(\beta_N)$, we estimate the function $\Lambda(\beta)$ via interpolation, with the constraint that the interpolated function should be monotone increasing since $\lambda(\beta) \ge 0$. Specifically, we use the Fritsch-Carlson monotone cubic spline method \cite{fritsch_monotone_1980} and denote the monotone interpolation by $\hat \Lambda(\beta)$. While we only use $\hat \Lambda(\beta)$ in our schedule optimization procedure, it is still useful to estimate $\lambda(\beta)$ for visualization purpose. We use the derivative of our interpolation, $\hat \lambda(\beta) = \hat \Lambda'(\beta)$, which is a piecewise quadratic function.

The ideas described in this section so far are summarized in Algorithm \ref{alg_lambda}, which given rejection statistics collected for a fixed annealing schedule provides an estimate of the communication barrier. 

	\begin{algorithm}
		\caption{CommunicationBarrier(rejection rate $\{r^{(i-1,i)}\}$, schedule $\partition$)}\label{alg_lambda} 
		\begin{algorithmic}[1]
			\State For each $\beta_i \in \partition$, compute $\hat \Lambda(\beta_i)$ \Comment{Equation (\ref{eq_Lambda_estimate})}
			\State $S \gets \{(\beta_0, \hat \Lambda(\beta_0)), (\beta_1, \hat \Lambda(\beta_1)), \dots, (\beta_N, \hat \Lambda(\beta_N))\}$
			\State Compute a monotone increasing interpolation $\hat \Lambda(\cdot)$ of the points $S$ \Comment{e.g.\ using \cite{fritsch_monotone_1980}}
			\State \Return $\hat{\Lambda}(\cdot)$ 
		\end{algorithmic}
	\end{algorithm}
	
As a byproduct of Algorithm~\ref{alg_lambda} we also obtain a consistent estimator $\hat{\tau}=(2+2\hat{\Lambda})^{-1}$ for the optimal round trip rate $\bar{\tau}$, where $\hat \Lambda=\hat \Lambda(1)$. We show in Figure~\ref{fig:lambda-robustness-composite} an example where (A2) is severely violated, yet  $\lambda$ is still reliably estimated. This allows us to compare the empirically observed round trip rate against $\hat{\tau}$, and hence estimate how far an implementation deviates from optimal performance.

\subsection{Tuning $N$ }\label{sec_tuning_N}
	
Theorem \ref{thm_efficiency_convergence} shows that non-reversible PT does not deteriorate in performance as we increase $N$ unlike reversible PT, however the gains in round trip rate eventually become marginal. When $N$ is very large we expect to accumulate more round trips by running $k > 1$ parallel copies of PT. As we shall see, the large $N$ asymptotic is still required however in order to determine the optimal number $k^*$ of PT copies.  
	
Suppose there are $k$ copies of PT running in parallel consisting of $N+1$ chains with optimal annealing schedule $\partition$. If $\bar{N}$ is the total number of cores available then $k$ and $N$ satisfy the constraint $k(N+1)=\bar{N}$. By Corollary~\ref{cor_rtr} the total round trip rate across all $k$ copies of PT is 
	\begin{equation}\label{eq_roundtrip_kcopies}
	\tau= k\ \tau_{\mathrm{DEO}}(\partition)= \frac{\bar{N}}{2(N+1)(1+E(\partition))}.
	\end{equation}
	
From Section~\ref{sec_optimal_schedule}, the optimal schedule $\popt$ has a corresponding swap rejection rate $r^*= \Lambda/N+O(N^{-3})$. Substituting this into \eqref{eq_roundtrip_kcopies} we get $\tau= \tau_\Lambda(N)+O(N^{-1})$, where
\begin{align}\label{eq_tau_lambda_N}
\tau_\Lambda(N)
=\frac{\bar{N}(1-\Lambda/N)}{2(N+1)(1-\Lambda/N+\Lambda)}.
\end{align} 
Ignoring error terms, $\tau_\Lambda(\cdot)$ is optimized when we run $k^*=\bar{N}/(1+N^*)$ copies of PT with $N^*+1=2\Lambda+1$ chains to achieve an optimum round trip rate $\tau^*=k^*/(2+4\Lambda)$. 

Note that when $\tau$ is optimized, we have $r^*\approx 1/2$, which differs from from the $0.77$ optimal rejection rate from the reversible PT literature \cite{atchade_towards_2011,kofke2002acceptance,predescu2004incomplete}. 
In Section~\ref{sec_ELE_violation}, we show empirically that when the ELE assumption (A2) is severely violated, the recommendation in the present section is turned into a bound, $r^* < 1/2$, as increasing $N$ appears to alleviate ELE violations.

\subsection{Iterative schedule optimization}\label{sec_adaptive_algo}

Suppose we have $\bar{N}$ cores available and a computational budget of $\nscan$ scans of PT, one scan consisting in one application of $\K^{\textrm{PT}}$. 
The first $\ntune$ scans are used to find an accurate estimate of the communication barrier $\hat{\Lambda}(\cdot)$ and the remaining $\nsample=\nscan-\ntune$ scans to sample from the target. Algorithm \ref{alg_adaptive}  approximates $\hat{\Lambda}$ with error $O_p(\sqrt{\bar N/\ntune} + \bar N^{-2})$ by iteratively using  Algorithms \ref{alg_update} and \ref{alg_lambda} for {\texttt maxRound}$ = \log_2\ntune$ rounds until the tuning budget is depleted.  Equipped with $\hat \Lambda(\cdot)$, we can use the remaining $\nsample$ scans to run $k^*=\bar{N}/(N^*+1)$ copies of non-reversible PT with $N^*=2\hat\Lambda(1)$ chains with optimal schedule $\mathcal{P}_{N^*}^*$ from Algorithm \ref{alg_update}. In cases where it is suspected that (A2) may be severely violated (as evidenced for example by a large gap between $\bar \tau$ and the estimated $\tau$ as described in Section~\ref{sec_estimation_communication}), it may be advantageous to also attempt $N^* = \bar{N}$ in line \ref{line_optimalN} of Algorithm~\ref{alg_adaptive}. 
We show empirically in Figure~\ref{fig:all-models} (top) that in a range of synthetic and real world problems our scheme converges using a small number of iterations, in the order of {\texttt maxRound} $= 10$. In our experiments we used $\ntune = \nscan / 2$ and discarded the samples from the first $\ntune$ scans as burn-in.

\begin{algorithm}
	\caption{NRPT(number of cores $\bar{N}$, tuning budget $\ntune$, sampling budget $\nsample$)} \label{alg_adaptive} 
	\begin{algorithmic}[1]
		\State $\mathcal{P}_{\bar{N}} \gets$\ initial annealing schedule of size $\bar N+1$ (e.g.\ uniform)
		\State $\texttt{maxRound} \gets \log_2(\ntune)$ 
		\State $n \gets 1$ \Comment{Scans per round}
		\For{\texttt{round}  {\bf in} 1, 2, \dots, \texttt{maxRound}} \label{line:round}
			\State $\{ r^{(i-1,i)}\} \gets \textrm{DEO}(n, \mathcal{P}_{\bar{N}})$ \Comment{Algorithm \ref{alg_deo}}
			\State $\hat\Lambda(\cdot) \gets$CommunicationBarrier$(\{ r^{(i-1,i)}\},\mathcal{P}_{\bar{N}})$ \Comment{Algorithm \ref{alg_lambda}}
			\State $\mathcal{P}_{\bar{N}}\gets$UpdateSchedule($\hat\Lambda(\cdot),\bar N$) \Comment{Algorithm \ref{alg_update}}
			\State $n \gets 2 n$ \Comment{Rounds use an exponentially increasing number of scans}
		\EndFor
		\State $N^*\gets 2\hat\Lambda(1)$  \label{line_optimalN} \Comment{See Section~\ref{sec_tuning_N}}
		\State $\mathcal{P}^*\gets$ UpdateSchedule($\hat\Lambda(\cdot), N^*$)  \Comment{Optimal schedule}
		\State $k^*\gets \bar N/(N^*+1)$  \Comment{Number of copies of PT}
		\For{ 1,\dots, $k^*$}  
		    \State $(\bar\states_1,\dots,\bar\states_{\nsample})\gets \textrm{DEO}(\nsample, \mathcal{P}^*)$
		    \State \Return $(\bar\states_1,\dots,\bar\states_{\nsample})$
		\EndFor
	\end{algorithmic}
\end{algorithm}

\subsection{Normalizing constant computation}\label{sec:normalizing-constant-computation} 

Algorithm \ref{alg_adaptive} can be easily modified to obtain efficient estimators for quantities of the form $\int_0^1\E[f(V^{(\beta)},\beta)]\mathrm{d}\beta$ with an error of $O_p(\sqrt{\bar{N}/\ntune}+\bar{N}^{-2})$ for well-behaved functions $f$. Examples of such quantities include log-normalizing constants \cite{kirkwood1935statistical,gelman1998simulating,xie2011improving}, KL-divergence \cite{dabak2002relations}, and Fisher-Rao distance \cite{amari2016information} between $\pi$ and $\pi_0$.

For example, by taking a Riemann sum in the thermodynamic integration identity \cite{ogata_monte_1989}, the log-normalizing constant $\log\mathcal{Z}(1)$ of $\pi$ can be approximated using
\begin{equation}\label{eq:thermodymanic_integration}
\log \frac{\mathcal{Z}(1)}{\mathcal{Z}(0)}=-\int_0^1 \mu(\beta)\mathrm{d}\beta=-\sum_{i=1}^N\mu(\beta_i)(\beta_i-\beta_{i-1})+O(N^{-2}),
\end{equation}
where $\mu(\beta)=\E[V^{(\beta)}]$ and $\mathcal{Z}(0)$ the normalizing constant of $\pi_0$. By running DEO with schedule $\partition$ for $n$ scans we substitute a Monte Carlo estimate $\hat{\mu}(\beta_k)$ for $\mu(\beta_k)$ into \eqref{eq:thermodymanic_integration} to get an estimator $\log (\hat{\mathcal{Z}}(1)/\hat{\mathcal{Z}}(0))$ with error $O_p(\sqrt{N/n}+N^{-2})$. See also \cite{gelman1998simulating,xie2011improving} for closely related normalization constant estimators. Figure \ref{fig:all-models} (bottom) in Section \ref{sec:empirical} shows how the estimate for $\log \mathcal{Z}(1)$ evolves with the number of rounds in Algorithm \ref{alg_adaptive} for 16 different models.

\section{Scaling limits of index process}\label{sec:scalinglim}

The non-reversible communication scheme introduced in \cite{okabe2001replica} was presumably devised on algorithmic grounds (it performs the maximum number of swap attempts in parallel) since no theoretical justification was provided and the non-reversibility of the scheme was not mentioned. The arguments given in \cite{lingenheil2009efficiency} to explain the superiority of non-reversible communication over various PT algorithms rely on a misleading assumption, namely a diffusive scaling limit for the index process. Figure \ref{fig:swaps} (see also Figure \ref{fig_trajectory} in the Supplementary Material) suggests that the index process behaves qualitatively differently as $N$ increases for reversible and non-reversible PT. The goal of this section is to investigate these differences by identifying some scaling limits. Such limits exist under assumptions (A1)-(A3) specified in Section \ref{sec_assumptions} and Section \ref{section_swap_acceptance}. 

Suppose $G:[0,1]\to[0,1]$ is an increasing differentiable function satisfying $G(0)=0$ and $G(1)=1$. We say that $G$ is a \emph{schedule generator} for $\mathcal{P}_N=\{\beta_0,\dots,\beta_N\}$ if
$\beta_i=G(i/N)$. We will now assume without loss of generality that the sequence of schedules $\mathcal{P}_N$ are generated by some common $G$. In particular the mean value theorem implies $\|\mathcal{P}_N\|=O(N^{-1})$ as $N\to\infty$. This is not a strict requirement as most annealing schedules commonly used fall within this framework: the uniform schedule $\mathcal{P}_N^{\text{uniform}}=\{0,1/N,\dots,1\}$ is generated by $G(w)=w$, the optimal schedule $\popt=\{\beta_0^*,\dots,\beta_N^*\}$ derived in Section \ref{sec_optimal_schedule} is approximately generated by $G(w)=F^{-1}(w)$, where $F(\beta)=\Lambda(\beta)/\Lambda$. If $\pi_0(x)\propto \pi(x)^{\gamma}$ for some $\gamma\in (0,1)$, and $L(x)\propto \pi(x)^{1-\gamma}$ then $G(w)=\frac{\gamma^{1-w}-\gamma}{1-\gamma}$ corresponds to the geometric schedule commonly used by practitioners. In contrast, we show in Figure~\ref{fig:Gs} examples of schedule generators $G$ corresponding to estimated optimal $\popt$ from various models. Figure~\ref{fig:Gs} shows that in real world inference scenarios, the optimal $G$ often qualitatively differs from the simple geometric schedule.

Suppose $(I_n,\eps_n)$ is the index process for an annealing schedule $\mathcal{P}_N$ generated by $G$. To establish a scaling limit for $(I_n,\eps_n)$, it will be convenient to work in a continuous time setting. To do this, we suppose the times that PT iterations occur are distributed according to a Poisson process $\{M(\cdot)\}$ with mean $\mu_N$. The number $M(t)$ of PT iterations that occur by time $t\geq 0$ thus satisfies $M(t)\sim \mathrm{Poisson}(\mu_N t)$. We define the \emph{scaled index process} by $Z^N(t)=(W^N(t),\eps^N(t))$ where $W^N(t)=I_{M(t)}/N$ and $\eps^N(t)=\eps_{M(t)}$.

Define the piecewise-deterministic Markov process (PDMP) \cite{davis1993markov} $Z(t)=(W(t),\eps(t))$ on $[0,1]\times\{-1,1\}$ as follows: $W(t)$ moves in $[0,1]$ with velocity $\eps(t)$ and the sign of $\eps(t)$ is reversed at an inhomogeneous rate $\lambda(G(W(t))G'(W(t))$ or when $W(t)$ hits a boundary; see \cite{bierkens2018piecewise} for a discussion of PDMP on restricted domains. The process $Z$ corresponds to an inhomogeneous persistent random walk with reflective boundary conditions; see \cite{masoliver1992solutions} for a description of the Fokker-Plank equation and first passage times. As the optimal schedule is generated by $G=F^{-1}$ for $F(\beta)=\Lambda(\beta)/\Lambda$, we have $\lambda(G(w))G'(w)=\Lambda$ for all $w\in[0,1]$, so $\eps(t)$ changes sign at a constant rate $\Lambda$. See Figure \ref{fig_trajectory_Z} for sample trajectories of $Z$ for different values of $\Lambda$. The main result of this section is the following theorem.

\begin{theorem}\label{theorem_weak_limit_main} Suppose $\partition$ is family of schedules generated by $G$, then
\begin{enumerate}
    \item [(a)]For reversible PT if $\mu_N=N^2$ and if $W^N(0)$ converges weakly to $W(0)$ then $W^N$ converges weakly to a diffusion $W$ with initial condition $W(0)$, where $W$ is a Brownian motion on $[0,1]$ with reflective boundary conditions. The process $W$ admits $\mathrm{Unif}([0,1])$ as stationary distribution.
    \item [(b)]For non-reversible PT if $\mu_N=N$ and if $Z^N(0)$ converges weakly to $Z(0)$, then $Z^N$ converges weakly to the PDMP $Z$ with initial condition $Z(0)$. The process $Z$ admits $\mathrm{Unif}([0,1]\times\{-1,1\})$ as stationary distribution.
\end{enumerate}
\end{theorem}
See Appendix \ref{app_scaling_limit} for a detailed construction of the scaled index process via their infinitesimal generators, and the proof of Theorem \ref{theorem_weak_limit_main}.

Theorem \ref{theorem_weak_limit_main} implies for reversible PT that, if we speed time by a factor of $N^2$, then the index process scales to a diffusion $W$ independent of the choice of $\pi_0,\pi$ and of the schedule. This is in contrast to non-reversible PT where, if we speed time by a factor of $N$, the index process converges to a PDMP $Z$ depending on $\pi_0,\pi$ through $\lambda$ and the schedule through $G$.

\section{Examples}\label{sec_examples}

\subsection{Empirical behaviour of the schedule optimization method}\label{sec:empirical}

We applied the proposed NRPT algorithm to 16 models from the statistics and physics literature to demonstrate its versatility. 
The models include 9 Bayesian models ranging from simple standard models such as generalized linear models and Bayesian mixtures, to complex ones such as cancer copy-number calling, ODE parameter estimation, spike-and-slab classification and two types of phylogenetic models. This is complemented by three models from statistical mechanics and four artificial models. The datasets considered include eight real datasets spanning diverse data-types and size, e.g.\ state-of-the-art measurements such as whole-genome single-cell sequencing data (494 individual cells from two types of cancer, triple negative breast cancer and high-grade serous ovarian \cite{dorri_efficient_2020}) and mRNA transfection time series \cite{leonhardt_single-cell_2014}, as well as more conventional ones such as mtDNA data and various feature selection/classification datasets. See  Figure~\ref{tab:models} and references therein for details. See also Appendix~\ref{app_experiments} for more information on the models and additional experimental results. All models and algorithms are implemented in the Blang modelling language \cite{bouchard-cote_blang_2019}.

In each of the 16 models and for each round of Algorithm~\ref{alg_adaptive} (line~\ref{line:round}), we computed the mean swap acceptance probability across all neighbour chains. We then summarized these means across chains using a box plot. The equi-acceptance objective function of Section~\ref{sec_optimal_schedule} can be visually understood as collapsing this box plot into a single point.  We show in Figure~\ref{fig:all-models} (top) the progression of these swap acceptance probabilities. In all examples considered, equi-acceptance is well approximated within 10 rounds. 

Within the range of models considered, we observed a diversity of local barriers $\lambda$ estimated by Algorithm \ref{alg_adaptive} (Figure~\ref{fig:all-models} (middle)).  Most statistical models exhibit a high but narrow peak in the neighbourhood of the reference ($\beta = 0$). However, a subset of models including statistical models (mixture, ode, phylo-cancer, spike-slab) and physics models (Ising, magnetic, rotor) exhibit additional peaks away from $\beta = 0$. See Figure~\ref{fig:Gs} in Appendix~\ref{app:empirical} for the corresponding schedule generators $G$ and global barriers $\Lambda$.

\begin{figure}
	\begin{center}
	\scalebox{0.7}{
	\begin{tabular}{llllll}  
		\toprule
		Model (and dataset when applicable)    & $n$ & $d$ & $\hat \Lambda$ & $N$ \\
		\midrule
		\emph{Change point} detection (text message data, \cite{davidson-pilon_bayesian_2015}) & $74$ & $3$ & $4.0$ & $20$ \\
		\emph{Copy number} inference (whole genome ovarian cancer data, Section~\ref{sec:chromobreak}) & $6\;206$ & $30$ & $13.0$ & $50$ \\
		\emph{Discrete} multimodal distribution (Appendix~\ref{sec_discrete}) & N/A & $3$ & $0.4$ & $30$ \\
		Weakly identifiable \emph{elliptic} curve (Appendix~\ref{description-models-data}) & N/A & $2$ & $4.4$ & $30$ \\
		General Linear Model (\emph{GLM}) (Challenger O-ring dataset, \cite{dalal_risk_1989}) & $23$ & $2$ & $3.3$ & $15$ \\
		Bayesian \emph{hierarchical} model (historical rocket failure data, Appendix~\ref{description-models-data}) & $5\;667$ & $369$ & $12.0$ & $30$ \\
		\emph{Ising} model (Appendix~\ref{sec_ising}) & N/A & $25$ & $3.1$ & $30$ \\
		Ising model with \emph{magnetic} field (Appendix~\ref{sec_ising}) & N/A & $25$ & $2.3$ & $30$ \\
		Bayesian \emph{mixture} model (Section~\ref{sec:mixture-model}) & $300$ & $305$ & $8.2$ & $30$ \\
		Bayesian \emph{mixture} model (subset of 150 datapoints) & $150$ & $155$ & $5.5$ & $20$ \\
		Isotropic \emph{normal} distribution & N/A & $5$ & $1.4$ & $30$ \\
		\emph{ODE} parameters (mRNA data, \cite{leonhardt_single-cell_2014}) & $52$ & $5$ & $6.4$ & $50$ \\
		\emph{Phylogenetic} inference (single cell breast cancer data, \cite{dorri_efficient_2020})  & $192\;763$ & $192\;765$ & $88$ & $300$ \\ % 192763 is size of X matrix = 493 cells x 391 loci
		\emph{Phylogenetic species} tree inference (mtDNA, \cite{hayasaka_molecular_1988})  & $249$ & $10\;395$ & $7.1$ & $30$ \\ % 10395 = 9!!
		Unidentifiable \emph{product} parameterization (Appendix~\ref{description-models-data}) & $100\;000$ & $2$ & $3.7$ & $15$ \\
		\emph{Rotor} (XY) model, \cite{hsieh_finite-size_2013} & N/A & $25$ & $3.3$ & $40$ \\
		\emph{Spike-and-slab} classification (RMS Titanic passengers data \cite{Hind_2019}) & $200$ & $19$ & $4.7$ & $30$ \\
		\bottomrule
	\end{tabular}}
	\end{center}
	\caption{Summary of models used in the experiments, with the number of observations $n$ (when applicable), the number of latent random variables $d$, estimated $\hat \Lambda$, and the default number of chains used. An abbreviation for each is shown in italic.}
	\label{tab:models}
\end{figure}

\begin{figure}[ht]
	\begin{center} 
		\includegraphics[width=\linewidth]{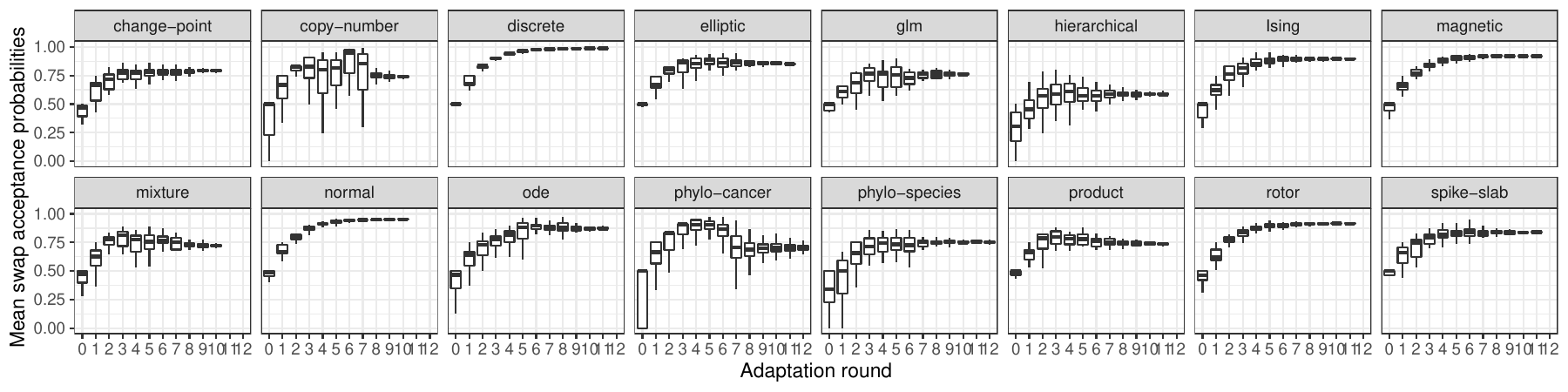}
		\includegraphics[width=\linewidth]{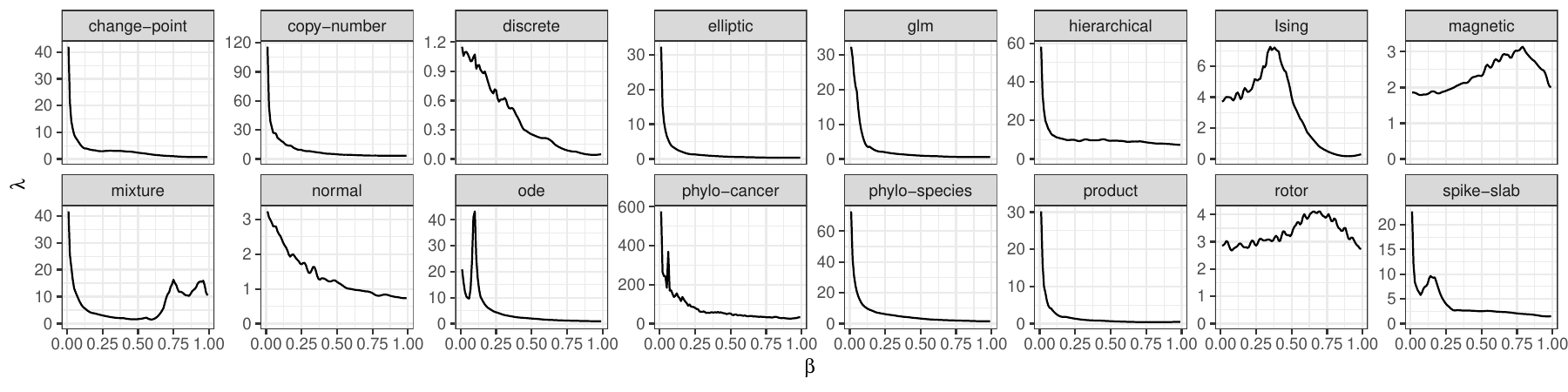}  
		\includegraphics[width=\linewidth]{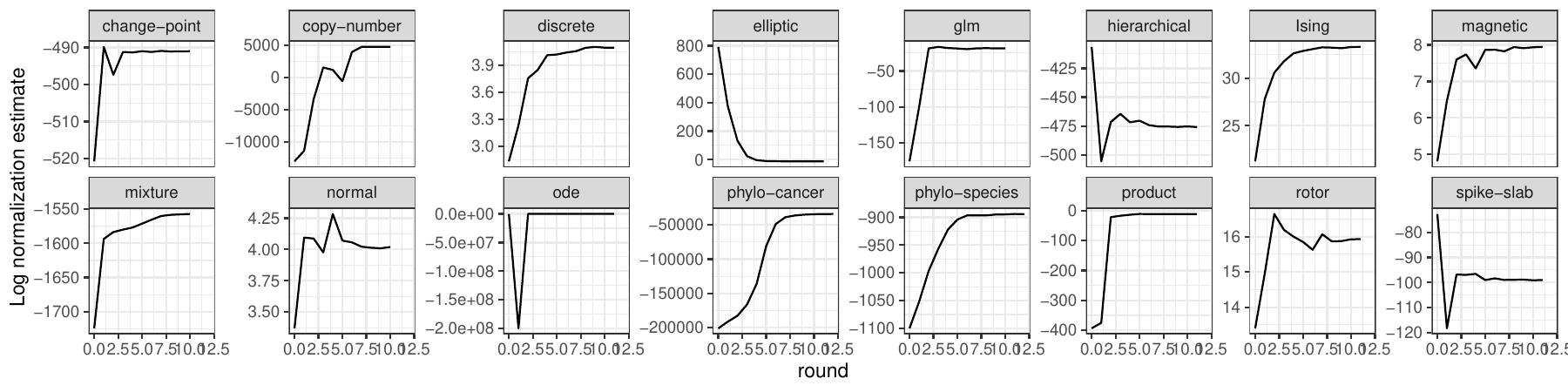}
	\end{center}
	\caption{Empirical behaviour of NRPT on 16 models. Top: distribution of acceptance rates observed at each round of Algorithm \ref{alg_adaptive}. Middle: estimates of the local communication barrier $\hat\lambda$. Bottom: progression of the log normalization constant estimates  (Section~\ref{sec:normalizing-constant-computation})  produced during each NRPT adaptation round.}
	\label{fig:all-models}
\end{figure}

\subsection{Robustness to ELE violation}\label{sec_ELE_violation}
 
To investigate empirically whether the NRPT methodology is robust to the violation of the ELE assumption, we ran NRPT with a range of values for $\nexpl$ on the models shown in Figure~\ref{fig:lambda-robustness-composite}. Let $d_{\text{var}}$ denote the number of variables in each model. We run experiments with $\nexpl = 0$, $(1/2)d_{\text{var}}$, $d_{\text{var}}$, $2d_{\text{var}}$, $4d_{\text{var}}$, \dots, $32d_{\text{var}}$ (the only exception is the reference chain ($\beta = 0$), where we always use $\nexpl = 1$ since we can get exact samples from $\pi_0$). 
The key quantity used by the NRPT algorithm for schedule optimization is the communication barrier $\lambda$. The results shown in Figure~\ref{fig:lambda-robustness-composite} demonstrate that in all models considered the function $\lambda$ is reliably estimated even when ELE is severely violated, provided $\nexpl > 0$. 
Similarly, since the estimates of $\Lambda$ and $\bar \tau$ are derived from $\lambda$, these quantities can be accurately estimated even when ELE is severely violated (see Figure~\ref{fig:tau-and-tau-bar} in Appendix~\ref{sec_ising}).
The results shown in Figure~\ref{fig_ESSvsRoundTrip} support that increasing $\nexpl$ reduces the difference between the theoretical and observed round trip rate. Moreover, Figure~\ref{fig_ESSvsRoundTrip} also demonstrate a second strategy to alleviate ELE violation, which is to increase $N$ while fixing $\nexpl$. 

In the next experiment shown in Figure~\ref{fig:vary-n-independent-chains}, we compare the two mechanisms available for alleviating ELE violation, namely increasing $N$ and increasing $\nexpl$.
First consider the black line in Figure~\ref{fig:vary-n-independent-chains} showing the regime where ELE is best approximated in these experiments. This is in close agreement with the theoretical guidelines developed in Section~\ref{sec_tuning_N}. At the same time, the light blue lines in Figure~\ref{fig:vary-n-independent-chains} show that in situations where the local exploration kernel is not easy to parallelize, it can be advantageous to decrease $\nexpl$ and to compensate with a number of chains $N$ higher than $2\Lambda$, leading to an average swap rejection rate $r < 1/2$.

 \begin{figure}
 	\begin{center} 
 		\includegraphics[width=1.0\linewidth]{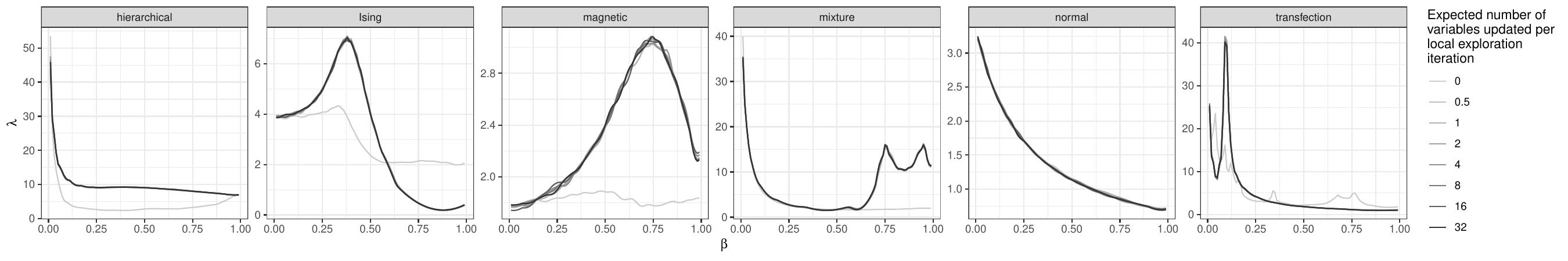}
 	\end{center}
 	\caption{Estimate $\hat{\lambda}$ of the local communication barrier $\lambda$ for different values of $\nexpl\geq 0$ and different models (refer to Figure \ref{tab:models} for model descriptions). The estimates are consistent as long as $\nexpl>0$, even when assumption (A2) is severely violated. Refer to Section \ref{sec_ELE_violation} for experiment details.}
 	\label{fig:lambda-robustness-composite}
 \end{figure}

\begin{figure}[ht]
	\centering
	\includegraphics[width=0.8\linewidth]{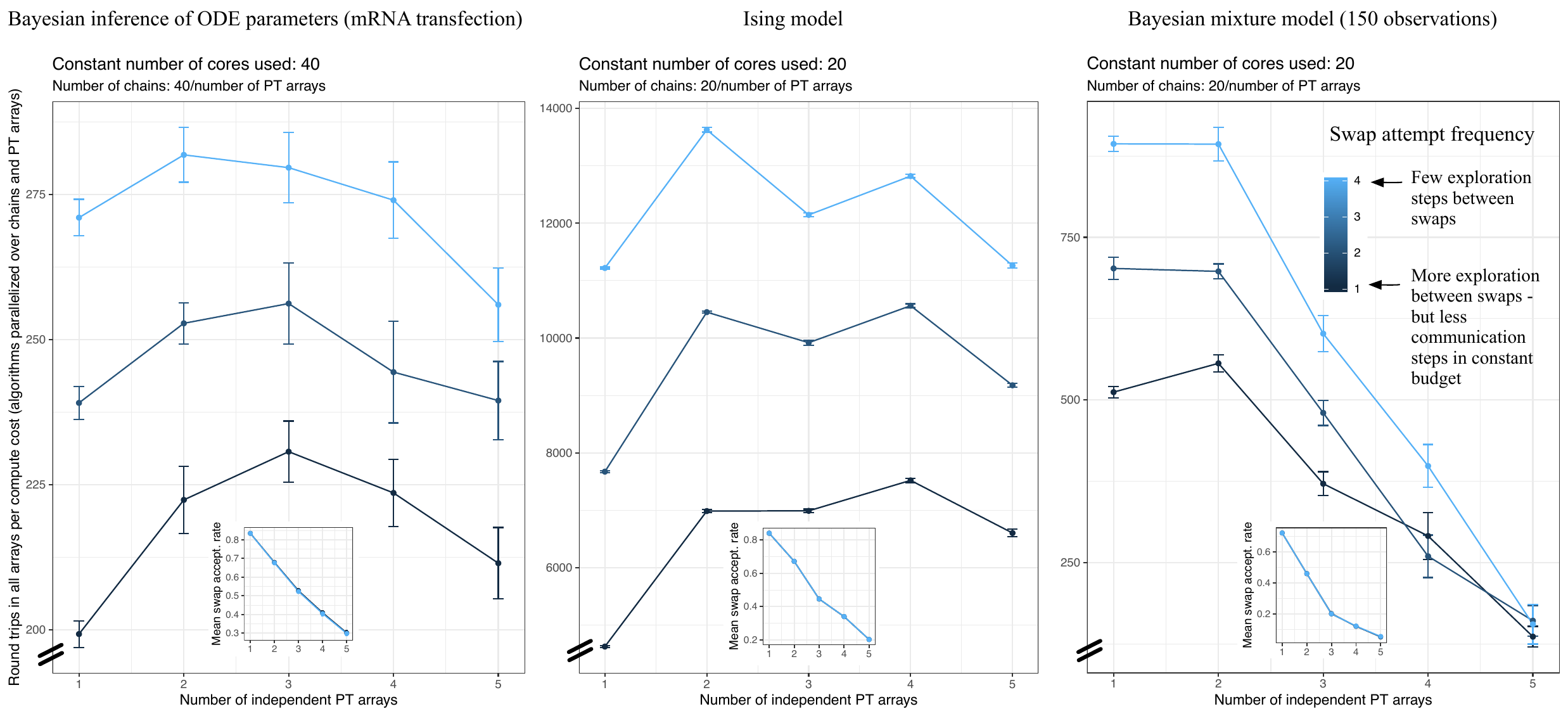}
	\caption{Trade-off between number of chains $N$, number of independent PT algorithms, $k$, and the frequency at which swaps are attempted ($\propto 1/\nexpl$). Each model uses a constant number of cores and set $N+1=\bar{N}/k$ chains per independent PT algorithm. Insets: swap acceptance probability for each $k$. The results in the black lines (where ELE is best approximated) agree with the theory of Section~\ref{sec_tuning_N}. For the ODE model, $\Lambda \approx 6.4$, $N^* \approx 12.8$ in agreement with $k^* = 3$ observed. For the Ising model, $\Lambda \approx 3.1$, $N^* \approx 6.2$ in agreement with $k^* = 4$ observed. For the Bayesian mixture model, $\Lambda \approx 5.5$, $N^* \approx 11$, in agreement with $k^* = 2$ observed. These results also show that in the context of local exploration moves that are not easily parallelized, it is better to use higher swap attempt frequencies (light blue lines) which achieve their optima at $N > 2 \Lambda$ and average swap rejection rate $r < 1/2$.}
	\label{fig:vary-n-independent-chains}
\end{figure}

\subsection{Comparison with other parallel tempering schemes}\label{sec_comparison}

We benchmarked the empirical running time of Algorithm \ref{alg_adaptive} compared to previous adaptive PT methods \cite{atchade_towards_2011,miasojedow2013adaptive}.
The methods we considered are: (1) the stochastic optimization adaptive method for reversible schemes proposed in \cite{atchade_towards_2011}; (2), a second stochastic optimization scheme, which still selects the optimal number of chains using the $23\%$ rule but uses an improved update scheme from \cite{miasojedow2013adaptive}, refer to Appendix~\ref{app_experiments_comparison} for details; (3) our non reversible schedule optimization scheme (NRPT); and finally, (4) our scheme, combined with a better initialization based on a preliminary execution of a sequential Monte Carlo algorithm (more precisely, based on a ``sequential change of measure,'' labelled SCM, as described in \cite{del_moral_sequential_2006}), we use this to investigate the effect on the violation of the stationarity assumption, and for fairness, we use this sophisticated initialization method for all the methods except (3). We benchmarked the methods on four models: (a) a 369-dimensional hierarchical model applied to a dataset of rocket launch failure/success indicator variables \cite{McDowell_2019}; (b) a 19-dimensional Spike-and-Slab variable selection model applied to the RMS Titanic Passenger Manifest dataset \cite{Hind_2019}; (c) A 25-dimensional Ising model from Section~\ref{sec_ising} ($M=5$); (d) a 9-dimensional model for an end-point conditioned Wright-Fisher stochastic differential equation (see, e.g., \cite{tataru_statistical_2017}). 

We provide the same number of cores and parallelized all the methods to the greatest extent possible. We refer the reader to Appendix \ref{app_experiments} for implementation details, experimental setup, and detailed description of the models and datasets used. 

In Figure~\ref{fig:benchmarking}, each dot summarized in the box plots represents the ESS per total wall clock time in seconds including schedule optimization time for the marginal of one of the model variables. The results show that schedule optimization is more efficient with our proposed non-reversible scheme, with a speed-up in the 10--100 range for the four models considered. The results also suggest that sequential Monte Carlo based initialization may not have a large impact on the performance of our method. See Figure~\ref{tab:summary-stat-comparisons} in Appendix~\ref{app:comparison} for other summary statistics from this benchmark.

\begin{figure}
\centering
 \includegraphics[width=0.49\linewidth]{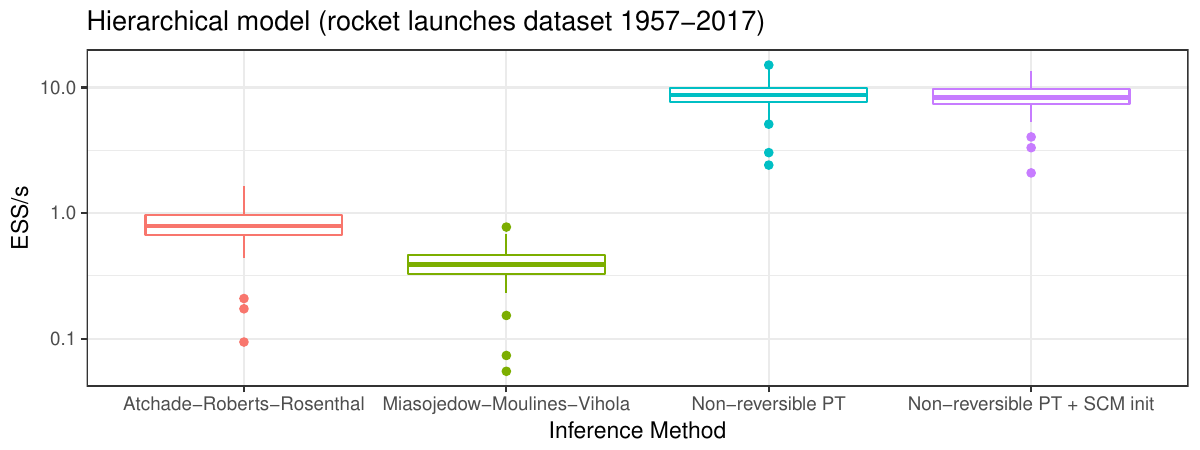}
 \includegraphics[width=0.49\linewidth]{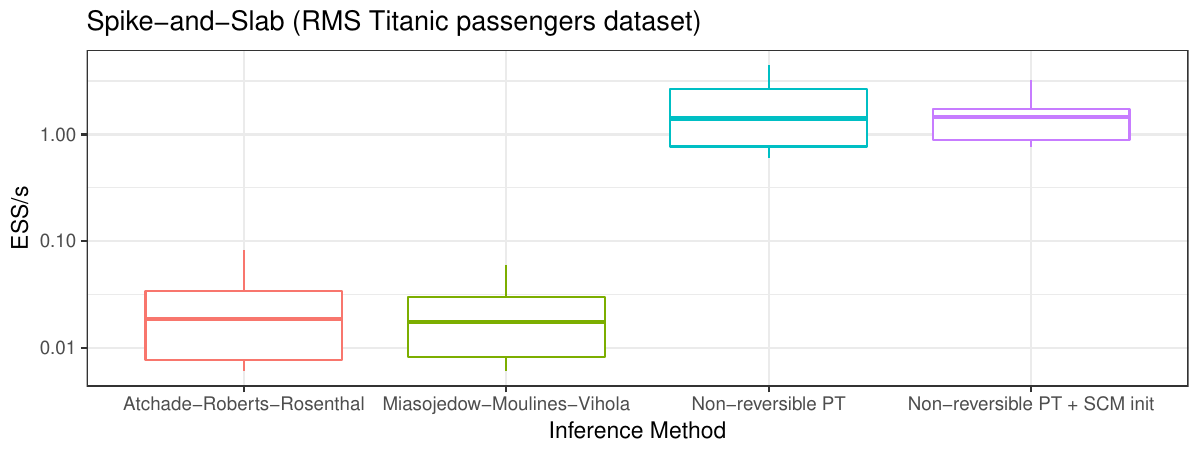} \\
 \includegraphics[width=0.49\linewidth]{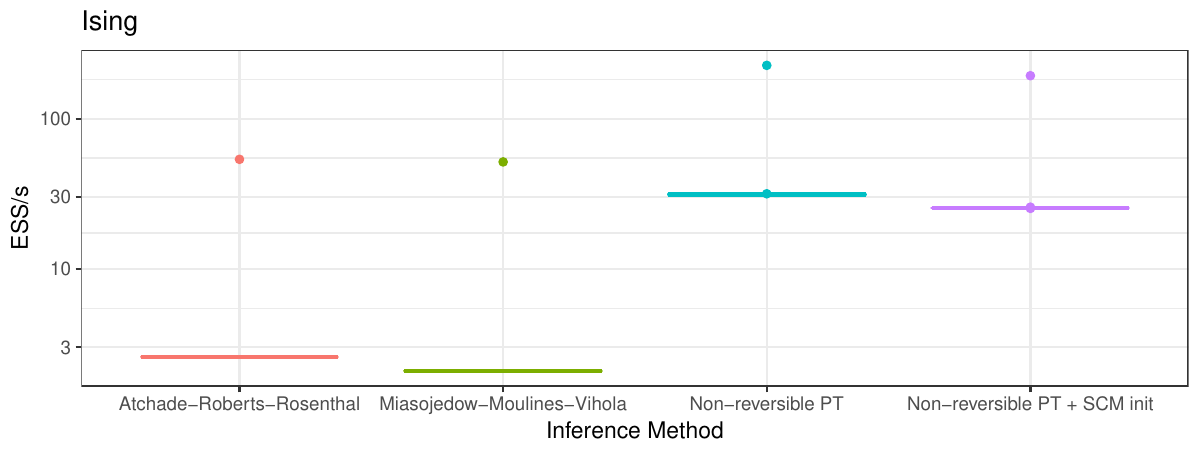}
 \includegraphics[width=0.49\linewidth]{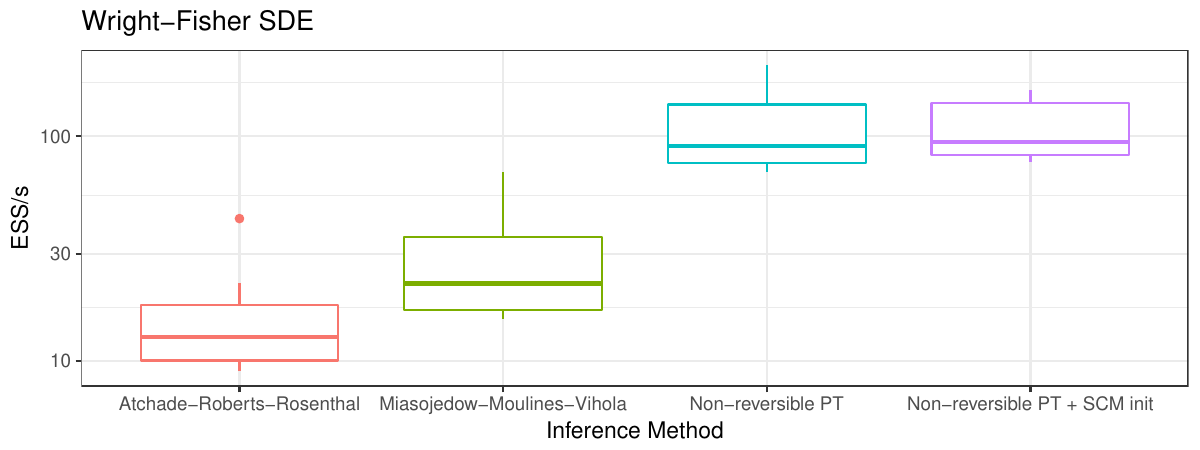}
\caption{Effective Sample Size (ESS) per second (ordinate, in log scale) for four PT methods (abscissa). The four facets show results for the four models described in Section~\ref{sec_comparison}.}
\label{fig:benchmarking}
\end{figure}

\subsection{Mixture models}\label{sec:mixture-example}

\begin{figure}[ht]
	\begin{center}
		\includegraphics[width=1.0\linewidth]{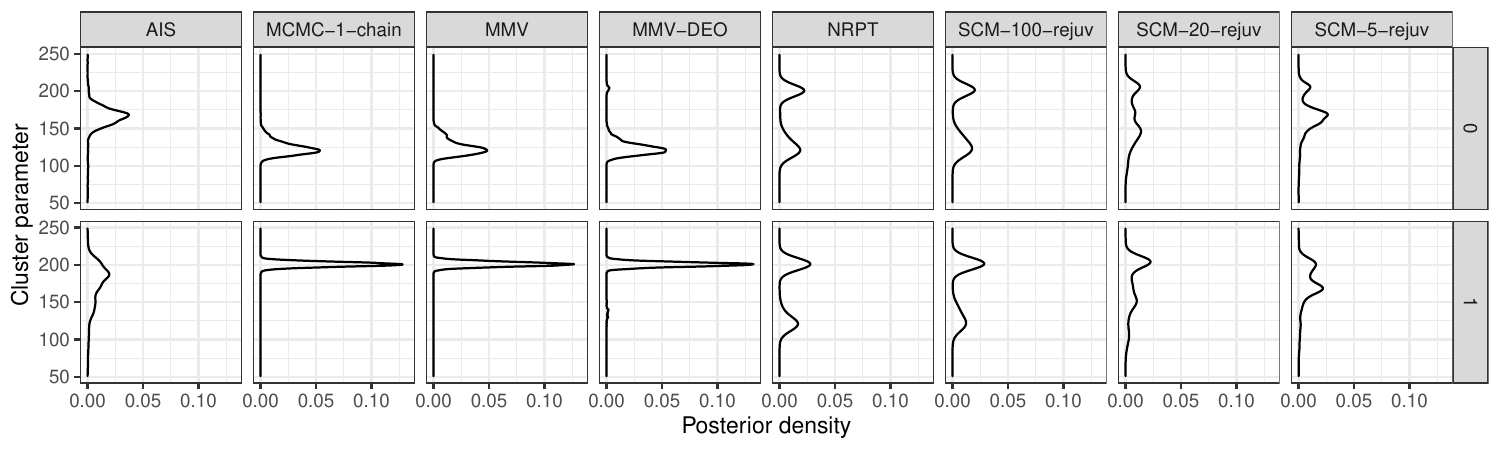} 
	\end{center}
	\caption{Mixture model example: eight approximations of the posterior distributions, showing for each approximation two of the 155 latent random variables, namely the two component mixtures' mean parameters.  Since the two latent random variables are exchangeable a posteriori (by label switching), the true marginal posterior distributions are identical for the two random variables. Amongst the approximation methods, multimodality is only captured by our algorithm (NRPT, 1.508min) and by a state-of-the-art sequential Monte Carlo combined with several rounds of rejuvenation (SCM-100, 5.064min). The benchmark is conservative in that all competing methods use a computational budget (in terms of both parallelism and wall clock time) greater or equal than NRPT.}
	\label{fig:mixture-means}
\end{figure}

Bayesian analysis of mixture models can give rise to a label-switching symmetry, leading to a  multimodal posterior distribution (see Appendix~\ref{sec:mixture-model}). 
In Figure~\ref{fig:mixture-means} we compare the following inference methods on a label-switching posterior distribution: our proposed algorithm (NRPT), an MCMC run based on a single chain (i.e.\ the exploration kernel alone), the stochastic optimization method of Miasojedow-Moulines-Vihola (MMV), DEO but optimized using MMV (MMV-DEO), Annealed Importance Sampling (AIS) \cite{neal_annealed_2001}, and a sequential Monte Carlo based on a sequence of annealed distributions \cite{del_moral_sequential_2006} (labelled SCM as before). For both SCM and AIS, we use the adaptive scheme of \cite{Zhou2016} (see Figure~\ref{fig:scm-diag} in Appendix~\ref{app:mixture-example} for diagnostics of the SCM adaptation).  The number of iterations are set so that the method with the smallest wall clock time is our proposed algorithm (NRPT: 1.508min, MMV: 2.019min, MMV-DEO:1.665min, AIS: 1.887min, SCM: 1.778min--5.064min). For all methods, timing includes the time spent to perform schedule optimization. 

Amongst MCMC methods, only NRPT correctly captures the multimodality of the target distribution. This is confirmed by the trace plots of the three MCMC methods, shown in Figure~\ref{fig:trace-plots-mcmcs} in Appendix~\ref{app:mixture-example}. MMV automatically selected $N=8$ by targeting an swap acceptance probability of $23\%$. Post-adaptation, the swap acceptance probability was $27\%$. For NRPT we use $N=30$ but to avoid penalizing MMV for a lack of parallelism, we limited all methods to use no more than $8$ threads. After schedule optimization, NRPT estimates the global communication barrier to a value of $\hat \Lambda \approx 8$, and the average swap acceptance probability was $72\%$ (see Figure~\ref{fig:nrpt-diag} in Appendix~\ref{app:mixture-example} for more NRPT diagnostics).

\subsection{Multimodality arising from single cell, whole genome copy number inference}\label{sec:chromobreak}

In this section we describe an application to copy number inference in which multimodality arises from the unknown ploidy of a cancer cell. The likelihood of this model is based on a hidden Markov model over $n=6\,206$ observations, and after analytic marginalization of the corresponding $6\,206$ hidden states, the multimodal sampling problem is defined over $d=30$ remaining latent variables. The model is described in Appendix~\ref{app:copy-number-inference}, in this section we summarize the results concerning the performance of the algorithms.

We compare the quality of posterior approximations from four samplers on the High-Grade Serous Ovarian cancer dataset from \cite{dorri_efficient_2020}. The four methods compared are: the stochastic optimization adaptive PT method MMV (adaptation selected $7$ chains, average post-adaptation swap acceptance, $28.7\%$), DEO but optimized using MMV (MMV-DEO, adaptation selected $11$ chains, average post-adaptation swap acceptance, $26.9\%$) our proposed algorithm algorithm (NRPT with 25 and 50 chains, average post-optimization swap acceptance of respectively $48.8\%$ and $73.5\%$), and inference based on a single MCMC chain. As in Section~\ref{sec:mixture-example}, to avoid overly penalizing MMV for a lack of parallelism, we limited all methods to use no more than $8$ threads, hence obtaining the following wall clock running times including schedule optimization when applicable: MMV, 15.91h; MMV+DEO, 14.33h; 25 chains NRPT, 8.809h; 50 chains NRPT, 15.35h; 1 chain MCMC, 2.91h. The running time of MMV was dominated by adaptation (88\% of the wall clock time), and is higher than the schedule optimization time used by NRPT (67\%).  

Refer to Figure~\ref{fig:nstates-traces-small} for a comparison of the trace plots between all three methods. Since the experimental protocol provides only proportionality between read counts and copy number, not absolute expected read count for a given copy number, the likelihood does not distinguish between a given copy number profile, and another one obtained by genome duplication of the same profile (refer to Appendix~\ref{app:copy-number-inference} for details). The prior favours the former, but when local copy number events occur after the genome duplication it may be possible to detect recently evolved genome duplication. Challenging multimodality arises when these subsequent events involve only small regions. In such case, the tension between the prior distribution favouring low ploidy and the noisy observation favouring higher ploidy creates a multimodal posterior distribution. We use data from one such cell in the following experiments to illustrate a realistic multimodal inference problem.

The multimodality of the posterior distribution was successfully captured by NRPT but not by a single-chain MCMC nor by the stochastic optimization method MMV (Figure~\ref{fig:nstates-traces-small}). MMV failed to achieve any round trip because the learnt schedule is highly suboptimal: while the mean swap acceptance probability across chains is $28.7\%$, the minimum across chains is nearly zero due to the noise in the optimization procedure. 
The difficulty of exploring this particular multimodal target is compounded by the fact that the state space of each chromosome's copy number is random, being upper-bounded by  random variables $m_c$ corresponding to each chromosome $c$. So in order to perform a jump doubling the cell's ploidy the sampler has to increase a large number of discrete variables $m_c$ simultaneously. This is reflected in the posterior distribution of the model variables denoted $m_c$, shown in Figure~\ref{fig:multimodality-chromo} in Appendix~\ref{app:chromobreak}. Hence the local exploration kernel, which in this example samples the variables $m_c$ one at the time, is insufficient to jump mode, and only excursions through the prior can achieve mode jumping, via a regeneration based on sampling from the prior at $\beta = 0$. 

NRPT used 50 chains, estimated $\hat \Lambda = 13$ (see Figure~\ref{fig:chromo-diag} in Appendix~\ref{app:chromobreak}), and converged quickly in this large scale example, performing drastic changes in the annealing schedule in the first 8 rounds then relatively little changes in rounds 8 through the final tenth round (Figure~\ref{fig:chromo-diag}). 

Comparing the performance of the two NRPT variants, we estimated a round trip rate of $9\times 10^{-4}$ for 25-chain, and $6\times 10^{-3}$ for 50-chain. Notice the increase being more than a factor two supports the use of large $N$ for exploration of complex multimodal posterior distributions. Neither MMV nor MMV-DEO achieved any round trips in the post-burn-in iterations.

\begin{figure}[ht]
	\begin{center}
		\includegraphics[width=\linewidth]{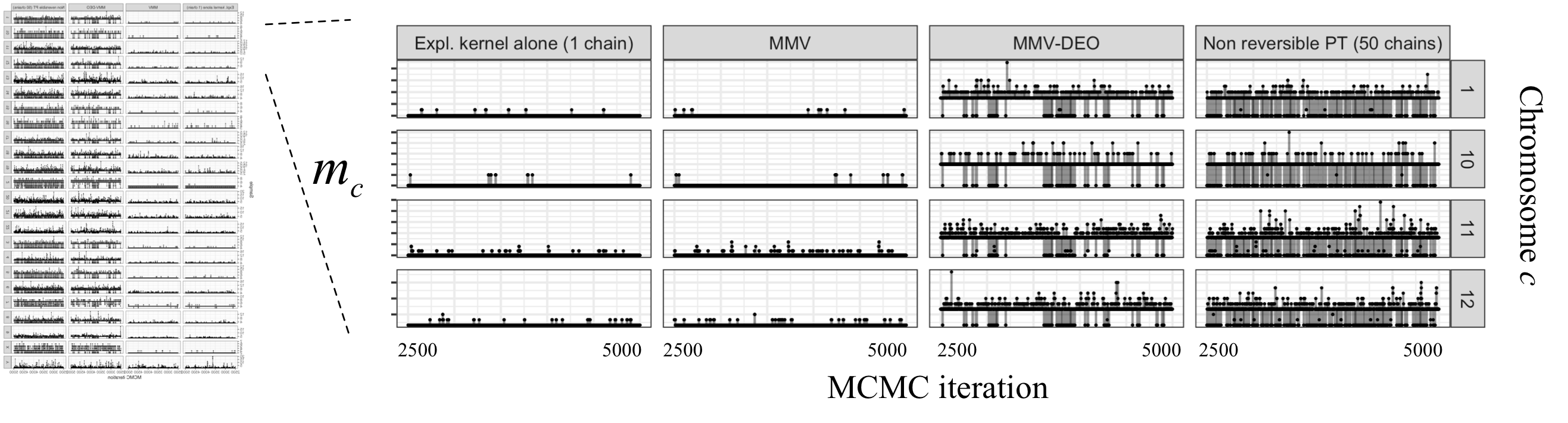}
	\end{center}
	\caption{Trace plots for the copy number inference problem. Each facet shows a post burn-in trace plot for one of the discrete random variables $m_c$. Facet rows are indexed by chromosomes. Facet columns are indexed by MCMC approximation methods. The trace plots show that only NRPT frequently jumps between the two modes. MMV and the exploration kernel alone do not achieve any jump, due to the required concerted changes in a large number of discrete variables. See Figure~\ref{fig:nstates-traces} in Appendix~\ref{app:chromobreak} for a larger version.}
	\label{fig:nstates-traces-small}
\end{figure}

\section{Discussion}
We have established that under the ELE assumption, non-reversible PT dominates reversible PT in terms of round trip rates for any choice of $N$ and annealing schedule $\partition$ (Corollary \ref{cor_rtr}). Moreover we have characterized the optimal round trip rate in terms of the communication barrier $\Lambda$, which we have used to identify an intrinsic limitation of both reversible and non-reversible PT (Section~\ref{sec_asymptotic_rtr}), and to design a novel schedule optimization method (Section~\ref{sec_computation}). 
We have shown that this schedule optimization method outperforms previous schedule optimizers in a range of numerical simulations, and we are already aware of other research groups that have successfully applied our methodology in astronomy \cite{collaboration_first_2021,issaoun2021persistent} and engineering \cite{langmore_hamiltonian_2021}. 

By combining (\ref{eq:rtr_bound_DEO_SEO}) with the estimate $r^* \approx 1/2$ derived from (\ref{eq_tau_lambda_N}), our theoretical results indicate a $1/r^* \approx 2$-fold increase in the round trip rate of DEO compared to SEO in an idealized setting where the ELE assumption holds and the schedule, including the number of chains, is well-tuned. 
The efficiency gains can be more pronounced during schedule optimization, in the course of which the number of chains can be initially larger than the optimal number of chains specified in Section~\ref{sec_tuning_N} as a consequence of Theorem~\ref{thm_efficiency_convergence}. 
Moreover, our experimental results show that previous methods to tune PT schedules can be fragile. 
Compared to these previous methods, we have observed speedups of one to two orders of magnitude thanks to our schedule optimization method.

\bibliographystyle{apalike}
\bibliography{report}
\clearpage

\begin{appendices}

\section*{\Large{Supplementary material}}
\section{Markov kernel for the index process}\label{app_markov_index_process}
For SEO, initialize $(I_0,\eps_0)=(i,\eps)$. We then define the Markov transition kernel, $(I_{n+1},\eps_{n+1})|(I_n,\eps_n)\sim P^{\mathrm{SEO}}((I_n,\eps_n),\cdot)$ in two steps. In the first step,  simulate
\begin{equation}\label{I_n_update}
I_{n+1} | (I_n,\eps_n)=(i,\eps) \sim 
\begin{cases}
(i + \eps) \wedge N \vee 0 &\text{ with probability }s^{(i,i+\eps)}, \\
i &\text{ otherwise,}
\end{cases}    
\end{equation}
where the expression ``$\wedge N \vee 0$" enforces the annealing parameter boundaries. In the second step, independently sample $\eps_{n+1} \sim \text{Unif}\{-1, +1\}$. 

Similarly for DEO, initialize $(I_0,\eps_0)=(i_0,\eps_0)$. Analogous to the SEO construction, we define $(I_{n+1},\eps_{n+1})|(I_n,\eps_n)\sim P^{\mathrm{DEO}}((I_n,\eps_n),\cdot)$ in two steps. We first update $I_{n+1} | (I_n,\eps_n)=(i,\eps)$ as in \eqref{I_n_update}, but apply the deterministic update in the second step,
\begin{equation}
\eps_{n+1} = 
\begin{cases}
\eps &\text{ if }I_{n+1}= i+\eps, \\
-\eps &\text{ otherwise.}
\end{cases}
\end{equation}

\section{Proof of Theorem \ref{prop_round_trip}}\label{app_round_trip}
To simplify notation for the rest of the proof, let $T_{\uparrow}$ and $T_{\downarrow}$ be the hitting times to the target and reference defined by,
\begin{align}
T_{\uparrow} = \min\{n: (I_n,\eps_n) = (N,1)\}, \;\;\;\ T_{\downarrow} = \min\{n: (I_n,\eps_n) = (0,-1)\}.
\end{align}
We will also denote
\begin{align}
s_i=s^{(i-1,i)},\qquad
r_i=r^{(i-1,i)}.    
\end{align}

\begin{enumerate}
    \item [(a)] If we define $a_\bullet^i=\E_{\mathrm{SEO}}(T_\bullet|I_0=i)$ for $i=0,\dots,N$ and $\bullet \in\{\uparrow,\downarrow\}$, then we have
    \begin{align}\label{round_trip_proof_rev}
        \E_{\mathrm{SEO}}(T)=a_{\uparrow}^0+a_{\downarrow}^N.
    \end{align}
    By the Markov property, for $i=1,\dots,N-1$, $a_{\bullet}^i$ satisfies the recursion
    \begin{align}\label{round_trip_proof_eq1}
        a_{\bullet}^i
        = \frac{1}{2}s_{i+1}(a^{i+1}_\bullet+1) + \frac{1}{2}s_{i}(a^{i-1}_\bullet+1)+\frac{1}{2}(r_{i+1}+r_i)(a^{i}_\bullet+1).
    \end{align}
    For $i=1,\dots,N$, we substitute in $b_{\bullet}^i=a_{\bullet}^i-a_\bullet^{i-1}$ into \eqref{round_trip_proof_eq1}. After simplification, $b^{i}_\bullet$ satisfies the following recursive relation
    \begin{align}\label{round_trip_proof_rev_recursion}
        -2 = s_{i+1}b^{i+1}_\bullet-s_ib^{i}_\bullet.
    \end{align}
    The solutions to \eqref{round_trip_proof_rev_recursion} are
    \begin{align}
        s_ib^{i}_\bullet  &=s_1b^{1}_\bullet-2(i-1),\label{round_trip_proof_rev_recursion_sol1}
    \end{align}
    or equivalently
    \begin{align}
        s_ib^{i}_\bullet 
        &=s_Nb^{N}_\bullet+2(N-i).\label{round_trip_proof_rev_recursion_sol2}
    \end{align}
    We now deal with the case of $\uparrow$ and $\downarrow$ separately.
    
    \begin{enumerate}
        \item To determine $a_\uparrow^0$, we note that if $I_0=0$ then $I_1=1$ with probability $\frac{1}{2}s_1$ and $I_1=0$ otherwise. So $a_\uparrow^0$ satisfies
    \begin{align}
        a^{0}_\uparrow=\frac{1}{2}s_1(a^1_\uparrow+1)+\left(1-\frac{1}{2}s_1\right)(a^0_\uparrow+1),
    \end{align}
    or equivalently
    \begin{align}\label{round_trip_proof_eq2}
        s_1b^{1}_\uparrow=-2.
    \end{align}
    Substituting this into \eqref{round_trip_proof_rev_recursion_sol1} implies $s_ib^{i}_\uparrow=-2i$. By summing $b^{i}_\uparrow=a^{i}_\uparrow-a^{i-1}_\uparrow$ from $i=1,\dots,N$ and, noting $a^{N}_\uparrow=0$, we get
    \begin{align}\label{round_trip_proof_A_01}
        a^{0}_\uparrow=\sum_{i=1}^N\frac{2i}{s_i}.
    \end{align}
    
    \item Similarly to determine $a_\downarrow^N$ we note that if $I_0=N$ then $I_1=N-1$ with probability $\frac{1}{2}s_N$ and $I_1=N$ otherwise. So $a_\downarrow^N$ satisfies
    \begin{align}
        a^{N}_\downarrow=\frac{1}{2}s_N(a^{N-1}_\downarrow+1)+\left(1-\frac{1}{2}s_N\right)(a^N_\downarrow+1),
    \end{align}
    or equivalently
    \begin{align}\label{round_trip_proof_eq4}
        s_Nb^{N}_\downarrow=2.
    \end{align}
  Substituting this into \eqref{round_trip_proof_rev_recursion_sol2} implies $s_ib^{i}_\downarrow=2+2(N-i)$. By summing $b^{i}_\downarrow=a^{i}_\downarrow-a^{i-1}_\downarrow$ from $i=1,\dots,N$ and, noting $a^{0}_\downarrow=0$, we get
    \begin{align}\label{round_trip_proof_A_N0}
        a^{N}_\downarrow=\sum_{i=1}^N\frac{2(N-i)+2}{s_i}.
    \end{align}
    \end{enumerate}

    Substituting in \eqref{round_trip_proof_A_01} and \eqref{round_trip_proof_A_N0} into \eqref{round_trip_proof_rev}, it follows that
    \begin{align}
        \E_{\mathrm{SEO}}(T)
        &=\sum_{i=1}^N\frac{2i}{s_i}+\sum_{i=1}^N\frac{2(N-i)+2}{s_i}\\
        &=2(N+1)\sum_{i=1}^N\frac{1}{s_i}\\
        &=2N(N+1)+2(N+1)\sum_{i=1}^N\frac{r_i}{s_i}.
    \end{align}

    \item[(b)] If we define $a^{i,\eps}_\bullet=\E_{\mathrm{DEO}}(T_\bullet|I_0=i,\eps_0=\eps)$ for $i=0,\dots, N$, $\eps \in \{+,-\}$ and $\bullet \in \{\uparrow,\downarrow\}$, then we have
    \begin{align}\label{round_trip_proof_L}
        \E_{\mathrm{DEO}}(T)=a^{0,-}_\uparrow+a^{N,+}_\downarrow.
    \end{align}
    Note that for $i=1,\dots,N-1$ $a_\bullet^{i,\eps}$ satisfies the recursion relations
    \begin{align}
        a_\bullet^{i,+} &= s_{i+1}(a_\bullet^{i+1,+} + 1) + r_{i+1}(a_\bullet^{i,-}+ 1),\label{round_trip_proof_L_eq1+}\\
        a_\bullet^{i,-} &= s_{i}(a_\bullet^{i-1,-} + 1) + r_{i}(a_\bullet^{i,-} + 1).\label{round_trip_proof_L_eq1-}
    \end{align}
    If we substitute $c_\bullet^i=a_\bullet^{i,+}+a_\bullet^{i-1,-}$, and $d_\bullet^i=a_\bullet^{i,+}-a_\bullet^{i-1,-}$ into \eqref{round_trip_proof_L_eq1+} and \eqref{round_trip_proof_L_eq1-} and simplify, we obtain
    \begin{align}
        a_\bullet^{i+1,+}-a_\bullet^{i,+}&=r_{i+1}d_\bullet^{i+1}-1\label{round_trip_proof_L_eq2+},\\
        a_\bullet^{i,-}-a_\bullet^{i-1,-}&=r_{i}d_\bullet^{i}+1.\label{round_trip_proof_L_eq2-}
    \end{align}
    By subtracting and adding \eqref{round_trip_proof_L_eq2+} and \eqref{round_trip_proof_L_eq2-}, we obtain a joint recursion relation for $c_\bullet^{i}$ and $d_\bullet^{i}$ of the form
    \begin{align}
        c_\bullet^{i+1}-c_\bullet^{i}&=r_{i+1}d_\bullet^{i+1}+r_{i}d_\bullet^{i},\label{round_trip_proof_L_eq3C}\\
        d_\bullet^{i+1}-d_\bullet^{i}&=r_{i+1}d_\bullet^{i+1}+r_{i}d_\bullet^{i}-2.\label{round_trip_proof_L_eq3D}
    \end{align}
    Note that \eqref{round_trip_proof_L_eq3D} can be rewritten as
    \begin{align}
        s_{i+1}d_\bullet^{i+1}-s_id_\bullet^{i}=-2.\label{round_trip_proof_L_Drecursion}
    \end{align}
    Once one has expressions for $c_\bullet^{i}$ and $d_\bullet^{i}$, then we can recover $a_\bullet^{i,\eps}$ by using
    \begin{align}
        a_\bullet^{i,+}&=\frac{c_\bullet^{i}+d_\bullet^{i}}{2},\label{round_trip_proof_L_CDtoA+}\\
        a_\bullet^{i-1,-}&= \frac{c_\bullet^{i}-d_\bullet^{i}}{2}.\label{round_trip_proof_L_CDtoA-}
    \end{align}
    We now deal with the $\uparrow$ and $\downarrow$ cases separately.
    
    \begin{enumerate}
        \item Note that $a_\uparrow^{0,-}=a_\uparrow^{0,+}+1$. We can substitute this into \eqref{round_trip_proof_L_eq2+} to get $s_1d_\uparrow^{1}=-2$, which combined with \eqref{round_trip_proof_L_Drecursion} implies
        \begin{align}
            s_id_\uparrow^{i}=-2i.\label{round_trip_proof_L_Dsol1}
        \end{align}
        Since $a_\uparrow^{N,+}=0$ we have $c_\uparrow^{N}=-d_\uparrow^{N}$, so by summing \eqref{round_trip_proof_L_eq3C} we get
        \begin{align}
            2a_\uparrow^{0,-}&=c_\uparrow^{1}-d_\uparrow^{1}\\
            &=c_\uparrow^{N}-d_\uparrow^{1}-\sum_{i=1}^{N-1}(c_\uparrow^{i+1}-c_\uparrow^{i})\\
            &=-d_\uparrow^{N}-d_\uparrow^{1}-\sum_{i=1}^{N-1}(r_{i+1}d_\uparrow^{i+1}+r_id_\uparrow^{i})\\
            &=-s_{N}d_\uparrow^{N}-s_1d_\uparrow^{1}-2\sum_{i=1}^Nr_id_\uparrow^{i}.\label{round_trip_proof_L_A-1}
        \end{align}
        After substituting in \eqref{round_trip_proof_L_Dsol1} into \eqref{round_trip_proof_L_A-1}, we obtain
        \begin{align}\label{round_trip_proof_A_0-sol}
            a_\uparrow^{0,-}=N+1+\sum_{i=1}^N\frac{2ir_i}{s_i}.
        \end{align}
        
        \item
        Note that $a_\downarrow^{N,+}=a_\uparrow^{N,-}+1$. We can substitute this expression into \eqref{round_trip_proof_L_eq2-} to get $s_Nd_\downarrow^{N}=2$, which combined with \eqref{round_trip_proof_L_Drecursion} implies
        \begin{align}
            s_id_\downarrow^{i}=2(N-i+1).\label{round_trip_proof_L_Dsol2}
        \end{align}
        Since $a_\downarrow^{0,-}=0$ we have $c_\downarrow^{1}=d_\downarrow^{1}$, so by summing \eqref{round_trip_proof_L_eq3C} we get
        \begin{align}
            2a_\downarrow^{N,+}&=c_\downarrow^{N}+d_\downarrow^{N}\\
            &=c_\downarrow^{1}+d_\downarrow^{N}+\sum_{i=1}^{N-1}(c_\downarrow^{i+1}-c_\downarrow^{i})\\
            &=d_\downarrow^{1}+d_\downarrow^{N}+\sum_{i=1}^{N-1}(r_{i+1}d_\downarrow^{i+1}+r_id_\downarrow^{i})\\
            &=s_{1}d_\downarrow^{1}+s_Nd_\downarrow^{N}+2\sum_{i=1}^Nr_id_\downarrow^{i}.\label{round_trip_proof_L_A+0}
        \end{align}
        After substituting in \eqref{round_trip_proof_L_Dsol2} into \eqref{round_trip_proof_L_A+0}, we obtain
        \begin{align}\label{round_trip_proof_A_N+sol}
            a_\downarrow^{N,+}=N+1+\sum_{i=1}^N\frac{2(N-i+1)r_i}{s_i}.
        \end{align}

    \end{enumerate}
    Finally, by substituting \eqref{round_trip_proof_A_0-sol} and \eqref{round_trip_proof_A_N+sol} into \eqref{round_trip_proof_L}, it follows that
    \begin{align}
        \E_{\mathrm{DEO}}(T)=2(N+1)+2(N+1)\sum_{i=1}^N\frac{r_i}{s_i}.
    \end{align}
\end{enumerate}

\section{Proof of Theorem \ref{theorem_rate_est}}\label{app_regularity}
We begin by recalling the following estimate for $r(\beta,\beta')$ from Proposition 1 in \cite{predescu2004incomplete}: 
\begin{proposition}\label{theorem_rate}\cite{predescu2004incomplete} Suppose $V^3$ is integrable with respect to $\pi_0$ and $\pi$. For $\beta\leq \beta'$, let $\bar{\beta}= \frac{\beta+\beta'}{2}$, then we have
\begin{align}
r(\beta,\beta')&=(\beta'-\beta)\lambda(\bar{\beta})+O(|\beta'-\beta|^3),\label{r_theorem}
\end{align}
where $\lambda$ satisfies \eqref{def_lambda}.
\end{proposition}

When $\beta<\beta'$, $(\beta'-\beta)\lambda(\bar{\beta})$ is the Riemann sum for $\int_{\beta}^{\beta'}\lambda(b)db$ with a single rectangle. Let $C^k([0,1])$ be the set of $k$ times continuously differentiable function on $[0,1]$. If $\lambda\in C^2([0,1])$, then standard midpoint rule error estimates yield
\begin{align}\label{riemann_sum_est}
\left|\int_{\beta}^{\beta'}\lambda(b)\mathrm{d}b-(\beta'-\beta)\lambda(\bar{\beta})\right|\leq \frac{1}{12}\left\|\frac{\mathrm{d}^2\lambda}{\mathrm{d}\beta^2}\right\|_\infty|\beta'-\beta|^3.
\end{align}
Therefore, if $\lambda\in C^2([0,1])$ then we can substitute
\eqref{riemann_sum_est} in \eqref{r_theorem} in Proposition \ref{theorem_rate}, to obtain Theorem \ref{theorem_rate_est}. This follows from Proposition \ref{prop_regularity}.

\begin{proposition}\label{prop_regularity}
If $V^{k}$ is integrable with respect to $\pi_0$ and $\pi$, then $\lambda\in C^{k-1}([0,1])$.		
\end{proposition}

\begin{proof}
Suppose $V^k$ is integrable with respect to $\pi_0$ and $\pi$, we want to show here that $\lambda:[0,1]\to\R_+$ given by
\begin{align}\label{lamda_integral_form}
\lambda(\beta)=\frac{1}{2}\int_{\mathcal{X}^2}|V(x)-V(y)|\pi^{(\beta)}(x)\pi^{(\beta)}(y)\mathrm{d}x\mathrm{d}y
\end{align}
is in $C^{k-1}([0,1])$. If we define $L(x,y)=L(x)L(y)$ and $\pi_0(x,y)=\pi_0(x)\pi_0(y)$, we can rewrite \eqref{lamda_integral_form} as
\begin{align}
\lambda(\beta)
&=\frac{1}{2\mathcal{Z}(\beta)^2}\int_{\mathcal{X}^2}|V(x)-V(y)|L(x,y)^\beta\pi_0(x,y)\mathrm{d}x\mathrm{d}y\\
&= \frac{g(\beta)}{2\mathcal{Z}(\beta)^2},
\end{align}
where $\mathcal{Z},g:[0,1]\to \R_+$ are defined by
\begin{align}
\mathcal{Z}(\beta)&=\int_\mathcal{X} L(x)^\beta\pi_0(x)\mathrm{d}x,\\
g(\beta)&=\int_{\mathcal{X}^2}|V(x)-V(y)| L(x,y)^\beta\pi_0(x,y)\mathrm{d}x\mathrm{d}y.
\end{align}
Since $\mathcal{Z}(\beta)>0$ on $[0,1]$, if we can show that $\mathcal{Z},g\in C^{k-1}([0,1])$ then it implies that $\lambda\in C^{k-1}([0,1])$. This is established in Lemma \ref{lemma_regularity2}.
\end{proof}

\begin{lemma}\label{lemma_regularity}
If $V^k$ is integrable with respect to $\pi_0$ and $\pi$ for $k\in \N$. Then for all $\beta\in [0,1]$, $j\leq k$, $V^j$ is $\pi^{(\beta)}$-integrable.
\end{lemma}

\begin{proof}
We begin by noting that for $L>0$ and $\beta\in[0,1]$, we have $L^\beta\leq 1+ L$. This implies
\begin{align}
&\int_\mathcal{X} |V(x)|^k\pi^{(\beta)}(x)\mathrm{d}x\\
=~&\frac{1}{\mathcal{Z}(\beta)}\int_\mathcal{X} |V(x)|^k L(x)^\beta\pi_0(x)\mathrm{d}x\\
\leq~& \frac{1}{\mathcal{Z}(\beta)}\int_\mathcal{X} |V(x)|^k \pi_0(x)\mathrm{d}x+\frac{1}{\mathcal{Z}(\beta)}\int_\mathcal{X} |V(x)|^k L(x)\pi_0(x)\mathrm{d}x\\
=~&\frac{\mathcal{Z}(0)}{\mathcal{Z}(\beta)}\int_\mathcal{X} |V(x)|^k\pi_0(x)\mathrm{d}x+\frac{\mathcal{Z}(1)}{\mathcal{Z}(\beta)}\int_\mathcal{X} |V(x)|^k\pi(x)\mathrm{d}x\\
<~&\infty.
\end{align}
Therefore since $V^k$ is $\pi_0$ and $\pi$-integrable, $V^k$ is $\pi^{(\beta)}$-integrable. Finally by Jensen's inequality we have for $j\geq k$,
\begin{align}
\int_\mathcal{X} |V(x)|^j\pi^{(\beta)}(x)\mathrm{d}x \leq \left(\int_\mathcal{X} |V(x)|^k\pi^{(\beta)}(x)\mathrm{d}x\right)^{\frac{j}{k}}<\infty.
\end{align}
\end{proof}

\begin{lemma}\label{lemma_regularity2} Suppose $V^k$ is integrable with respect to $\pi_0$ and $\pi$ for some $k\in\N$ then:
\begin{enumerate}
\item [(a)] $\mathcal{Z}\in C^k([0,1])$ with derivatives satisfying,
\begin{align}
\frac{d^j\mathcal{Z}}{d\beta^j}=\int_\mathcal{X} (-1)^j V(x)^j L(x)^\beta\pi_0(x)\mathrm{d}x,
\end{align}
for $j\leq k$.
\item [(b)] $g\in C^{k-1}([0,1])$ with derivatives satisfying,
\begin{align}
\frac{d^jg}{d\beta^j}=\int_{\mathcal{X}^2} (-1)^j |V(x)-V(y)|(V(x)+V(y))^j L(x,y)^\beta\pi_0(x,y)\mathrm{d}x\mathrm{d}y,
\end{align}
for $j< k$.
\end{enumerate}
\end{lemma}

\begin{proof}
\begin{enumerate}
\item [(a)] Let $h(x,\beta)=  L(x)^\beta\pi_0(x)=\exp(-\beta V(x))\pi_0(x)$ which satisfies
\begin{align}
 \frac{\partial^j }{\partial\beta^j}h(x,\beta)= (-1)^jV(x)^j L(x)^\beta\pi_0(x).
 \end{align}
Note for all $\beta\in[0,1]$ and $j\leq k$,
\begin{align}
\sup_{\beta\in[0,1]}\left|\frac{\partial^j }{\partial\beta^j}h(x,\beta)\right|&\leq  |V(x)|^j \pi_0(x)+ |V(x)|^j L(x)\pi_0(x).\label{h_bound}
\end{align}
The left hand side of \eqref{h_bound} dominates $ \frac{\partial^j h}{\partial\beta^j}$ uniformly in $\beta$ and is integrable by Lemma \ref{lemma_regularity}. The result follows using the Leibniz integration rule. 

\item [(b)] Let $\tilde{h}(x,y,\beta)=|V(x)-V(y)| L(x,y)^\beta\pi_0(x,y)$. By using $\log L(x,y)=-V(x)-V(y)$, we get
\begin{align}
 \frac{\partial^j }{\partial\beta^j}\tilde{h}(x,y,\beta) = (-1)^j|V(x)-V(y)|(V(x)+V(y))^j L(x,y)^\beta\pi_0(x,y).
\end{align}
Similar to (a), we have for all $\beta\in[0,1]$, $j\leq k-1$,
\begin{align}
\sup_{\beta\in[0,1]}\left| \frac{\partial^j }{\partial\beta^j}\tilde{h}(x,y,\beta)\right|
 \leq~&  |V(x)-V(y)||V(x)+V(y)|^j \pi_0(x,y) \notag\\
&+  |V(x)-V(y)||V(x)+V(y)|^j L(x,y)\pi_0(x,y),\label{htilde_bound}
\end{align}
The left hand side of \eqref{htilde_bound} dominates $ \frac{\partial^j \tilde{h}}{\partial\beta^j}$ uniformly in $\beta$. It is integrable by Lemma \ref{lemma_regularity} and using the fact that $V^k$ is integrable with respect to $\pi_0$ and $\pi$. The result follows using the Leibniz integration rule.
\end{enumerate}
\end{proof}

\section{Proof of Theorem \ref{thm_efficiency_convergence}}\label{app_efficiency_convergence}
We first note that (b) and (c) follow immediately from (a) and Corollary \ref{cor_rtr}. So to prove Theorem \ref{thm_efficiency_convergence} it is sufficient to show (a).

We use the fact that for all $\mathcal{P}_N$
\begin{align}\label{eff_est1}
    \sum_{i=1}^Nr^{(i-1,i)}\leq E(\mathcal{P}_N)\leq \frac{1}{\min_js^{(j-1,j)}}\sum_{i=1}^Nr^{(i-1,i)}.
\end{align}
By Theorem \ref{theorem_rate_est}, we have
$\min_j s_j=1+O(\|\mathcal{P}_N\|)$ and, by Corollary \ref{cor_invariant}, we have  $\sum_{i=1}^Nr^{(i-1,i)}=\Lambda+O(N\|\mathcal{P}_N\|^3)$ which combined with \eqref{eff_est1} implies
\begin{equation}
    E(\mathcal{P}_N)=\Lambda+O(\|\mathcal{P}_N\|).
\end{equation}
Therefore $E(\mathcal{P}_N)$ converges to $\Lambda$ at a $O(\|\mathcal{P}_N\|)$ rate as $\|\mathcal{P}_N\|\to 0$.

\section{Multimodal decomposition of communication barrier}\label{sec_multimodal}

Since PT is often used to sample from multimodal targets, it is natural to ask how the communication barrier behaves under the presence of modes. Similar to \cite{woodard2009conditions}, we partition $\statespace$ into the disjoint union $\statespace=\bigcup_{k=1}^K\statespace_k$ where we have $\statespace_k$ represents region of the $k$-th mode of $\pi$ and the target decomposes as a mixture of its modes its modes $\pi(x)=\sum_{k=1}^Kp_k\pi(x|\statespace_k)$ where $p_k=\pi(\statespace_k)$ and $\pi(x|\statespace_k)=p_k^{-1}\pi(x) \mathbb{I}_{\statespace_k}(x)$ are the probability mass and distribution of the $k$-th mode. If we assume that the reference distribution puts the same relative mass on each mode as the target, $\pi_0(\statespace_k)=\pi(\statespace_k)$, then $\pi^{(\beta)}$ decomposes as
\begin{equation}\label{pi_anneal_mode_def}
\pi^{(\beta)}(x)\propto \sum_{i=1}^Kp_k\pi^{(\beta)}(x|\statespace_k).
\end{equation}
Similarly $V(x)=\sum_{k=1}^KV_k(x)\mathbb{I}_{\statespace_k}(x)$ where $V_k(x)=-\log(\pi(x|\statespace_k)/\pi_0(x|\statespace_k))$. Define $\lambda_{k,k'}(\beta)$ and $\Lambda_{k,k'}$ as the local and global communication barrier between mode $k$ and $k'$ by
\begin{align}
\lambda_{k,k'}(\beta)=\frac{1}{2}\E\left[|V_k^{(\beta)}-V_{k'}^{(\beta)}|\right],
\qquad \Lambda_{k,k'}=\int_0^1\lambda_{k,k'}(\beta)\mathrm{d}\beta,
\end{align}
where $V_k^{(\beta)}\stackrel{d}{=}V_k(X_k^{(\beta)})$ for $X_k^{\beta}\sim \pi^{(\beta(}(x|\statespace_k)$ and $V_k^{(\beta)},V_{k'}^{(\beta)}$ are independent. An immediate consequence of \eqref{def_lambda} in Theorem \ref{theorem_rate_est} implies the following decomposition of the communication barrier.
\begin{proposition}[Multimodal decomposition]\label{prop_multimodal}
If $\pi_0(\statespace_k)=\pi(\statespace_k)$ then,
\begin{align}\label{mode_local}
\lambda(\beta)=\sum_{k=1}^K\sum_{k'=1}^Kp_{k}p_{k'}\lambda_{k,k'}(\beta),
\qquad \Lambda=\sum_{k=1}^K\sum_{k'=1}^Kp_{k}p_{k'}\Lambda_{k,k'},
\end{align}
\end{proposition}
In particular \eqref{mode_local} implies $\Lambda$ is a weighted average over the communication barriers between modes. In the particular case where the modes are exchangeable, we have $\Lambda_{k,k'}=\Lambda$ and thus $\Lambda$ is exactly invariant to $K$. We remark that the meta-model used in this section is not realistic for practical problems as we do not know a priori the location of the modes. This meta-model is only used to obtain intuition on the multimodal scalability of the method. Section~\ref{sec:empirical} illustrates the empirical behaviour of the method in several genuinely challenging multimodal problems.

\section{High-dimensional scaling of communication barrier.}\label{sec_high_dim}
We determine here the asymptotic behaviour of $\lambda$ and $\Lambda$ when the dimension of $\statespace$ is large. To make the analysis tractable, we assume that $\pi_d(x)=\prod_{i=1}^d\pi(x_i)$ as in \cite{atchade_towards_2011,roberts2014minimising}. This provides a model for weakly dependent high-dimensional distributions. We only make this structural assumption on the state space and distribution to establish Proposition \ref{prop_high_dim} below.

The corresponding annealed distributions are thus given by
\begin{equation}\label{pi_d_def}
\pi_d^{(\beta)}(x)=\prod_{i=1}^d \pi^{(\beta)}(x_i).
\end{equation}
Let $\lambda_d$ and $\Lambda_d$ be the local and global communication barriers for $\pi_d$ respectively. 
\begin{proposition}[High Dimensional Scaling]\label{prop_high_dim}
Define $\sigma^2(\beta)=\Var(V^{(\beta)})$, for all $\beta \in[0,1]$ we have as $d\to\infty$,
\begin{align}\label{high_dim_local}
    \lambda_d(\beta)\sim \sqrt{\frac{d}{\pi}}\sigma(\beta),\qquad
    \Lambda_d \sim \sqrt{\frac{d}{\pi}}\int_0^1\sigma(\beta)\mathrm{d}\beta.
\end{align}
\end{proposition}
It follows from Proposition \ref{prop_high_dim} that $\lambda_d$ and $\Lambda_d$ increase at a $O(d^{1/2})$ rate as $d\to\infty$.

\begin{proof}
For $k=1,2$, let us define $\textbf{V}_k^{(\beta)}\stackrel{d}{=}\textbf{V}(X_k^{(\beta)})$ where
$X_k^{(\beta)}\sim\pi_d^{(\beta)}$ and $\textbf{V}(x)=\sum_{i=1}^dV(x_i)$. The independence structure from Equation \eqref{pi_d_def} tells us that $\textbf{V}_k^{(\beta)}$ can be decomposed as  $\textbf{V}_k^{(\beta)}=\sum_{i=1}^dV_{ki}^{(\beta)}$ where $V_{ki}^{(\beta)}$ are iid with a distribution identical to $V^{(\beta)}$, and therefore
\begin{equation}
 \textbf{V}_1^{(\beta)}-\textbf{V}_2^{(\beta)}=\sum_{i=1}^d  V_{1i}^{(\beta)}-V_{2i}^{(\beta)}.
\end{equation}
The random variables $\{V_{1i}^{(\beta)}-V_{2i}^{(\beta)}\}_{i=1}^d$ are iid with mean zero and variance $2\sigma^2(\beta)$. By the central limit theorem,
\begin{equation}\label{CLT_rate}
\frac{\textbf{V}_1^{(\beta)}-\textbf{V}_2^{(\beta)}}{\sqrt{2\sigma^2(\beta)d}}
=\frac{1}{\sqrt{d}}\sum_{i=1}^d \frac{V_{1i}^{(\beta)}-V_{2i}^{(\beta)}}{\sqrt{2\sigma^2(\beta)}}
\xRightarrow[d\rightarrow\infty]{} \tilde{Z} \sim N(0,1).
\end{equation}
Thus we have
\begin{align}
\lambda_d(\beta)
&=\frac{1}{2}\E\left[| \textbf{V}_1^{(\beta)}-\textbf{V}_2^{(\beta)}|\right]\\
&=\frac{1}{2}\sqrt{2\sigma^2(\beta) d}\E\left[\left|\frac{\textbf{V}_1^{(\beta)}-\textbf{V}_2^{(\beta)}}{\sqrt{2\sigma^2(\beta)d}}\right|\right].\label{high_dim_lam_int}
\end{align}
The sequence of variables indexed by $d$ in the expectation in \eqref{high_dim_lam_int} is also uniformly integrable. This follows by noting that the second moment of the integrand in \eqref{high_dim_lam_int} is uniformly bounded in $d$:
\begin{align}
    \sup_{d} \E\left[\left|\frac{\textbf{V}_1^{(\beta)}-\textbf{V}_2^{(\beta)}}{\sqrt{2\sigma^2(\beta)d}}\right|^2\right] =\sup_d \frac{1}{2\sigma^2(\beta)d}\sum_{i=1}^d\mathrm{Var}\left[V_{1i}^{(\beta)}-V_{2i}^{(\beta)}\right]
    =1.
\end{align}
By $d\to\infty$ and using \eqref{CLT_rate} we have
\begin{equation}
\lim_{d\to\infty}\sqrt{\frac{2}{\sigma^2(\beta)d}}\lambda_d(\beta)=\E|\tilde{Z}|=\sqrt{\frac{2}{\pi}},
\end{equation}
which proves \eqref{high_dim_local}. 

To obtain the dimensional scaling limit for $\Lambda_d$, we use Jensen's inequality, 
\begin{align}
    \frac{\lambda_d(\beta)}{\sqrt{d}}
    &= \frac{1}{2\sqrt{d}}\E\left[| \textbf{V}_1^{(\beta)}-\textbf{V}_2^{(\beta)}|\right]\\
    &\leq \frac{1}{2\sqrt{d}}\sqrt{\Var\left[| \textbf{V}_1^{(\beta)}-\textbf{V}_2^{(\beta)}|\right]}\\
    &=\frac{\sigma(\beta)}{\sqrt{2}}.\label{high_dim_domination}
\end{align}
Finally, \eqref{high_dim_local}, \eqref{high_dim_domination} along with dominated convergence yield
\begin{align}
    \lim_{d\to\infty}\frac{\Lambda_d}{\sqrt{d}}
    =\int_0^1\lim_{d\to\infty}\frac{\lambda_d(\beta)}{\sqrt{d}}\mathrm{d}\beta 
    =\int_0^1\frac{\sigma(\beta)}{\sqrt{\pi}}\mathrm{d}\beta.
\end{align}
\end{proof}

\section{Scaling limit of index process}\label{app_scaling_limit}

\subsection{Scaled index process}\label{sec_scaled_index}

\begin{figure}
	\begin{center}
		\includegraphics[width = 0.9\linewidth]{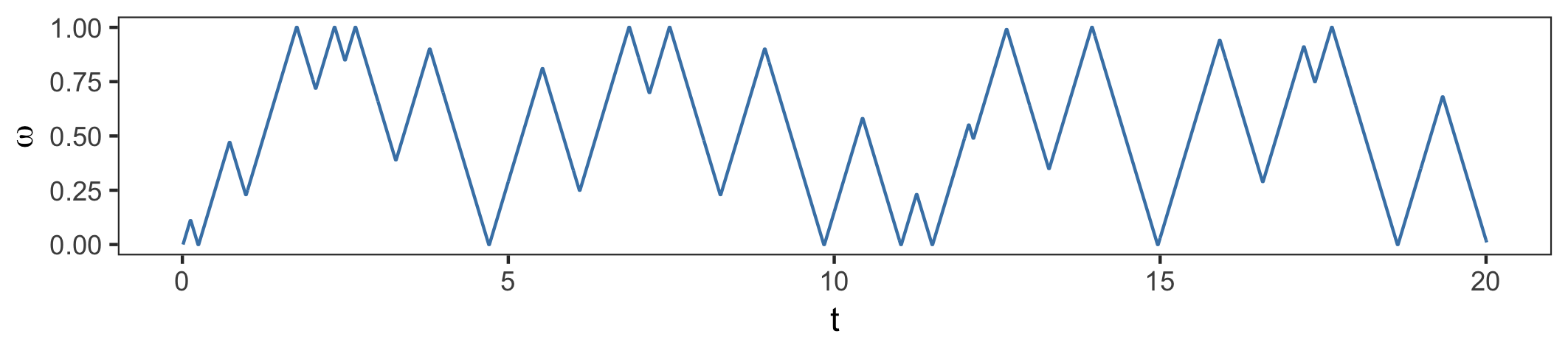}\\
		\includegraphics[width = 0.9\linewidth]{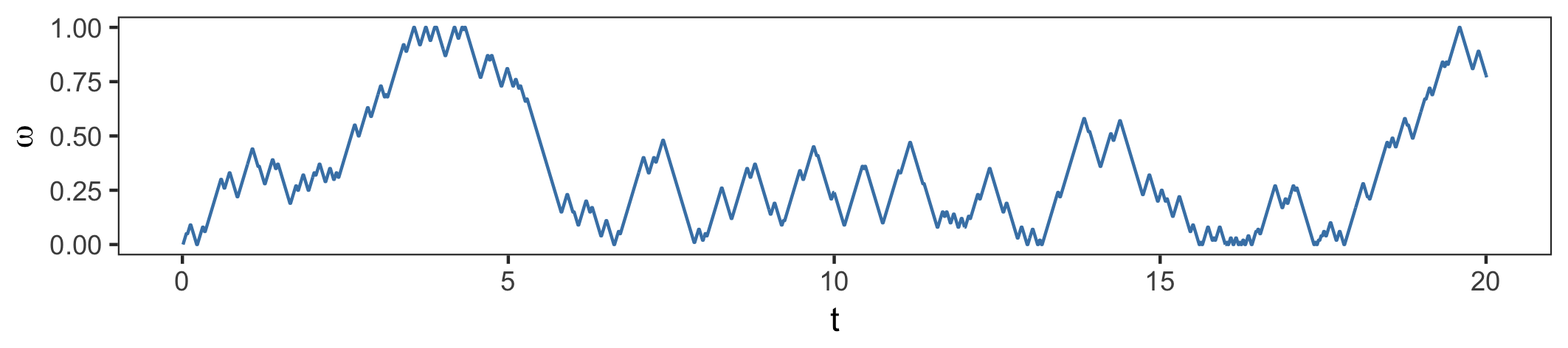}
	\end{center}
	\caption{Sample trajectories of $W(t)$ where $Z(t)=(W(t),\eps(t))$ under an optimal schedule generated by $G=F^{-1}$ for $F(\beta)=\Lambda(\beta)/\Lambda$ for $\Lambda=1$ (top) and $\Lambda=10$ (bottom) respectively.}
	\label{fig_trajectory_Z}
\end{figure}

To establish scaling limits for $(I_n,\eps_n)$, it will be convenient to work in a continuous time setting. To do this, we suppose the times that PT iterations occur are distributed according to a Poisson process $\{M(\cdot)\}$ with mean $\mu_N$. The number of PT iterations that occur by time $t\geq 0$ satisfies $M(t)\sim \mathrm{Poisson}(\mu_N t)$. We define the \emph{scaled index process} by $Z^N(t)=(W^N(t),\eps^N(t))$ where $W^N(t)=I_{M(t)}/N$ and $\eps^N(t)=\eps_{M(t)}$ (example realizations in  Figure~\ref{fig_trajectory_Z}). 

For convenience, we will denote $\beta_w=G(w)$ and use $z=(w,\eps)\in [0,1]\times\{-1,1\}$ to be a \emph{scaled index}. Define $C(\R_+,\mathcal{S})$ and $D(\R_+,\mathcal{S})$ to be set of functions $f:\R_+\to \mathcal{S}$ that are continuous and c\`adl\`ag respectively.

The process $Z^N\in D(\R_+,[0,1]\times\{-1,1\})$ takes values on the discrete set $\mathcal{P}_N^{\text{uniform}}\times\{-1,1\}$ and is only well-defined when $Z^N(0)=z_0\in \mathcal{P}_N^{\text{uniform}}\times\{-1,1\}$. To establish convergence, it is useful to extend it to a process $Z^N$ which can be initialized at any $z_0\in[0,1]\times\{-1,1\}$. Suppose $Z^N(0)=z_0\in [0,1]\times\{-1,1\}$ and let $T_1,T_2,\dots$ be the iteration times generated by the Poisson process $M$. We construct $Z^N(t)$ as follows: define $Z^N(t)=z_n$ for $t\in[T_{n},T_{n+1})$ and update $z_{n+1}|z_n$ via a transition kernel dependent on the communication scheme. We determine this transition kernel mirroring the construction from Section \ref{sec_domination}.

Before doing this, it will be useful to define the backward and forward shift operators $\Phi^N_-,\Phi^N_+:[0,1]\to[0,1]$ by,
\begin{align}\label{Phi_minus_def}
\Phi^N_- (w)=
\begin{cases}
w-\frac{1}{N} & w\in \left[\frac{1}{N},1\right],\\
\frac{1}{N}-w & w\in \left[0,\frac{1}{N}\right),
\end{cases}
\end{align}
and similarly,
\begin{align}\label{Phi_plus_def}
\Phi^N_+ (w)=
\begin{cases}
w+\frac{1}{N} & w\in \left[0,1-\frac{1}{N}\right],\\
1-\left(\frac{1}{N}-(1-w)\right) & w\in \left(1-\frac{1}{N},1\right].
\end{cases}
\end{align}
Intuitively $\Phi^N_\eps(w)$ represents the location in $[0,1]$ after $w$ moves a distance $\frac{1}{N}$ in the direction of $\eps$ with a reflection at $0$ and $1$.

\subsection{Scaled index process for reversible PT}
Under the reversible communication scheme, if $z_n=(w_n,\eps_n)\in \{0,1/N,\dots,1\}\times\{-1,1\}$, then we have $w_{n+1}=\Phi^N_{\eps_n}(w_n)$ if a swap successfully occurred and $w_{n+1}=w_n$ otherwise. In both cases, $\eps_{n+1}\sim \mathrm{Unif}\{-1,+1\}$. Since $\Phi^N_{\eps}(w)$ is not only well-defined for $w\in\mathcal{P}_N^{\text{uniform}}$ but for $w\in[0,1]$, we naturally extend this construction to any $w\in[0,1]$.

Formally, we generate $(w_{n+1},\eps_{n+1})$ in two steps. In the first step we simulate, 
\begin{align}
    w_{n+1}|w_n,\eps_n\sim \begin{cases}
\Phi^N_{\eps_n}(w_n) 	& \text{ with probability }	s(\beta_{w_n},\beta_{\Phi^N_{\eps_n}(w_n)}),\\
w_n 					& \text{ otherwise}.\\
\end{cases}
\end{align}
In the second step we simulate $\eps_{n+1}\sim \mathrm{Unif}\{-1,+1\}$. This defines a continuous time Markov pure jump process $W^N\in D(\R_+,[0,1])$ with jumps occurring according to an exponential of rate $\mu_N$ and is well defined when initialized at any state $w_0\in[0,1]$.  

From Theorem 19.2 in \cite{kallenberg2002foundations}, the infinitesimal generator for $W^N$ with reversible communication is
\begin{align}\label{generator_WN_def}
\mathcal{L}_{W^N}f(w)
&=\frac{\mu_N}{2}\sum_{\eps \in\{\pm 1\}}\left(f(\Phi_\eps^N(w))-f(w)\right)s(\beta_w,\beta_{\Phi^N_\eps(w)}),
\end{align}
where the domain $\mathcal{D}(\mathcal{L}_{W^N})$ is given by the set of functions such that $\mathcal{L}_{W^N}f$ is continuous. Since $\Phi^N_+, \Phi^N_-$ are continuous, we have $\mathcal{D}(\mathcal{L}_{W^N})=C([0,1])$.

\subsection{Scaled index process for non-reversible PT}
Before defining the transition kernel for the scaled index process under non-reversible communication, it will be convenient to define the propagation function $\Phi^N:[0,1]\times\{-1,1\}\to [0,1]\times\{-1,1\}$ for $z=(w,\eps)$,
\begin{align}
    \Phi^N(z)=
    \begin{cases}
    (\Phi^N_\eps(w),\eps) 	& \text{ if } \Phi^N_\eps(w)=w+\frac{\eps}{N},\\
    (\Phi^N_\eps(w),-\eps)	& \text{ otherwise},
    \end{cases}
\end{align}
and similarly the rejection function $R:[0,1]\times\{-1,1\}\to [0,1]\times\{-1,1\}$,
\begin{align}
    R(z)=(w,-\eps).
\end{align}

Under the non-reversible scheme, if $z_n=(w_n,\eps_n)\in \mathcal{P}_N^{\text{uniform}}\times\{-1,1\}$, then we have $z_{n+1}=\Phi^N(z_n)$ when a swap is accepted and $z_{n+1}=R(z_n)$ otherwise. Since $\Phi^N(z)$ and $R(z)$ are well-defined for all of $z\in [0,1]\times\{-1,1\}$, we naturally extend this construction to any $z\in[0,1]\times\{-1,1\}$. 

Formally, we generate $z_{n+1}$ according to the transition kernel,
\begin{align}
    z_{n+1}|z_n\sim \begin{cases}
    \Phi^N(z_n) 	& \text{ with probability } s(\beta_{w_n},\beta_{\Phi^N_{\eps_n}(w_n)}),\\
    R(z_n) 					& \text{ otherwise}.\\
\end{cases}
\end{align}
This defines a continuous time Markov pure jump process $Z^N\in D(\R_+,[0,1]\times\{-1,1\})$ with jumps occurring at an exponential of rate $\mu_N$. This process is well defined when initialized at any $z_0\in[0,1]\times\{-1,1\}$. 

Analogously to the reversible case, under non-reversible communication, the infinitesimal generator for $Z^N$ is
\begin{align}\label{generator_ZN_def}
\mathcal{L}_{Z^N}f(z)
&=\mu_N\left(f(\Phi^N(z))-f(z)\right)s(\beta_{w},\beta_{{\Phi}^N_\eps(w)})
+\mu_N\left(f(R(z))-f(z)\right)r(\beta_{w},\beta_{\Phi^N_\eps(w)}),
\end{align}
where $z=(w,\eps)$ and $\mathcal{D}(\mathcal{L}_{Z^N})$ is given by the set of functions $f$ such that $\mathcal{L}_{Z^N}f$ is continuous. Since $\Phi^N$ has discontinuities at $(\frac{1}{N},-1)$ and $(1-\frac{1}{N},1)$, we can verify that $f\in\mathcal{D}(\mathcal{L}_{Z^N})$ if and only if $f(w_0,-1)=f(w_0,1)$ for $w_0\in\{0,1\}$.

\subsection{Weak limits for scaled index processes}

Define $W\in C(\R_+,[0,1])$ to be the diffusion on $[0,1]$ with generator
\begin{equation}\label{generator_W_def}
\mathcal{L}_Wf(w)=\frac{1}{2}\frac{\mathrm{d}^2f}{\mathrm{d}w^2},
\end{equation}
where the domain $\mathcal{D}(\mathcal{L}_W)$ is the set of functions $f\in C^2([0,1])$ such that $f'(0)=f'(1)=0$. $W$ is a Brownian motion on $[0,1]$ with reflective boundary conditions admitting the uniform distribution $\mathrm{Unif}([0,1])$ as stationary distribution. 

Define $Z\in C(\R_+,[0,1]\times\{-1,1\})$ to be the PDMP on $[0,1]\times\{-1,1\}$ 
given by $Z(t)=(W(t),\eps(t))$ where $W(t)$ moves in $[0,1]$ with velocity $\eps(t)$ and the sign of $\eps(t)$ is reversed at the arrivals times of a non-homogeneous Poisson process of rate $\lambda(G(W(t)))G'(W(t))$ or when $W(t)$ reaches the boundary $\{(0,-1),(1,+1)\}$; see \cite{bierkens2018piecewise} for a discussion of PDMP on restricted domains. 
The infinitesimal generator of $Z$ is given by
\begin{equation}\label{generator_Z_def}
\mathcal{L}_{Z}f(z)=\eps\frac{\partial f}{\partial w}(z)+\lambda(\beta_w)G'(w)\left(f(R(z))-f(z)\right),
\end{equation}
for any $f\in \mathcal{D}(\mathcal{L}_Z)$, the set of functions $f\in C^1([0,1]\times\{-1,1\})$ such that $f(w_0,-1)=f(w_0,1)$ and $\frac{\partial f}{\partial w}(w_0,-1)=-\frac{\partial f}{\partial w}(w_0,1)$ for $w_0\in\{0,1\}$.

\subsection{Proof of scaling limit for reversible PT} \label{sec_reversible_proof_lemma}
\newcommand{\rd}{\mathrm{d}}
We will prove Theorem \ref{theorem_weak_limit_main}(a) by using Theorem 17.25 from \cite{kallenberg2002foundations}.
\begin{theorem}[Trotter, Sova, Kurtz, Mackevi\u{c}ius]\label{theorem_feller_convergence}
Let $X, X^1, X^2,\dots$ be Feller processes defined on a state space $S$ with generators $\mathcal{L}, \mathcal{L}_1, \mathcal{L}_2,\dots$ respectively. If $D$ is a core for $\mathcal{L}$, then the following statements are equivalent:
\begin{enumerate}
\item If $f\in D$, there exist $f_N\in\mathcal{D}(\mathcal{L}_N)$ such that $\|f_N-f\|_\infty\to 0$ and $\|\mathcal{L}_Nf_N-\mathcal{L}f\|_\infty\to 0$ as $N\to \infty$.
\item If $X^N(0)$ converges weakly to $X(0)$ in $S$, then $X^N$ converges weakly to $X$ in $D(\R_+,S)$.
\end{enumerate}
\end{theorem}

We will be applying Theorem~\ref{theorem_feller_convergence} with $\mathcal{L}= \mathcal{L}_W$ defined as $\mathcal{L}_W f = \tfrac{1}{2} f''$ for $f\in \mathcal{D}(\mathcal{L}_W)$ where
\begin{equation}
\mathcal{D}(\mathcal{L}_W) := \left\{f\in C^2\left([0,1]\right): f'(0)=f'(1)=0 \right\},    
\end{equation}
and $\mathcal{L}_N=\mathcal{L}_{W^N}$ defined in \eqref{generator_WN_def}, which we recall here for the reader's sake
\begin{equation}
\mathcal{L}_{W^N}f (w) = \frac{N^2}{2}\sum_{\eps \in\{\pm 1\}}\left(f(\Phi_\eps^N(w))-f(w)\right)s(\beta_w,\beta_{\Phi^N_\eps(w)}), \quad w \in[0,1],
\end{equation}
with $\Phi_{\pm}^N(w)$ defined in \eqref{Phi_minus_def}, \eqref{Phi_plus_def} and $\beta_w=G(w)$. Recall from the discussion just before \eqref{generator_WN_def} that $\mathcal{L}_{W^N}$ defines a Feller semigroup. 

First notice that in \cite{kallenberg2002foundations}, the transition semi-group and generator of a Feller process taking values in a metric space $S$ are defined on $C_0(S)$, the space of functions vanishing at infinity. Equivalently $f\in C_0(S)$ if and only for any $\delta>0$ there exists a compact set $K\subset S$ such that for $x\notin K$, $|f(x)|< \delta$. In our case since $S=[0,1]$ is compact $C_0(S)=C(S)$, which justifies the definition of the generator $\mathcal{L}_W$ given above. 

\subsubsection{The Feller property of $\mathcal{L}_W$.}
Similarly $\mathcal{L}_W$ can be seen to define a Feller semigroup on $C([0,1])$ by the Hille-Yosida theorem; see \cite[Theorem~19.11]{kallenberg2002foundations}. 

Indeed the first condition is satisfied since any function $f\in C([0,1])$ can be uniformly approximated within $\epsilon>0$ by a polynomial $p_\epsilon$, that is a smooth function, by the Stone-Weirstrass theorem. We can further uniformly approximate $p_\epsilon$ within $\epsilon$  by a $C^2$ function $\hat{p}_\epsilon$ with vanishing derivatives at the endpoints. For example one can let, for a $\delta$ to be chosen later, $\hat{p}_\epsilon(x) = p_\epsilon(x)$ for $x\in (\delta, 1-\delta)$ and for $x\leq\delta$ set $\hat{p}_\epsilon(x)=\int_0^x \rho_\delta (y) p_\epsilon'(y) \rd y+c$, where $\rho_\delta$ is a smooth, increasing transition function such that $\rho_\delta(x)=0$ for $x<0$, $\rho_\delta(x)=1$ for $x>\delta$, for example let $\rho_\delta=\rho(x/\delta)$, $\rho(x)=g(x)/(g(x)+g(1-x))$ and $g(x)=\exp(-1/x)\mathbf{1}_{\{x>0\}}$. We choose $c$ so that $\hat{p}_\epsilon(x)$ is continuous at $\delta$. A similar construction can be used for the right-endpoint. One can then check that indeed $\hat{p}_\epsilon \in C^2([0,1])$, $\hat{p}_\epsilon'(0)=\hat{p}_\epsilon'(1)=0$ and that for $\delta$ small enough $\|\hat{p}_\epsilon - p_\epsilon\|_\infty < \epsilon$. 

The second condition of \cite[Theorem~19.11]{kallenberg2002foundations} requires that for some $\mu>0$, the set $(\mu-\mathcal{L}_W)(\mathcal{D}\big(\mathcal{L}_W)\big)$ is dense in $C([0, 1])$. Let $g\in C([0,1])$ be given. We apply
\cite[Corollary~2.2]{saranen1988solution}, with $f(t,y)=2\mu y- 2 g$,  which is clearly square integrable in $t$ and $2\mu$-Lipschitz in $y$. Then \cite[Corollary~2.2]{saranen1988solution} implies that 
 for small enough $\mu>0$ the two-point Neumann-boundary value problem
\begin{align*}
    \mu u - \frac{1}{2}u'' &= g\\
    u'(0) = u'(1)&=0
\end{align*}
admits a solution in the Sobolev space $H^2([0,1])$ of functions with square integrable first and second derivatives. This already implies that $u\in C^1([0,1])$, whereas the continuity of $g$ and of $u$ a priori implies the continuity of $u''$ since $u''=2\mu u -g$. Overall, for any $g\in C([0,1])$ we can find $u\in \mathcal{D}(\mathcal{L}_W)$ such that $g=(\mu-\mathcal{L}_W)g$ establishing the second condition of \cite[Theorem~19.11]{kallenberg2002foundations}.  

The third condition of \cite[Theorem~19.11]{kallenberg2002foundations} is that $(\mathcal{L}_W, \mathcal{D}(\mathcal{L}_W))$ satisfies the \textit{positive maximum principle}, that is if for some $f\in \mathcal{D}(\mathcal{L}_W))$ and $x_0\in[0,1]$ we have $f(x_0)\geq f(x)\vee 0$ for all $x\in [0,1]$, then $f''(x_0)\leq 0$. Suppose first that the maximum is attained at an interior point $x_0\in (0,1)$; since $f\in C^2([0,1])$, by definition of $\mathcal{D}(\mathcal{L}_W))$, $f''(x_0)\geq 0$. If on the other hand the positive maximum is attained at $x_0=0$, suppose that $f''(0)>\epsilon$ for all $x
\leq \epsilon$. Thus for $0<y<\epsilon$ small enough, since $f'(0)=0$ we have
$$f(y)=f(0)+\int_0^y f'(s) \rd s=f(0)+\int_0^y \int_0^s f''(r) \rd r \rd y 
\geq f(0)+ \frac{\epsilon}{2} y^2> f(0),$$
thus arriving at a contradiction. 

We have thus established that $\left(\mathcal{L}_W, \mathcal{D}\left(\mathcal{L}_W\right) \right)$ satisfies all conditions of \cite[Theorem~19.11]{kallenberg2002foundations} and therefore generates a Feller process. 
Now we can apply Theorem~\ref{theorem_feller_convergence} to prove Theorem~\ref{theorem_weak_limit_main}(a). We only need to check the first condition of Theorem~\ref{theorem_feller_convergence}. In this direction, first note that by definition 
$\Phi_{\pm}^N(w) = w \pm 1/N$ for $w\in [1/N, 1-1/N]$. 
Thus in this case using a Taylor expansion we have for $w^\ast_- \in [w-1/N, w]$ and $w^\ast_+\in [w,w+1/N]$ that for any $f\in\mathcal{D}\left(\mathcal{L}_W\right)$,
\begin{align}
f(\Phi_+^N(w)) - 2 f(w) + f(\Phi_-^N(w))
&= f(w) + \frac{1}{N} f'(w) + \frac{1}{2N^2}f''(w^\ast_+) \notag\\
&\qquad+ f(w) - \frac{1}{N} f'(w) + \frac{1}{2N^2}f''(w^\ast_-) - 2f(w)\\
&= \frac{1}{2N^2}\left( f''(w^\ast_+) + f''(w^\ast_-)\right).
\end{align}
Since $f''$ is uniformly continuous it follows that as $N\to\infty$,
\begin{equation}
    \sup_{w\in[0,1]} |f''(w^\ast_\pm) - f''(w)| =o(1), 
\end{equation}
and therefore for $w\in [1/N, 1-1/N]$ we have

\begin{align}
\sup_{w\in [0,1]} \Big| f(\Phi_+^N(w)) - 2 f(w) + f(\Phi_-^N(w)) - \frac{f''(w)}{N^2}\Big|
&= o\left( \frac{1}{N^2}\right).
\end{align}

When $w\in [0,1/N)$ or $w\in(1-1/N, 1]$ we instead perform a Taylor expansion around 0 or 1 respectively. We only do the calculation in the first case, the other case being similar. 
Let $w\in [0,1/N)$ in which case, since $f'(0)=0$,  for $w^\ast, w^\ast_- , w^\ast_+\in [0,2/N]$

\begin{align}
f(\Phi_+^N(w)) - 2 f(w) + f(\Phi_-^N(w))
&= f(0) +\Phi_+^N(w)f'(0) + \frac{1}{2} \left[ \Phi_+^N(w)\right]^2 f''(w^\ast_+) \notag\\
&\qquad + f(0) +\Phi_-^N(w) f'(0) +  \frac{1}{2} \left[ \Phi_-^N(w)\right]^2 f''(w^\ast_-)\notag\\
&\qquad - 2f(0)-2 f'(0)w - 2\frac{f''(w^\ast)}{2}w^2\\
&= \frac{f''(0)}{2} \left\{\left[ \Phi_+^N(w)\right]^2 + \left[ \Phi_-^N(w)\right]^2 - 2w^2\right\} + o\left(N^{-2} \right),
\end{align}
where the error term is uniform in $w$ and was obtained by combining the facts that $f''$ is uniformly continuous and that $|\Phi^N_{\pm}|, |w|/\leq 2/N$. 
Finally notice that since $w\in [0,1/N]$, then
\begin{equation}
\left[ \Phi_+^N(w)\right]^2 + \left[ \Phi_-^N(w)\right]^2 - 2w^2
= \left[ w+ \frac{1}{N}\right]^2 + \left[\frac{1}{N}-w\right]^2 - 2w^2 = \frac{2}{N^2}.
\end{equation}

Finally we will need the following weaker version of Theorem~\ref{theorem_rate}, whose proof is postponed to the end of this section.  
\begin{lemma}\label{lem:s_beta}
Suppose that $\pi_0(|V|), \pi(|V|)<\infty$. Then there exists a constant $C>0$ such that 
\begin{equation}
\sup_\beta \left|s(\beta, \beta+\delta)-1\right|\leq C \delta.    
\end{equation}
\end{lemma}
Using Lemma \ref{lem:s_beta}, we see that for some constant $C>0$ 
\begin{equation}
\sup_{w\in[0,1]} \Big|s\left(\beta_w, \beta_{\Phi_\pm^N(w)}\right)-1 \Big| \leq C \sup_{w}\left|G(w)-G\left(\Phi_\pm^{N}(w)\right) \right|
\leq \frac{C \|G'\|_\infty}{N},
\end{equation} 
and therefore 
\begin{align}
\mathcal{L}_N f(w)
&= \frac{N^2}{2}\sum_{\eps \in\{\pm 1\}}\left(f(\Phi_\eps^N(w))-f(w)\right)s(\beta_w,\beta_{\Phi^N_\eps(w)})\\
&= \frac{N^2}{2}\sum_{\eps \in\{\pm 1\}}\left(f(\Phi_\eps^N(w))-f(w)\right)\left[ 1 + o(N^{-1}) \right]\\
&= \frac{N^2}{2}\left( f(\Phi_+^N(w)) - 2 f(w) + f(\Phi_-^N(w)) \right) \left[ 1 + o(N^{-1}) \right]= \frac{N^2}{2} \frac{f''(w)}{N^2}  \left[ 1 + o(1) \right],
\end{align}
where the error term is uniform in $w$. Thus $\mathcal{L}_N f \to \mathcal{L}f$ uniformly. 
\begin{proof}[of Lemma~\ref{lem:s_beta}]
Using the bound $0\leq 1-\exp(-x)\leq x$ for $x\geq 0$ we have

\begin{align}\belowdisplayskip=-12pt
\lefteqn{\left| s(\beta, \beta') -1 \right|}\notag\\
&\leq \int\!\int \pi^{(\beta)}(\rd x)\pi^{(\beta')}(\rd y) \Big[1- \exp \big(-\max\big\{0, (\beta'-\beta)[V(x)-V(y)] \big\} \big) \Big]\\
&\leq |\beta'-\beta|\int\!\int \pi^{(\beta)}(\rd x)\pi^{(\beta')}(\rd y) \max\big\{0, [V(x)-V(y)] \big\} \big) \Big]\\
&\leq |\beta'-\beta|\int\!\int \pi^{(\beta)}(\rd x)\pi^{(\beta')}(\rd y) \big( |V(x)| + |V(y)| \big) \Big] \\
&\leq 2 |\beta'-\beta| \sup_{\beta}\pi^{(\beta)}(|V|).
\end{align}

\end{proof}

\subsection{Proof of scaling limit for non-reversible PT}\label{sec_lifted_proof_lemma}
We will prove Theorem \ref{theorem_weak_limit_main}(b) in a slightly round about way. We will define the auxiliary processes $\{U^N(\cdot)\}$, $\{U(\cdot)\}$ living on the 
unit circle $\mathds{S}^1:=\{z\in \mathds{C}: |z|=1\}$ along with a mapping $\phi: \mathds{S}^1\mapsto [0,1]\times \{\pm 1\}$ such that $Z^N = \phi(U^N)$ and $Z=\phi(U)$. 
We will first show that the law of $U^N$ converges weakly to $U$. 

Before defining the processes we point out that we will identify $\mathds{S}^1$ with $[0,2\pi)$ in the usual way by working in$\mod 2\pi$ arithmetic. Notice that in this way 
\begin{equation}
    C(\mathds{S}^1)  = \{f\in C([0,2\pi]): f(0)=f(2\pi)\}.
\end{equation}
The reason for working with these auxiliary processes is that we can now avoid working with PDMPs with boundaries, helping us to remove a layer of technicalities. 

For any $N$ we define $\Sigma^N:\mathds{S}^1\mapsto\mathds{S}^1$ through $\Sigma^N(\theta) = \theta + 2\pi/N$. 
Consider then a continuous-time process $U^N$ that 
jumps at the arrival times of a homogeneous Poisson process with rate $N$ according to the kernel 
\begin{equation}\label{eq:kernel_aux}
Q^N(\theta, \rd \theta') = 
s\left(\tilde{\beta}_\theta, \tilde{\beta}_{\Sigma^N(\theta)} \right)\delta_{\Sigma^N(\theta)} (\rd \theta') 
+ \left[1-s\left(\tilde{\beta}_\theta, \tilde{\beta}_{\Sigma^N(\theta)} \right) \right] \delta_{2\pi -\theta} (\rd \theta'),
\end{equation}
where 
\begin{equation}
\tilde{\beta}_\theta = 
\begin{cases} 
G\left(\frac{\theta}{\pi}\right), & \theta \in [0,\pi),\\
G\left(\frac{2\pi -\theta}{\pi}\right), & \theta \in [\pi,2\pi).
\end{cases}    
\end{equation}

Define the map 
\begin{equation}
\phi(\theta)
=\begin{cases}
\left( \frac{\theta}{\pi}, +1\right), & \theta \in [0,\pi),\\
\left( \frac{2\pi-\theta}{\pi},-1\right), & \theta \in [\pi,2\pi).
\end{cases}    
\end{equation}
Essentially we think of the circle as comprising of two copies of $[0,1]$ glued together at the end points. The top one is traversed in an increasing direction and the bottom one in a decreasing direction. When glued together and viewed as a circle these dynamics translate in a counter-clockwise rotation with occasional reflections with respect to  the $x$-axis at the time of events. With this picture in mind it should be clear that $\phi(U^N) = Z^N$. 

We also define the limiting process $U$ as follows. First let 
\begin{equation}
    \tilde{\lambda}(\theta)
=(\lambda\circ G)(\phi^1(\theta)) G' (\phi^1(\theta)),
\end{equation}
where $\phi^1(\theta)$ is the first coordinate of $\phi(\theta)$. Notice at this point that $\phi^1: \mathds{S}^1\mapsto [0,1]$ is continuous and satisfies $\phi^1(\theta)=\phi^1(-\theta)$ for any $\theta\in[0,2\pi)$, whence we obtain that $\tilde{\lambda}(-\theta)=\tilde{\lambda}(\theta)$. 
Given $U(0)=\theta$, let $T_1$ be a random variable such that 
\begin{equation}
    \mathbb{P}[T_1 \geq t] = \exp \left\{-\int_0^t \tilde{\lambda}(\theta+s) \rd s \right\},
\end{equation}
and define the process as $U(s)=\theta+s \mod 2\pi$ for all $s<T_1$  and set $U(T_1) = - U(T_1-) \mod 2\pi$. 
Iterating this procedure will define the $\mathds{S}^1$-valued PDMP $\{U(\cdot)\}$. 
We first need the next lemma. 

\begin{lemma}\label{lem:lifted_aux_properties}
Suppose $V$ is integrable with respect to $\pi_0$ and $\pi$. 
The process $U$ defined above is a Feller process, its infinitesimal generator is given by 
\begin{equation}
    \mathcal{L}_U f(\theta) = f'(\theta) + \tilde{\lambda}(\theta) \left[ f(2\pi-\theta) - f(\theta)\right],
\end{equation}
with domain 
\begin{equation}
    \mathcal{D}(\mathcal{L}_U) = \{f\in C^1([0,\pi]): f(0)=f(2\pi)\},
\end{equation}
and invariant measure $\rd \theta/2\pi$. 
\end{lemma}

\begin{proof}
First, note that since $\mathds{S}^1$ is compact $C_0(\mathds{S}^1)=C(\mathds{S}^1)$ and thus to study the Feller process we consider the semi-group $\{P_U^t\}_t$ defined by the process $U$ as acting on $C(\mathds{S}^1)$. 
To prove the Feller property we can thus use \cite[Theorem~27.6]{davis1993markov}. Since there is no boundary in the definition of $U$ the first assumption is automatically verified, $Qf(\theta)=f(-\theta)\in C\left(\mathds{S}^1\right)$ for any continuous $f$. We also know that the rate $\tilde{\lambda}$ is bounded
whereas by Proposition~\ref{prop_regularity} and the fact that $G\in C^1[0,1]$ we know that $\tilde{\lambda}$ is also continuous. Therefore the third condition of \cite[Theorem~27.6]{davis1993markov} holds and thus $U$ is Feller. 

The infinitesimal generator will be defined on 
$\mathcal{D}(\mathcal{L}_U)\subseteq C(\mathds{S}^1)$. 
The domain is defined as the class of functions $f\in C(\mathds{S}^1)$ such that 
\begin{equation}
g(\theta) = \lim_{h\to 0} \frac{1}{h} \left[ P_U^t f(\theta)-f(\theta)\right]\in C(\mathds{S}^1),
\end{equation}
where
the limit is uniform in $\theta$. However by \cite[Theorem~1.33]{bottcher2013levy}, we can also consider pointwise limits without enlarging the domain. Using the definition of $U$ we then have for $\theta\in[0,2\pi)$ that 
\begin{align}
\frac{1}{h}\mathbb{E}^\theta \left[ f(U_h) - f(\theta)\right]
&=\frac{1}{h}
\left[f \left( \theta + h\right)\mathbb{P}^\theta\left[ T_1\geq h\right]-f(\theta)\right]  + 
\frac{1}{h}\mathbb{E}^\theta\left[
f\left( U_h\right) \mathds{1}\left\{T_1 \leq h\right\}
\right].
\end{align}
Since for $x\geq 0$ we have
$|\exp(-x)-1+x| \leq C x^2$ for some constant $C>0$,
and using the uniform continuity of $\tilde{\lambda}$ we can see that 
\begin{align}
\left|\exp \left\{-\int_0^h \tilde{\lambda}(\theta+s) \rd s \right\}
- 1+ \tilde{\lambda}(\theta)h
\right|&=o(h),
\end{align}
uniformly in $\theta$, and thus 

\begin{align}
\frac{1}{h}
f \left( \theta + h\right)\mathbb{P}^\theta\left[ T_1\geq h\right] -f(\theta) 
&= \frac{1}{h}
\left[f \left( \theta + h\right)\left[1-\tilde{\lambda}(\theta)h + o(h)\right] - f(\theta)\right]\\
&=\frac{1}{h}\left[ f(\theta+h)-f(\theta)\right] - \tilde{\lambda}(\theta)f(\theta) + o(1).
\end{align}
In addition 

\begin{align}
\lefteqn{\frac{1}{h}\mathbb{E}^\theta
\left[f\left( U_h\right) \mathds{1}\left\{T_1 \leq h\right\}\right]}\notag\\
&=\frac{1}{h}\int_0^h \tilde{\lambda}(\theta+s)
\exp\left\{-\int_0^s \tilde{\lambda}(\theta+r)\rd r \right\}
 \rd s P_U^{h-s} Q f(\theta)\\
&\to \tilde{\lambda}(\theta)Qf(\theta),
\end{align}
for any $f\in C(\mathds{S}^1)$ by strong continuity of $\{P_U^t\}$ (Feller  property) and continuity of $\tilde{\lambda}$. Overall we thus have that $f\in \mathcal{D}(\mathcal{L}_U)$ if and only if 
\begin{align}
\frac{1}{h}\mathbb{E}^\theta \left[ f(U_h) - f(\theta)\right]
&= \frac{f(\theta+h)-f(\theta)}{h} + \tilde{\lambda}(\theta)\left[ Qf(\theta)-f(\theta)\right] +o(1)\\
&\to g(\theta) \in C(\mathds{S}^1), 
\end{align}
which is clearly equivalent to $f\in C^1(\mathds{S}^1)$. 

Finally to see that $\rd \theta/2\pi$ is invariant, having identified the domain we can easily check that for any $f\in C(\mathds{S}^1)$ we have
\begin{align}
\int \rd \theta
P_U^t f(\theta)
&= \int_{s=0}^t \int \rd \theta \mathcal{L}_U P_U^s f(\theta) \rd\theta
 \rd s.
\end{align} 
Since $f\in \mathcal{D}(\mathcal{L}_U)$ we have that $P_U^s g \in\mathcal{D}(\mathcal{L}_U)$. Since for any $g\in \mathcal{D}(\mathcal{L}_U)$ we have
\begin{align}
\int\rd \theta \mathcal{L}_U f(\theta)
&= \int_{\theta=0}^{2\pi} f'(\theta) \rd \theta
+ \int_{\theta=0}^{2\pi} \tilde{\lambda}(\theta) f(Q(\theta))\rd \theta - \int_{\theta=0}^{2\pi} \tilde{\lambda}(\theta) f(\theta)\rd \theta\\
&= f(2\pi)-f(0) + \int_{\theta=0}^{2\pi} \tilde{\lambda}(\theta) f(Q(\theta))\rd \theta.    
\end{align}

\end{proof}

\begin{proposition}\label{prop:aux_weak_conv}
Suppose $U^N(0)$ converges weakly to $U(0)$, then $U^N$ converges weakly to $U$ in $D(\R_+,[0,1])$.
\end{proposition}
\begin{proof}
We will once again use Theorem~\ref{theorem_feller_convergence}. The generator of $U_N$ is given by 

\begin{align}
\mathcal{L}_U^N f(\theta)
&= N \left[f(\theta+1/N)-f(\theta)\right]
s\left(\tilde{\beta}_\theta, \tilde{\beta}_{\Sigma^N(\theta)} \right)
+ N \left[f(-\theta)-f(\theta)\right]r\left(\tilde{\beta}_\theta, \tilde{\beta}_{\Sigma^N(\theta)} \right). 
\end{align}

We will consider the two terms separately. To this end notice that  by \eqref{estimate_rate_integral_r}, the boundedness of $\lambda$ and the fact that $G\in C^1[0,1]$ 
\begin{equation}
    \left|1-s\left(\tilde{\beta}_\theta, \tilde{\beta}_{\Sigma^N(\theta)} \right)\right| \leq \frac{C}{N},
\end{equation}
for some $C>0$. Thus, using the mean value theorem, for each $\theta\in[0,2\pi)$, there exists $g_N(\theta)\in [\theta, \theta+1/N]$ such that 
\begin{align}
N \left[f(\theta+1/N)-f(\theta)\right]
s\left(\tilde{\beta}_\theta, \tilde{\beta}_{\Sigma^N(\theta)} \right)
&= f'\left(g_N(\theta)\right)\left(1+ O(1/N)\right) = f'\left(\theta\right)(\left(1+ o(1)\right),
\end{align}
where the errors are uniformly bounded and to obtain the second equality above we have used the fact that $|g_N(\theta)-\theta|\leq 1/N$ and that $f'$ is uniformly continuous, being continuous on a compact set. 

Overall we can see that as $N\to \infty$
\begin{equation}
    \sup_\theta \left| N \left[f(\theta+1/N)-f(\theta)\right]
s\left(\tilde{\beta}_\theta, \tilde{\beta}_{\Sigma^N(\theta)} \right)
-f'(\theta)\right| \to 0.
\end{equation}
Next, using \eqref{estimate_rate_integral_r} we have that 
\begin{equation}
r\left(\tilde{\beta}_\theta, \tilde{\beta}_{\Sigma^N(\theta)} \right)
= \tilde{\lambda}(\theta)\frac{1}{N} + o(N^{-1}),
\end{equation}
where the error is uniform in $\theta$, 
whence we easily conclude that 
\begin{equation}
N \left[f(-\theta)-f(\theta)\right]r\left(\tilde{\beta}_\theta, \tilde{\beta}_{\Sigma^N(\theta)} \right) \to 
\tilde{\lambda}(\theta) \left[Q f(\theta)-f(\theta)\right],
\end{equation}
uniformly in $\theta$. 
\end{proof}

\subsubsection{Proof of Theorem~\ref{theorem_weak_limit_main}(b)}
Now we are ready to prove the main result of this section. Notice that $Z^N (\cdot)=\phi\left(U^N (\cdot)\right)$ and $Z(\cdot)=\phi\left( U(\cdot) \right)$. 

From Proposition~\ref{prop:aux_weak_conv} we know that the finite dimensional distributions of $U_N$ converge to those of $U$. If $\phi$ were continuous we could conclude using the continuous mapping theorem. Since it is not continuous at the points $\{0, \pi\}$, we will be using \cite[Theorem~2.7]{billingsley2013convergence}. We have to check that the law of the limiting process, that is the law of $\{U(\cdot)\}$ places zero mass on finite dimensional distributions that hit $\{0, \pi\}$, that is 
for $n\in \mathbb{N}$ and $0<t_1<\dots<t_n$ we want 
\begin{equation}
\mathbb{P}\left[ \text{$U(t_i)\in\{0,\pi\}$ for some $i\in \{1, \dots,n\}$}\right]=0,
\end{equation}
when $U(0)$ is initialized according to $\rd\theta/2\pi$. 
But the above follows from the fact that 
$\mathbb{P}[ U(t_i)\in\{0,\pi\}]=0$, by stationarity when $U(0)$ is initialised uniformly on $\mathds{S}^1$. 

Relative compactness of $\{Z_N(\cdot)\}_N$ can be easily seen to follow from the compact containment condition \cite[Remark~3.7.3]{ethier2009markov}. This combined with convergence of the finite dimensional distributions of $Z^N$ to those of $Z$ concludes the proof. 

\section{Distributed implementation of PT}\label{app:distributed}
	For simplicity, in Section~\ref{sec_setup} we have shown pseudo-code for PT with a parallel computing context in mind rather than a distributed computing context. The key difference is that in the latter case, it is advantageous to swap annealing parameters instead of states to minimize network communication. For completeness, we show the pseudo-code for the distributed implementation in this section.

\begin{algorithm}[ht]
	\caption{Distributed DEO}\label{alg_deo_distributed} 
	\begin{algorithmic}[1]
		\State $\states\gets \states_0$ \Comment{Initialize chain} 
		\State $\textbf{i} \gets (0, 1, 2, \dots, N)$ \Comment{For a machine $j$, $\textbf{i}^j$ gives the annealing parameter index for that machine. Initialized with the identity map.}
		\State $\textbf{j} \gets (0, 1, 2, \dots, N)$ \Comment{For an annealing parameter index $i$, $\textbf{j}^i$ gives the machine processing the corresponding chain. This is the inverse of the above mapping, hence it is also initialized with the identity map.}
		\For{$n$ {\bf in} 1, 2, \dots, $\nscan$}
		\If{$n$ is even} \Comment{Non-reversibility inducing alternation}
		\State $P \gets \{i: 0 \le i < N, i\text{ is even} \}$ \Comment{Even subset of $\{0,\dots,N-1\}$}
		\Else
		\State $P \gets \{i: 0 \le i < N, i\text{ is odd} \}$ \Comment{Odd subset of $\{0,\dots,N-1\}$}
		\EndIf
		\For{$j$ \textbf{in} $0,\dots,N$}  \Comment{Distributed local exploration phase}
		\State $\beta \gets \beta_{\textbf{i}^{j}}$ \Comment{Fetch current annealing parameter for machine $j$}
		\State $x^j \sim K^{(\beta)}(x^j_{n-1},\cdot)$
		\EndFor
		\For{$i$ \textbf{in} $0,\dots,N-1$} \Comment{Distributed communication phase}
		\State $\alpha \gets \alpha^{(i,i+1)}$  \Comment{Equation (\ref{accept_ratio}).}
		\State $A \sim \Bern(\alpha)$
		\If{$i\in P$ \textbf{and} $A = 1$}
		\State $(\textbf{j}^i, \textbf{j}^{i+1}) \gets (\textbf{j}^{i+1}, \textbf{j}^{i})$ 
		\EndIf
		\EndFor
		\For{$i$ \textbf{in} $0,\dots,N$} \Comment{Recompute $\textbf{i}$ from $\textbf{j}$ such that $\textbf{j}^i = j \Longleftrightarrow \textbf{i}^j = i$}
			\State $j \gets \textbf{j}^i$
			\State $\textbf{i}^j \gets i$
		\EndFor
		\State $\states_n\gets \states$ 
		\EndFor
	\end{algorithmic}
\end{algorithm}

\section{Experiment supplements}\label{app_experiments}

\begin{figure}
	\centering
	\includegraphics[width=0.32\linewidth]{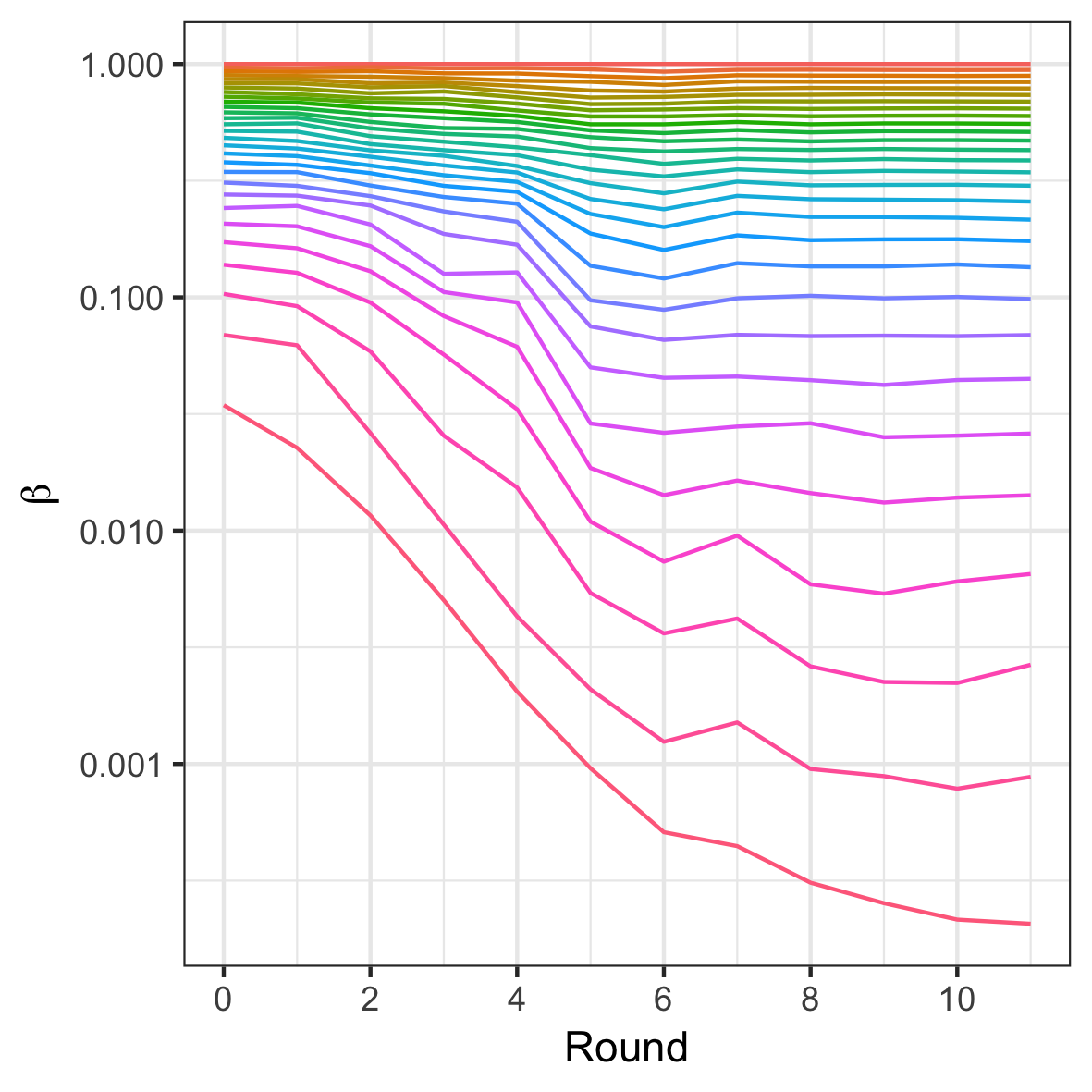}
	\includegraphics[width=0.32\linewidth]{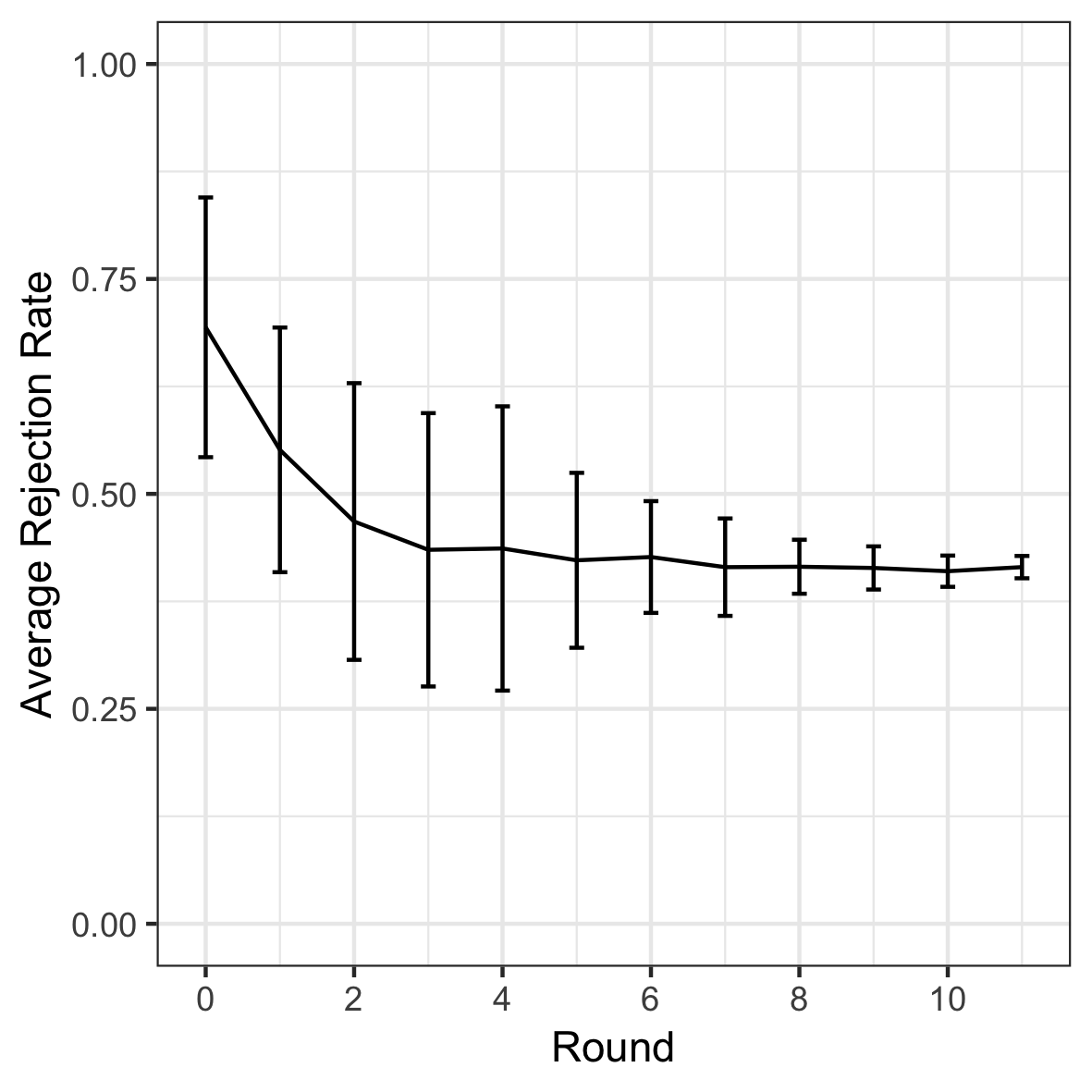}
	\includegraphics[width=0.32\linewidth]{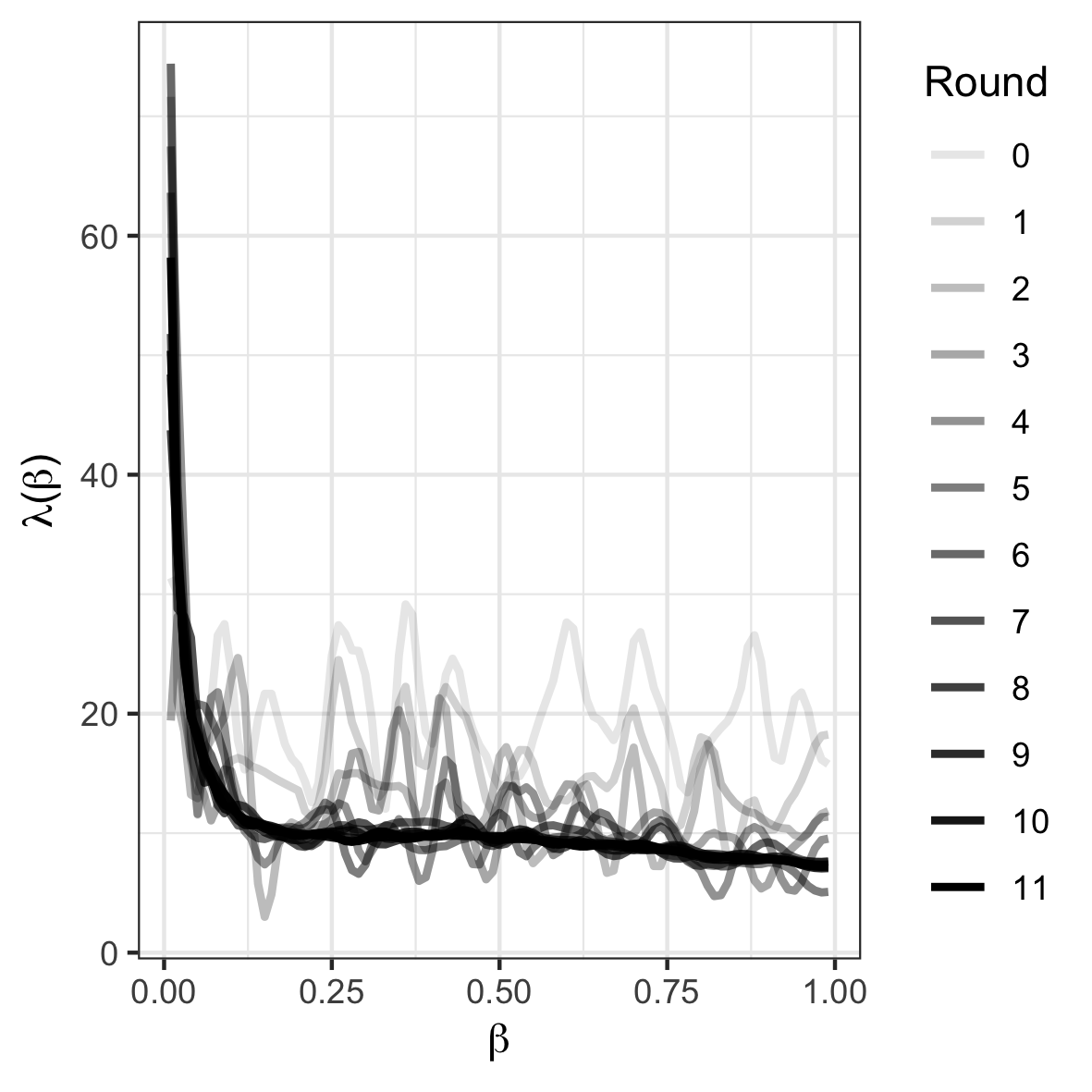}
	\caption{A demonstration of the tuning phase in Algorithm \ref{alg_adaptive} ran on a hierarchical Bayesian model applied to the historical failure rates of $5\,667$ launches for 367 types of rockets (a 369 dimensional problem, see Appendix \ref{description-models-data} for details). We use $\bar N=30$ cores and $11$ schedule optimization rounds, the last one consisting of $\ntune=5\,000$ and $\hat{\Lambda}=12.03$. (Left) Progression of the  annealing schedule (colours index parallel chains, y-axis, the values $\beta_k$ for each schedule optimization round, in log scale). (Centre) Progression of the sample mean and standard deviation of empirical rejection probabilities $\{\hat{r}^{(i-1,i)}\}_{i=1}^N$. The mean stabilizes quickly, and as the schedule optimization rounds increase, the rejection probabilities converge to the mean as desired. (Right) Progression of the estimated $\hat \lambda(\beta)$ as a function of the schedule optimization round.} 
	\label{fig_adapt_example} 
\end{figure}

\subsection{Reproducibility.} To make our NRPT method easy to use we implemented it as an inference engine in the open source probabilistic programming language (PPL) Blang \url{https://github.com/UBC-Stat-ML/blangSDK}. A full description of the models used in the paper are available at \url{https://github.com/UBC-Stat-ML/blangDemos}, see in particular \url{https://github.com/UBC-Stat-ML/blangDemos/blob/master/src/main/resources/demos/models.csv} for a list of command line options and data paths used for each model.
All methods use the same local exploration kernels, namely slice sampling with exponential doubling followed by shrinking \cite{neal_slice_2003}. Scripts documenting replication of our experiments are available at \url{https://github.com/UBC-Stat-ML/ptbenchmark}.

\subsection{Multi-core implementation.} We use lightweight threads \cite{friesen_java_2015} to parallelize both the local exploration and communication phases, as shown in Algorithm \ref{alg_deo}. We use the algorithm of \cite{leiserson_deterministic_2012} as implemented in \cite{jdk_splittable} to allow each PT chain to have its own random stream. This technique avoids any blocking across threads and hence makes the inner loop of our algorithm embarrassingly parallel in $N$. Moreover, the method of \cite{leiserson_deterministic_2012} combined with the fact that we fix random seeds means that the numerical value output by the algorithm is not affected by the number of threads used. Increasing the number of threads simply makes the algorithm run faster. In all experiments unless noted otherwise we use the maximum number of threads available in the host machine, by default an Intel i5 2.7 GHz (which supports 8 threads via hyper-threading) except for Section~\ref{sec_comparison} where we use an Amazon EC2 instance of type \texttt{c4.8xlarge}, which is backed by a 2.9 GHz Intel Xeon E5-2666 v3 Processor (20 threads).

\subsection{Stochastic optimization methods.} \label{app_experiments_comparison}
All baseline methods are implemented in Blang (\url{https://github.com/UBC-Stat-ML/blangSDK}), the same probabilistic programming language used to implement our method. The code for the baseline adaption methods are available at \url{https://github.com/UBC-Stat-ML/blangDemos}. All methods therefore run on the Java Virtual Machine, so their wall clock running times are all comparable. 

Both \cite{atchade_towards_2011} and \cite{miasojedow2013adaptive} are based on reversible PT together with two different flavours of stochastic optimization to adaptively select the annealing schedule. In \cite{atchade_towards_2011}, the chains are added one by one, each chain targeting a swap acceptance rate of $23\%$ from the previous one. In \cite{miasojedow2013adaptive}, this scheme is modified in two ways: first, all annealing parameters are optimized simultaneously, and second, a different update for performing the stochastic optimization is proposed. To optimize all chains simultaneously, the authors assume that both the number of chains and the equi-acceptance probability are specified. Since this information is not provided to the other methods, in order to perform a fair comparison, for the method we label as ``Miasojedow, Moulines, Vihola'' we implemented a method which adds the chain one at the time while targeting the swap acceptance rate of $23\%$ but based on the improved stochastic optimization update of \cite{miasojedow2013adaptive}. Specifically, both \cite{atchade_towards_2011} and \cite{miasojedow2013adaptive} rely on updates of the form $\rho_{n+1} = \rho_n + \gamma_n (\alpha_{n+1} - 0.23)$ where $\gamma_n$ is an update schedule and $\rho_n$ is a re-parameterization of difference in annealing parameter from the previous chain $\beta$ to the one being added $\beta'$. The work of \cite{atchade_towards_2011} uses the update $\beta'_n = \beta (1 + \exp(\rho_n))^{-1}$, whereas the work of \cite{miasojedow2013adaptive} specifies the explicit parameterization used for $\rho$, namely $\rho = \log(\beta'^{-1} - \beta^{-1})$, from which the update becomes $\beta'_n = \beta (1 + \beta \exp(\rho_n))^{-1}$. Moreover, while \cite{atchade_towards_2011} use $\gamma_n = (n+1)^{-1}$, \cite{miasojedow2013adaptive} suggest to use $\gamma_n = (n+1)^{-0.6}$. The experiments confirm that the latter algorithm is more stable.

\subsection{Description of models and datasets}\label{description-models-data}

In Section~\ref{sec_comparison} we benchmark the methods on the following four models. First, a hierarchical model applied to a dataset of rocket launch failure/success indicator variables \cite{McDowell_2019}. We organized the data by types of launcher, obtaining $5,667$ launches for $367$ types of rockets (processed data available at \url{https://github.com/UBC-Stat-ML/blangDemos/blob/master/data/failure_counts.csv}). Each type is associated with a Beta-distributed parameter with parameters tied across rocket types, with the likelihood given by a Binomial distribution (full model specification available at \url{https://github.com/UBC-Stat-ML/blangDemos/blob/master/src/main/java/hier/HierarchicalRockets.bl}). The second model is a Spike-and-Slab variable selection model applied to the RMS Titanic Passenger Manifest dataset \cite{Hind_2019}. The preprocessed data is available at \url{https://github.com/UBC-Stat-ML/blangDemos/tree/master/data/titanic}. The data consist in binary classification indicators for the survival of each individual passenger as well as covariates such as age, fare paid, etc. We used a Spike-and-Slab prior with a point mass at zero and a Student-t continuous component (full model specification available at \url{https://github.com/UBC-Stat-ML/blangDemos/blob/master/src/main/java/glms/SpikeSlabClassification.bl}). Third, we used the Ising model from Section~\ref{sec_ising}. Finally, we also used an end-point conditioned Wright-Fisher stochastic differential equation (see, e.g., \cite{tataru_statistical_2017}). For this last model we used synthetic data generated by the model. The specification of this last model is available at \url{https://github.com/UBC-Stat-ML/blangSDK/blob/master/src/main/java/blang/validation/internals/fixtures/Diffusion.bl}. 

The model-specific command line options used for all four experiments is available at \url{https://github.com/UBC-Stat-ML/blangDemos/blob/master/src/main/resources/demos/models.csv}. The same file also documents the location of the probabilistic programming source code and precise command line arguments for the 16 models summarized in Figure~\ref{fig:all-models}.

\subsubsection{Gaussian model}\label{app_normal}

Suppose $\pi\sim N(0,\sigma^2\mathbb{I}_d)$, and $\pi_0\sim N(0,\sigma^2_0\mathbb{I}_d)$ with $\sigma_0>\sigma$. It can be shown that $\pi^{(\beta)}\sim N(0,\sigma^2_\beta\mathbb{I}_d)$ where $\sigma_\beta^{-2}=(1-\beta)\sigma_0^{-2}+\beta\sigma^{-2}$. Theorem 1 in \cite{predescu2004incomplete} implies the following closed form expressions for $\lambda(\beta)$ and $\Lambda(\beta)$
\begin{equation}\label{gaussian_local}
\lambda(\beta)=\frac{2^{1-d}(\sigma^{-2}-\sigma_0^{-2})}{B\left(\frac{d}{2},\frac{d}{2}\right)}\sigma_\beta^2,\qquad \Lambda(\beta) = \frac{2^{2-d}}{B\left(\frac{d}{2},\frac{d}{2}\right)}\log\left(\frac{\sigma_0}{\sigma_\beta}\right),
\end{equation}
where $B(a,b)$ is the Beta function. As $d\to\infty$, we have $ \Lambda\sim \sqrt{\frac{2d}{\pi}}\log\left(\frac{\sigma_0}{\sigma}\right)$, which is consistent with Proposition \ref{prop_high_dim}. 
We see from Figure \ref{fig_gaussian_communication} that the empirical approximation of $\hat\lambda,\hat\Lambda$ from Algorithm \ref{alg_adaptive} are consistent with the theoretical values from in \eqref{gaussian_local}.

Substituting \eqref{gaussian_local} into \eqref{beta_opt_constraint} we have $\beta^*_k$ satisfies
\begin{equation}\label{gaussian_opt_schedule}
\sigma_{\beta_k^*}=\sigma^{\frac{k}{N}}\sigma_0^{1-\frac{k}{N}}.
\end{equation}
This is the same spacing obtained (based on a different theoretical approach) in \cite{atchade_towards_2011} and \cite{predescu2004incomplete} for the Gaussian model. This combined with Proposition \ref{prop_multimodal} also justifies why the geometric schedule works for Gaussian mixture models with well-separated modes. 
See Figure~\ref{fig_trajectory} for realizations of the index process for the reversible and non-reversible PT algorithms. 

\begin{figure}[ht]
	\begin{center}
		\includegraphics[width = 0.450\linewidth]{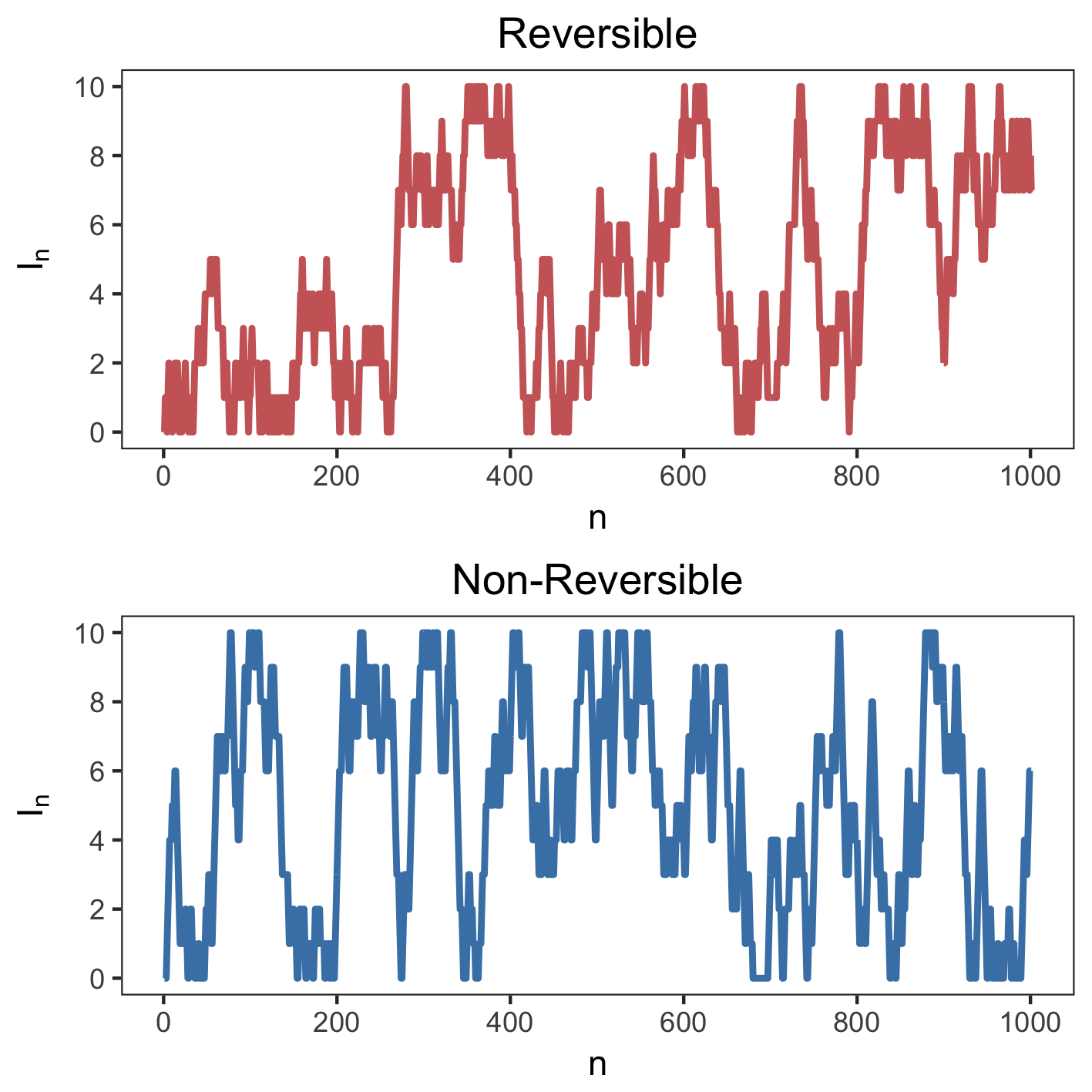}
		\includegraphics[width = 0.45\linewidth]{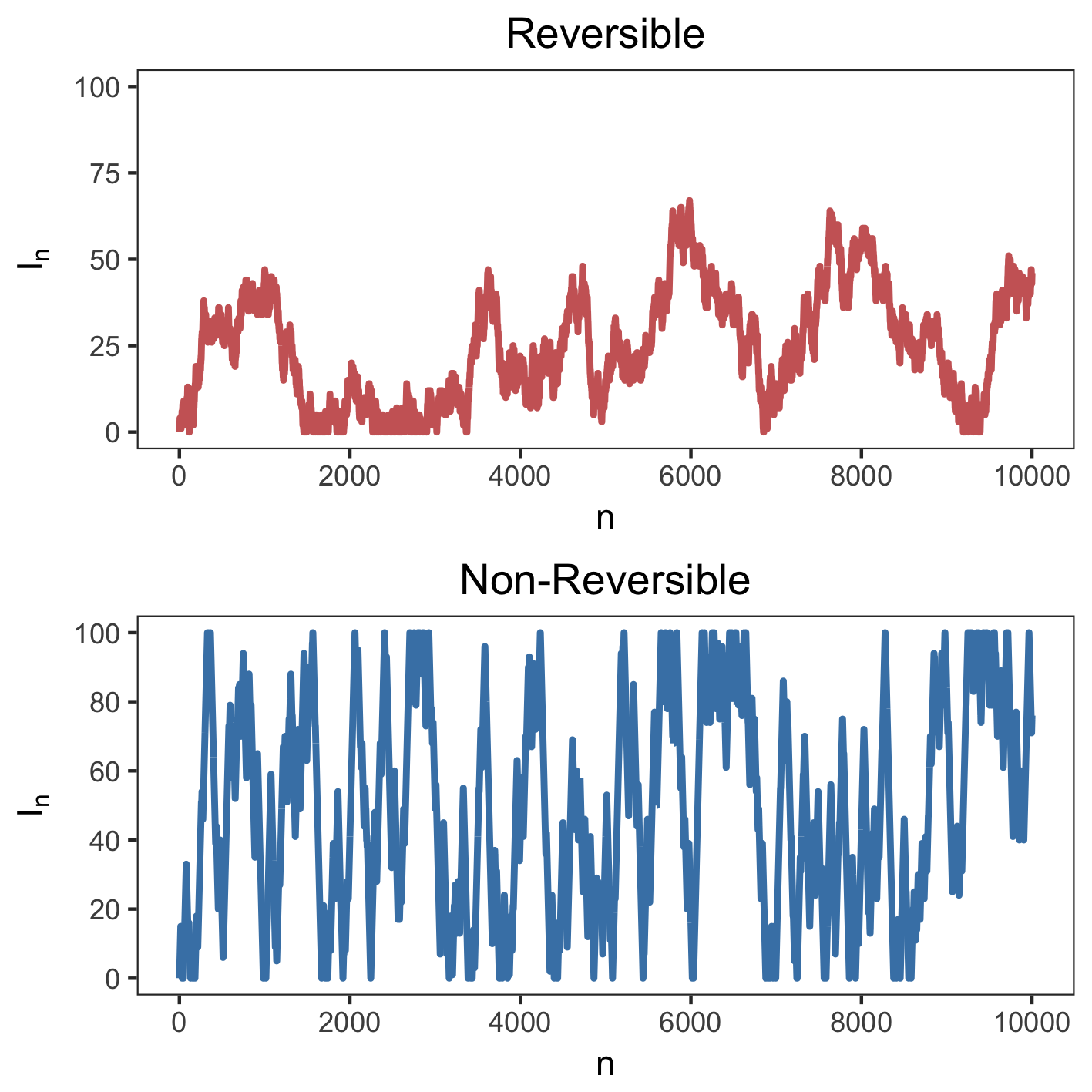}
	\end{center}
	\caption{Sample trajectories of the component $I^{0}_n$ of the index process for a Gaussian model with $\Lambda=5$  with schedule $\partition$ satisfying \eqref{gaussian_opt_schedule} for $N=10$ (left) and $N=100$ (right). The trajectories are run over $100N$ iterations for reversible (top) and non-reversible (bottom) PT.}
	\label{fig_trajectory}
\end{figure}

\begin{figure}
	\begin{center}
		\includegraphics[width=0.3\linewidth]{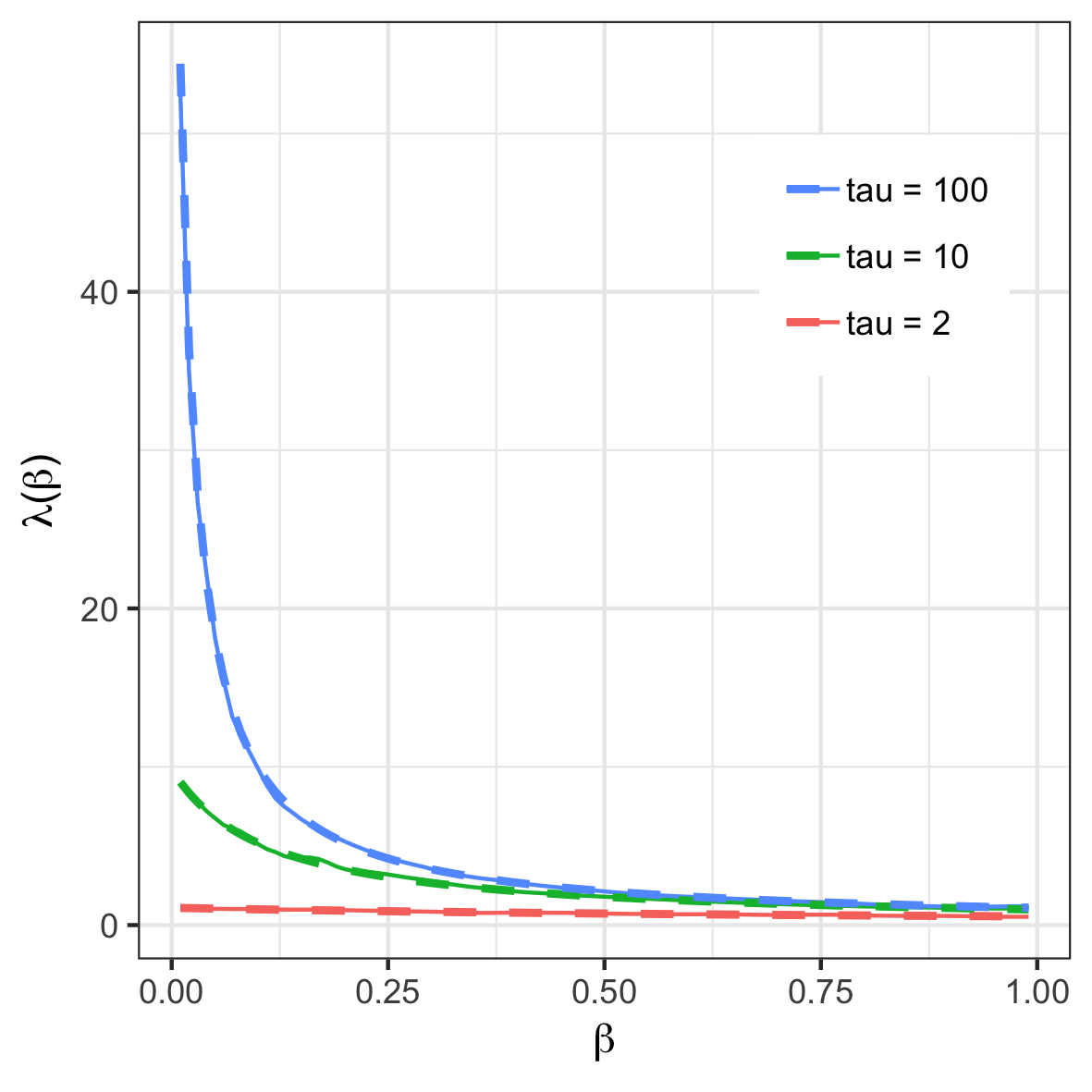}
		\hspace{.5in}
		\includegraphics[width=0.3\linewidth]{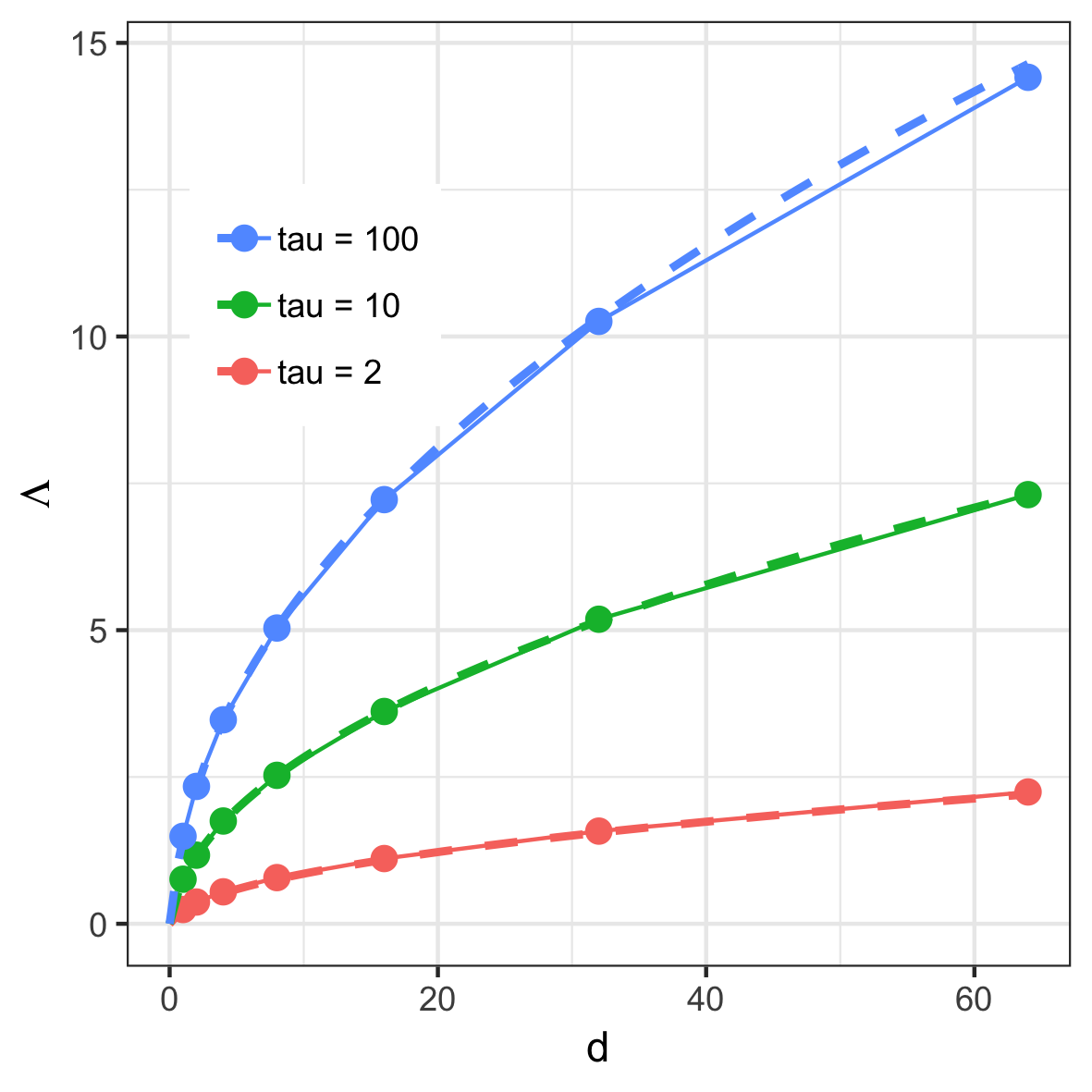}
		\caption{(Left) The local communication barrier for $d=8,\sigma=1/2,1/10,1/100$ and $\sigma_0=1$. (Right) The global communication barrier as a function of $d$ for $\tau=2,10,100$ and $\sigma_0=1$. The solid line is the approximation $\hat{\lambda}(\beta)$ (respectively $\hat{\Lambda}$), resulting from Algorithm \ref{alg_adaptive} ($N=60$ and $n=10\,000$ scans) and the dotted line is the analytic expression from \eqref{gaussian_local}.}
		\label{fig_gaussian_communication}
	\end{center}
\end{figure}

\subsubsection{Ising model}\label{sec_ising}

We now compute numerically $\lambda$ for the Ising model on a 2-dimensional lattice of size $M\times M$ with magnetic moment $\mu$. Note that the terminology ``Ising'' is technically incorrect since we are using a finite $M$ instead of an infinite lattice. The terminology ``autologistic'' would be more appropriate but we abuse the terminology ``Ising'' since it is widespread in the computational statistics literature. Similarly, when we talk about ``phase transition'' this should be interpreted as an approximate phase transition since $M$ is finite. 

Using the notation $x_i\sim x_j$ to indicate sites are nearest neighbours on the lattice, the target distribution is annealed by the inverse temperature $\beta$ and the tempered distributions are given by
 	\begin{align}
 	\pi^{(\beta)}(x)=\frac{1}{Z(\beta)} \exp\left(\beta\sum_{x_i\sim x_j} x_ix_j+\mu \sum_{i}x_i\right).
 	\end{align}
This is an $M^2$ dimensional model which undergoes a phase transition as $M\to\infty$ at some critical inverse-temperature $\beta_c$. When $\mu=0$ it is known that $\beta_c=\log(1+\sqrt{2})/2$ \cite{baxter2007exactly}. We consider experiments with $\mu=0$ and with $\mu = 0.1$, the latter denoted ``Ising with magnetic field'' and abbreviated ``magnetic'' in composite figures.  

We observe that $\lambda$ exhibits very different characteristics in this scenario compared to the Gaussian model: it is not monotonic and is maximized at the critical temperature. Consequently, the optimal annealing schedule is denser near the phase transition. We also note from Figure \ref{fig:ising_communication} that both $\lambda$ and $\Lambda$ increases roughly linearly with respect to $M$, similarly to the conclusion of Proposition \ref{prop_high_dim}, however here the conditions of Proposition \ref{prop_high_dim} do not apply. As $N$ increases, the round trip rate of reversible PT decays to 0 and non-reversible PT increase towards $(2+2\Lambda)^{-1}$ as seen in Figure \ref{fig:ising_rtr}. This is consistent with Theorem \ref{thm_efficiency_convergence}.

 \begin{figure}[ht]
 \begin{center}
 \includegraphics[height=0.3\linewidth]{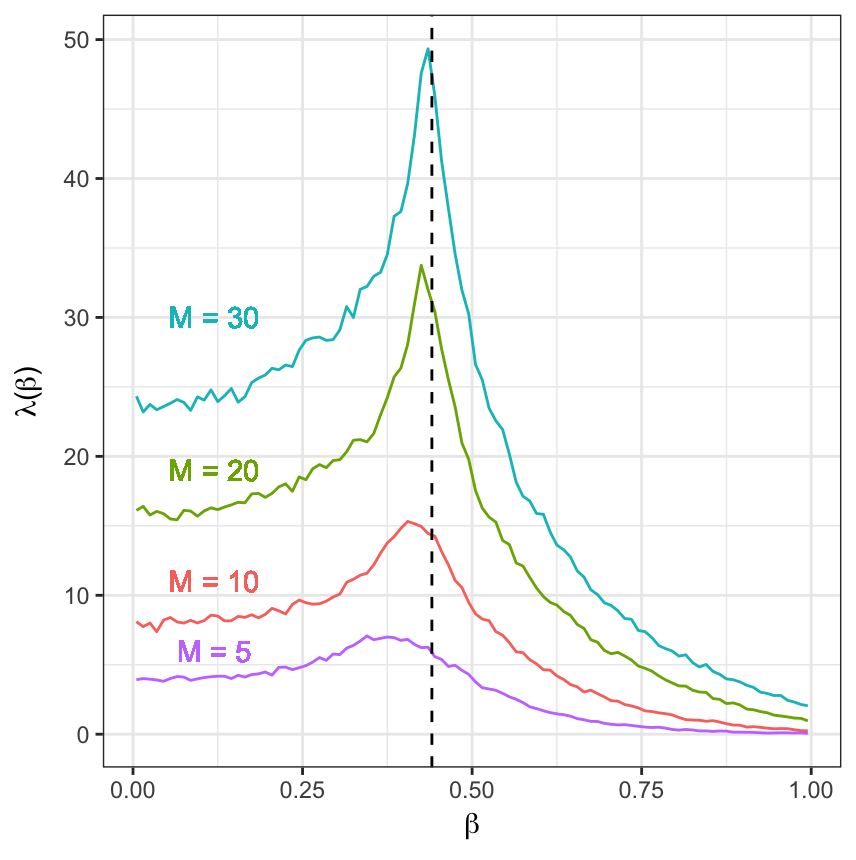}
 \hspace{.5in}
 \includegraphics[height=0.3\linewidth]{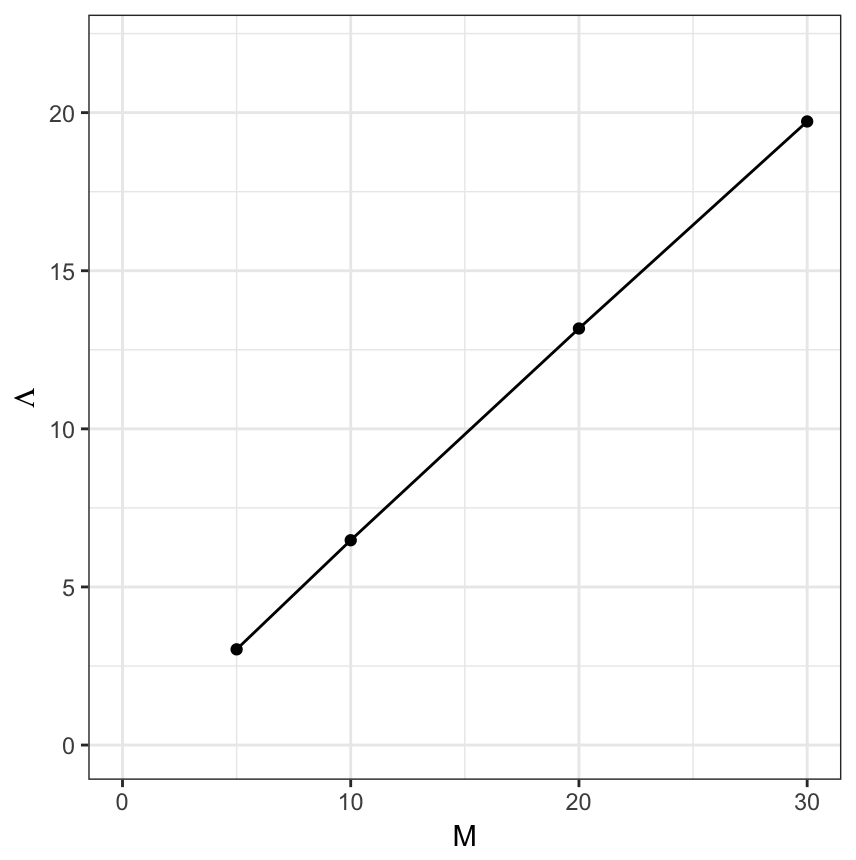}
 \end{center}
  \caption{Estimate of the local communication barrier (left) and global communication barrier (right) for the Ising model with $\mu=0$ and $M=5,10,20,30$. The vertical line is at the phase transition.}
 \label{fig:ising_communication}
 \end{figure}

 \begin{figure}[ht]
 \begin{center}
 \includegraphics[height=0.3\linewidth]{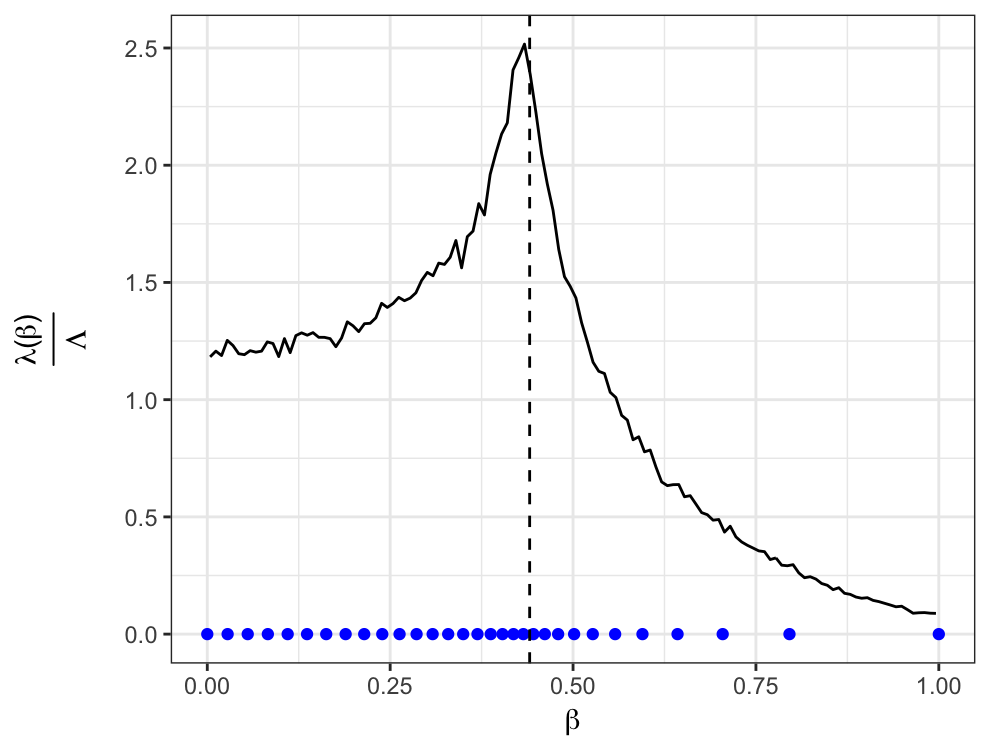}
 \hspace{.3in}
 \includegraphics[height=0.3\linewidth]{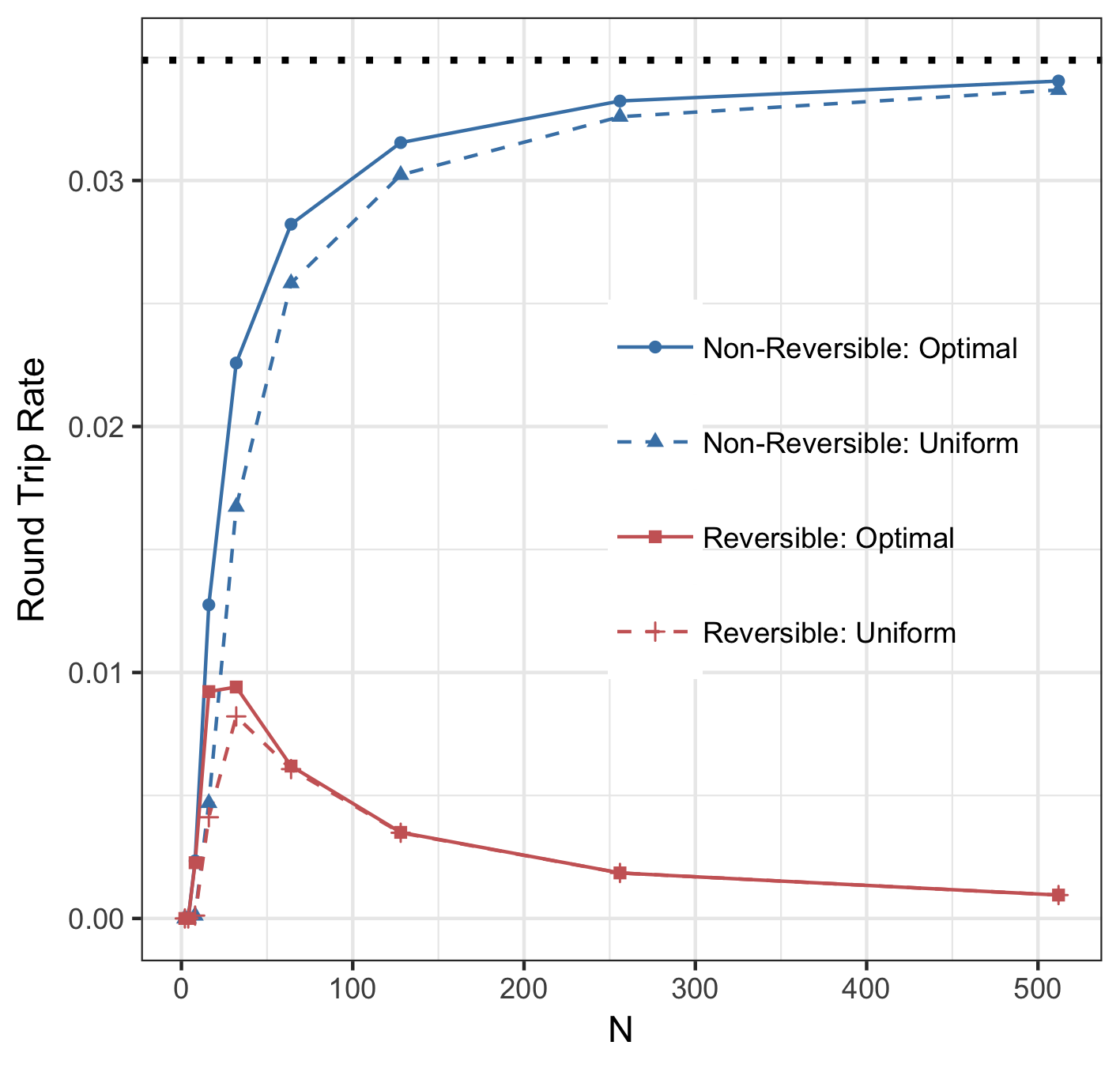}
 \end{center}
  \caption{(Left) Optimal annealing schedule for the Ising model with $M=20$, $\Lambda=13.33$ with $N=30$. The vertical line is at the phase transition. (Right) The round trip rates when $M=20$ with a uniform schedule (dashed) to the optimal schedule (solid) for both non-reversible (blue) and reversible (red) PT. The dotted horizontal line represents the approximation of the optimal round trip rate $(2+2\hat\Lambda)^{-1}$. }
 \label{fig:ising_rtr}
 \end{figure}
 
  \begin{figure}[H]
 	\centering
 	\includegraphics[width=0.3\linewidth]{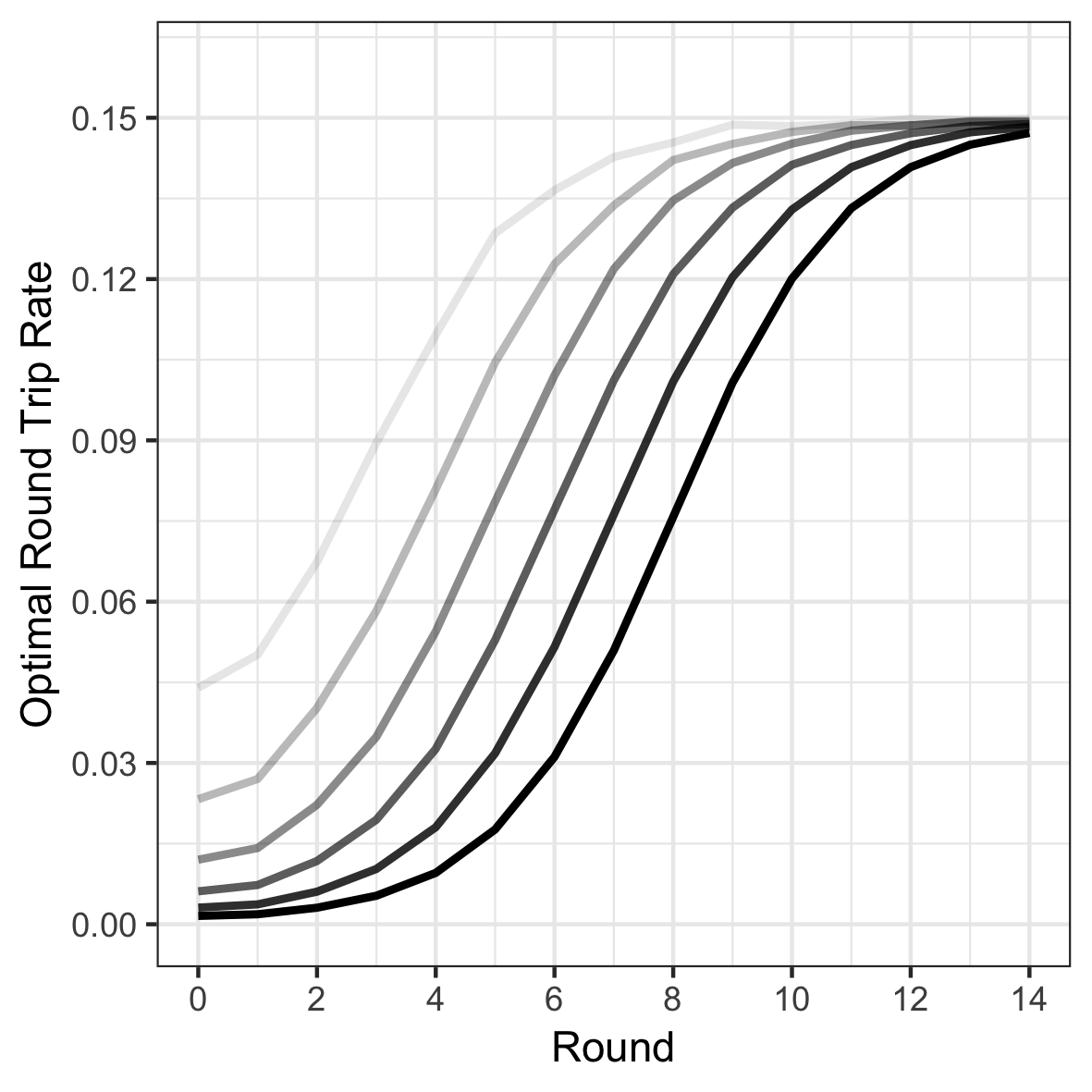}
 	\includegraphics[width=0.45\linewidth]{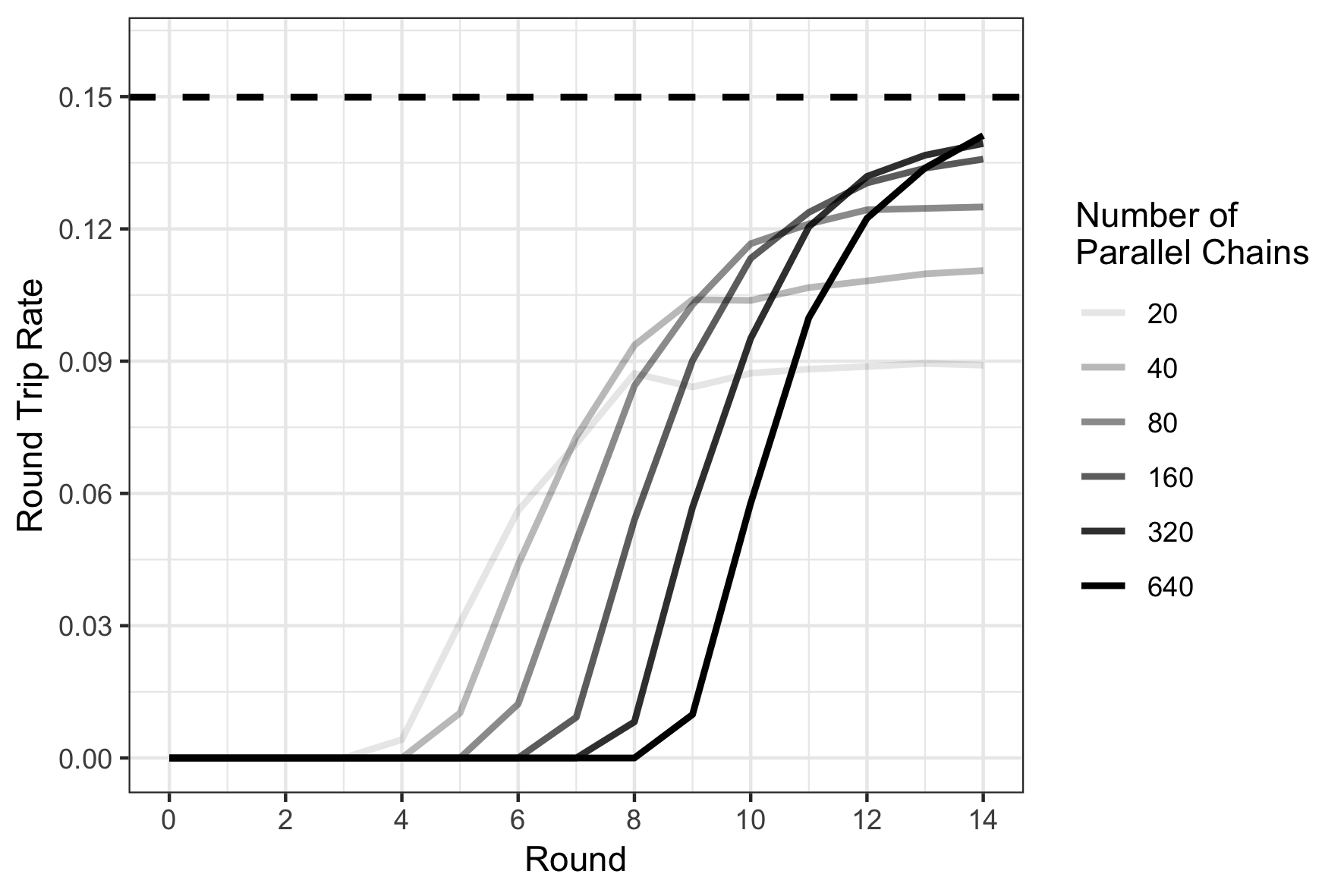}
 	\caption{Further results on robustness to ELE violation on the Ising model with a magnetic moment, $\mu = 0.1$. Impact of increasing the number $N$ of parallel chains for an example where ELE is severely violated (only half the variables are updated at each exploration step). (Left) Estimated upper bound $\bar \tau$. This quantity is well estimated for any $N$. (Right) Round trip rate $\tau$ measured from the empirical index process trajectories. Dotted line shows $\bar \tau$. The result shows $\tau$ approaches $\bar \tau$ by increasing $N$. This experiment suggests that the result in Theorem~\ref{thm_efficiency_convergence}(c) may hold under more general conditions. Moreover since $\bar \tau$ is well estimated even when the number of chains is modest, the gap between $\bar \tau$ and observed $\tau$ can be readily estimated.} 
 	\label{fig:tau-and-tau-bar} 
 \end{figure}

 \subsubsection{Discrete-multimodal problem}\label{sec_discrete} Consider a discrete state space $\statespace=\{0,\dots,2k\}$, and let $1_{\mathrm{Even}}:\Omega\to \{0,1\}$ denote the indicator function for even numbers. Define $\pi(x)\propto a^{1_{\mathrm{Even}}(x)}$ for $a>1$ and $\pi_0(x)\propto 1$ with  $V(x)=-1_{\mathrm{Even}}(x)\log a$. The distribution $\pi$ has $k+1$ modes located where $x$ is even with low probability ``barriers" located at $x$ odd. The parameter $a$ controls the relative mass put on the modes. Therefore we have
 	\begin{equation}
 	\pi^{(\beta)}(x)= \frac{a^{\beta1_{\mathrm{Even}}(x)}}{Z(\beta)},
 	\end{equation}
 	where $Z(\beta)=k+(k+1)a^\beta$. A simple computation using \eqref{def_lambda} shows that the local communication barrier is,
 	\begin{equation}\label{toy_local}
 	\lambda(\beta)=\frac{k(k+1)a^\beta\log a}{(k+(k+1)a^\beta)^2}\xrightarrow[k\to\infty]{}\frac{a^\beta\log a}{(1+a^\beta)^2}.
 	\end{equation}
 	By integrating we obtain the global communication barrier between $\pi$ and $\pi_0$,
 	\begin{align}\label{toy_global}
 	\Lambda=\frac{k(k+1)(a-1)}{(2k+1)(k+(k+1)a)}\xrightarrow[k\to\infty]{} \frac{a-1}{2(a+1)}.
 	\end{align}

\subsubsection{Mixture model}\label{sec:mixture-model}

We consider a Bayesian mixture model with two mixture components. The likelihood for each component is a normal distribution with a non-conjugate Uniform$(0, 100)$ prior on the standard deviation and a normal prior on the means (standard deviation of $100$). We placed a uniform prior on the mixture proportion. We used simulated data generated from the model. While the mixture membership indicator latent random variables can be marginalized in this model, we sample them to make the posterior inference problem more challenging. Sampling these mixture membership random variables is representative of more complex models from the Bayesian non-parametric literature where marginalization of the latent variable is intractable; for example, this is the case for the stick-breaking representation of general completely random measures \cite{Zhu2020Slice}. 

In addition to the three MCMC methods described in the main text, we also ran baselines based on SMC and AIS, which are popular methods to explore complex posterior distributions. 
These methods also depend on the construction of a sequence of annealed distributions from prior to posterior. For the SMC and AIS baselines, to select the sequence of distributions we used an adaptive scheme based on relative ESS as described in \cite{Zhou2016}. Diagnostics of the adaptation are shown in Figure~\ref{fig:nrpt-diag}. These methods were parallelized at the particle level. We set the number of particles to achieve a similar running time compared to NRTP, namely $2\,000$ particles for SMC and $2\,500$ particles for AIS. We found that the quality of the posterior approximation was highly dependent on performing several rounds of rejuvenations on the final particle population. The wall-clock time with 5 and 20 rounds of rejuvenation (SMC-5, SMC-20) was comparable to NRPT (1.778min, 2.302min) however the posterior approximation is markedly poor compared to NRPT. With 100 rounds of rejuvenation, the posterior approximation matches closely that of NRPT, however this brings the computation cost to 5.064min. AIS did not perform well since the weights were highly unbalanced in the last iteration, effectively resulting in an approximation putting all mass to a single particle.

\subsubsection{Copy number inference}\label{app:copy-number-inference}

A copy number alteration is a widespread type of mutation in cancers. Whereas in healthy cells each non-sexual chromosomes comes in two copies, in certain cancer cells, chunks of chromosomes of small and large size are frequently deleted or duplicated \cite{beroukhim_landscape_2010}---one extreme type of such event is a whole-genome duplication, in which the copy number of every genomic position is doubled. 
We consider the task of inferring the copy number for each location in the genome based on single-cell, whole genome data. For simplicity, we consider here a model for performing inference over the copy number profile of one individual cell at a time. 

The data consists in the number of reads $y_{i,c}$ measured at different loci. Here $i$ denotes a sub-region or bin of a chromosome $c$.  We use the same binning method as \cite{lai_hmmcopy_2020}. The goal is to infer for each bin $i$ and chromosome $c$ an integer copy number $x_{i,c}$. As previously reported in \cite{zahn_scalable_2017}, and confirmed in Figure~\ref{fig:chromo}, to a first order approximation the location parameter of the variable $y_{i,c}$ is determined by (1) the known GC contents of each bin, denoted $g_{i,c}$ and (2) the unknown copy number $x_{i,c}$. Here instead of performing GC-normalization as a pre-processing step, as done in \cite{zahn_scalable_2017}, we perform GC-normalization jointly with copy number inference. We describe in this section how doing so allows us to capture the uncertainty over plausible whole-genome duplication events, and induces multimodal posterior distributions.  

The single cell protocol used for the dataset considered here \cite{zahn_scalable_2017} minimizes the PCR amplification biases so that after correcting the GC-bias, the count in a bin is approximately proportional to the copy number of the bin: $ \log y_{i,c} = f(\theta, \log g_{i,c}) + \log(x_{i,c}) + \epsilon_{i,c},$ for $x_{i,c} > 0$. 
Based on Figure~\ref{fig:chromo} in Appendix~\ref{app:chromobreak}, we approximate $f$ using a latent  quadratic function, and hence $\theta$ is a global parameter containing the 3 coefficients specifying a quadratic polynomial. 
We assume $\epsilon_{i,c}$ are independent zero-mean normal random variables with standard deviation $\sigma$. 
To model the prior information that copy number events affect contiguous regions of the genome, for each chromosome $c$, we model the collection of random variables $x_{1,c}, x_{2,c}, \dots$ using a hidden Markov model (HMM). The maximum copy number in each chromosome is modelled using a geometric random variable $m_c$ with a shared latent parameter $p$. 

We placed an exponential prior on $\sigma$ with rate $1/10$. 
We posit normal priors on the coefficients $\theta$ (mean zero, variance 100). 
We put a uniform prior on $p$. The transition probabilities of the Markov chain are determined by the marginals of a Juke Cantor model with a global rate parameter $\tau$. We put a unit exponential prior on $\tau$. Full model specification is available at \url{https://github.com/UBC-Stat-ML/nowellpack/blob/multimodality-example/src/main/java/chromobreak/SingleCell.bl}.

\subsection{Multimodality of the examples considered}

Figures~\ref{fig:mixture-V_vs_X} (bottom left), \ref{fig:multimodality-spike-slab} and \ref{fig:multimodality-ising} support the multimodality of two of the examples considered in Section~\ref{sec_comparison}, namely the Ising and Spike-and-Slab examples. 
Multimodality of other examples considered in Section~\ref{sec:empirical}, \ref{sec_ELE_violation}, \ref{sec:mixture-example} and \ref{sec:chromobreak} is demonstrated in  
Figures~\ref{fig:mixture-V_vs_X}, \ref{fig:chromo}, \ref{fig:trace-plots-mcmcs}, \ref{fig:multimodality-chromo} and \ref{fig:nstates-traces}. Moreover, Figures~\ref{fig:mixture-V_vs_X}, \ref{fig:trace-plots-mcmcs}, \ref{fig:multimodality-chromo} and \ref{fig:nstates-traces} demonstrate that using standard MCMC is insufficient to explore these multimodal distributions.

\subsection{Selection of $\pi_0$}\label{sec:selection-of-pi0}

In our experiments all Bayesian models considered have proper priors, and hence we set $\pi_0$ to the prior in all these situations. The local exploration kernel in this case consists in an independent draw from the prior (performed by sorting the latent random variables according to a linearization of the partial order induced by the directed graphical model, and sampling their values according to these sorted laws).

For the Ising model, we let $\pi_0$ denote a product of independent and identically distributed Bernoulli($1/2$) random variables, one for each node in the Ising grid. It is then straightforward to interpolate between this product distribution and the Ising model of interest using a geometric average. Independent sampling from $\pi_0$ is used for the local exploration kernel at $\beta = 0$. 

Similarly, for the rotor (XY) model, we take $\pi_0$ to be a product of independent and identically distributed Uniform($-\pi, \pi$) random variables, and we use a geometric interpolation between $\pi_0$ and the target distribution.
 
\clearpage

\section{Supplementary figures for Section~\ref{sec_examples}}\label{app_algos}
\subsection{Section~\ref{sec:empirical}}\label{app:empirical}

\begin{figure}[H]
	\begin{center} 
		\includegraphics[width=0.57\linewidth]{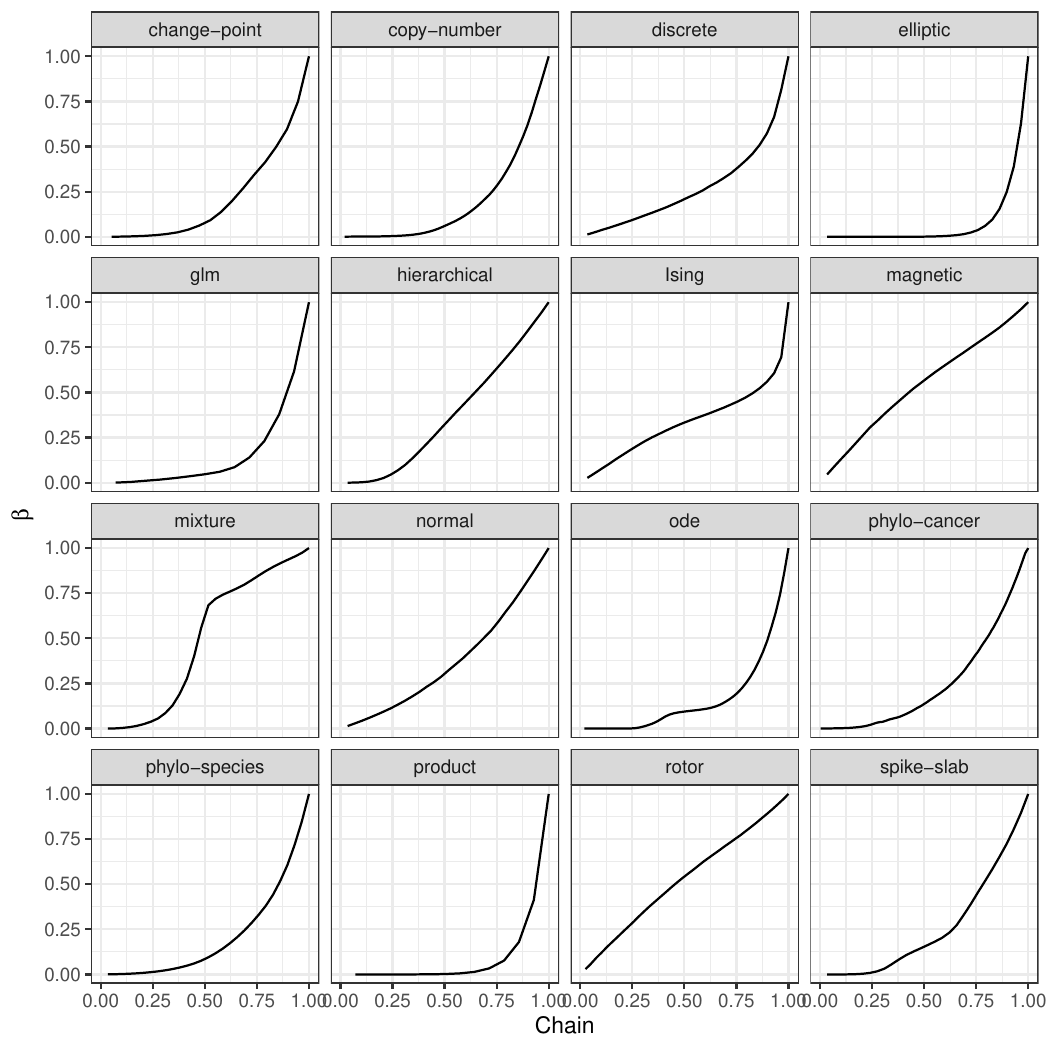} \\
		\includegraphics[width=0.57\linewidth]{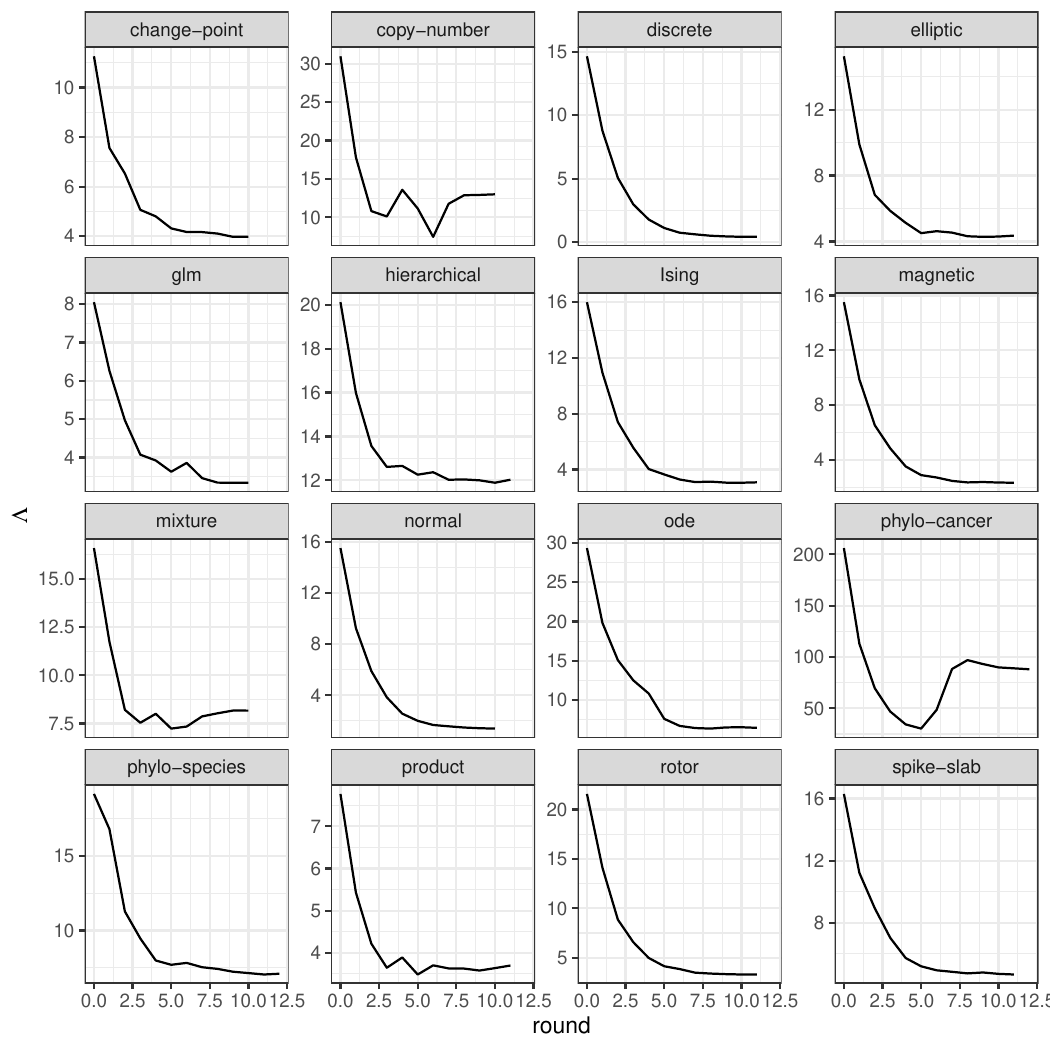} 
	\end{center}
	\caption{Top: Estimates of the schedule generator $G$ for 16 models. The abscissa denotes the normalized chain indices $i/N$, the ordinates, parameters $\beta$. The function $G$ is such that the partition $G(0), G(1/N), G(2/N), \dots, G(1)$ approximates equi-acceptance. Bottom: Estimates of $\Lambda$ for 16 models. The abscissa denotes the schedule optimization round.}
	\label{fig:Gs}
\end{figure}

\subsection{Section~\ref{sec_comparison}}\label{app:comparison}
 
\begin{figure}[H]
	\begin{center}
			\begin{tabular}{lllll}  
				\toprule
				Model (and dataset when applicable)  & Method & $N$ & $\hat s$ (\%) & $\hat \Lambda$ \\
				\midrule
				Spike-and-slab classification & NRPT+SCM & $15$ & $67.6$ & $4.54$ \\
				\;\;(RMS Titanic passengers data, \cite{Hind_2019}) & ARR & $6$ & $36.3$ & $3.82$ \\
				& MMV & $6$ & $39.9$ & $3.61$ \\
				& NRPT & $15$ & $67.2$ & $4.59$ \\
				\midrule
				Bayesian hierarchical model  & NRPT+SCM & $34$ & $63.7$ & $11.96$ \\
				\;\;(historical rocket failure data, \cite{McDowell_2019}) & ARR$^*$ & $2^*$ & $0^*$ & $1^*$ \\
				& MMV & $15$ & $33.6$ & $9.96$ \\
				& NRPT & $34$ & $63.4$ & $12.08$ \\
				\midrule
				Wright-Fisher diffusion & NRPT+SCM & $15$ & $75.2$ & $3.47$ \\
				& ARR & $5$ & $43.8$ & $2.81$ \\
				& MMV & $5$ & $43.7$ & $2.82$ \\
				& NRPT & $15$ & $75.5$ & $3.42$ \\
				\midrule	
				Ising model & NRPT+SCM & $15$ & $78.2$ & $3.05$ \\
				& ARR & $4$ & $37.9$ & $2.48$ \\
				& MMV & $4$ & $36.9$ & $2.52$ \\
				& NRPT & $15$ & $78.5$ & $3.01$ \\
				\bottomrule
			\end{tabular}
		\end{center}
		\caption{Summary statistics for the experiments in Section~\ref{sec_comparison}. The row marked with a star indicates failed optimization of one of the stochastic optimization schemes on the Bayesian hierarchical model.}
		\label{tab:summary-stat-comparisons}
\end{figure}

\begin{figure}[H]
	\begin{center}
		\includegraphics[width=0.24\linewidth]{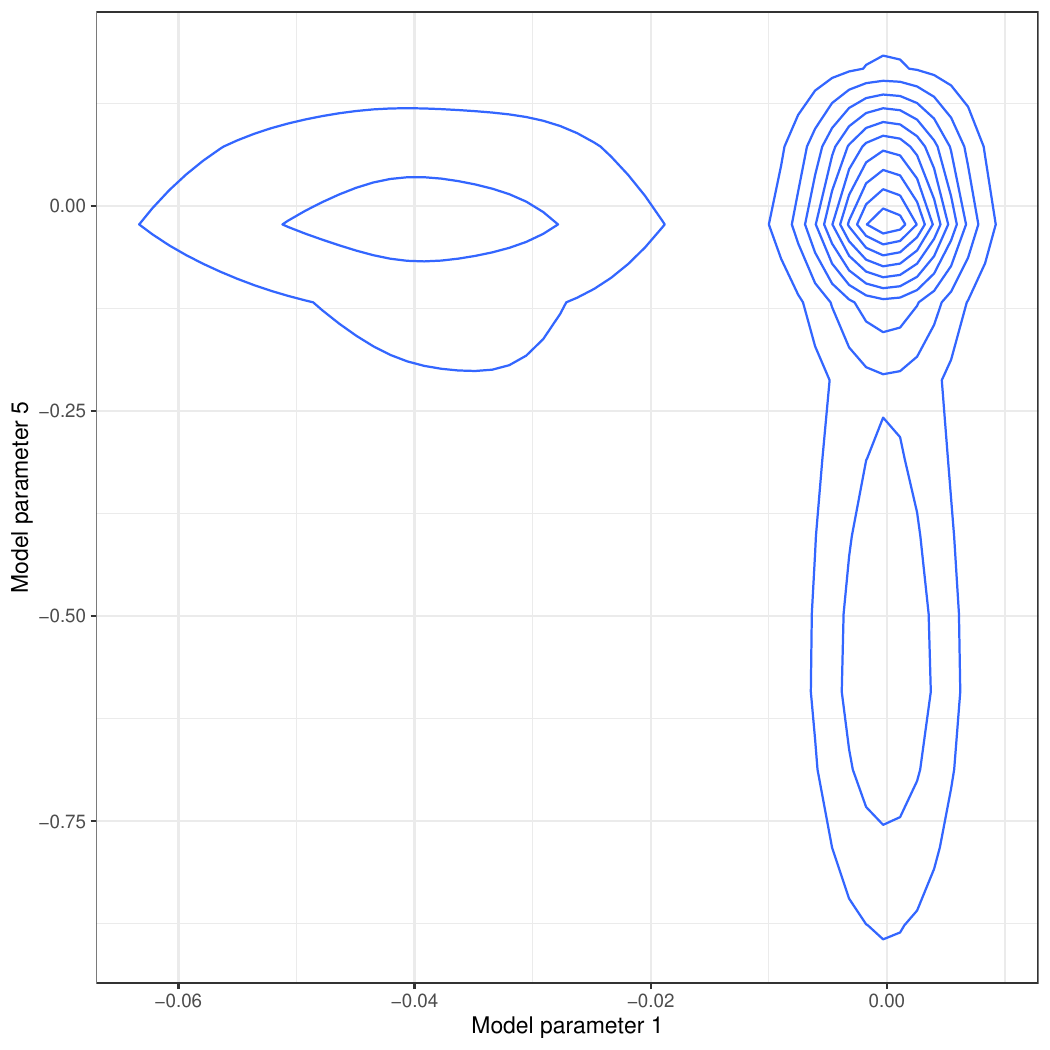}
		\includegraphics[width=0.24\linewidth]{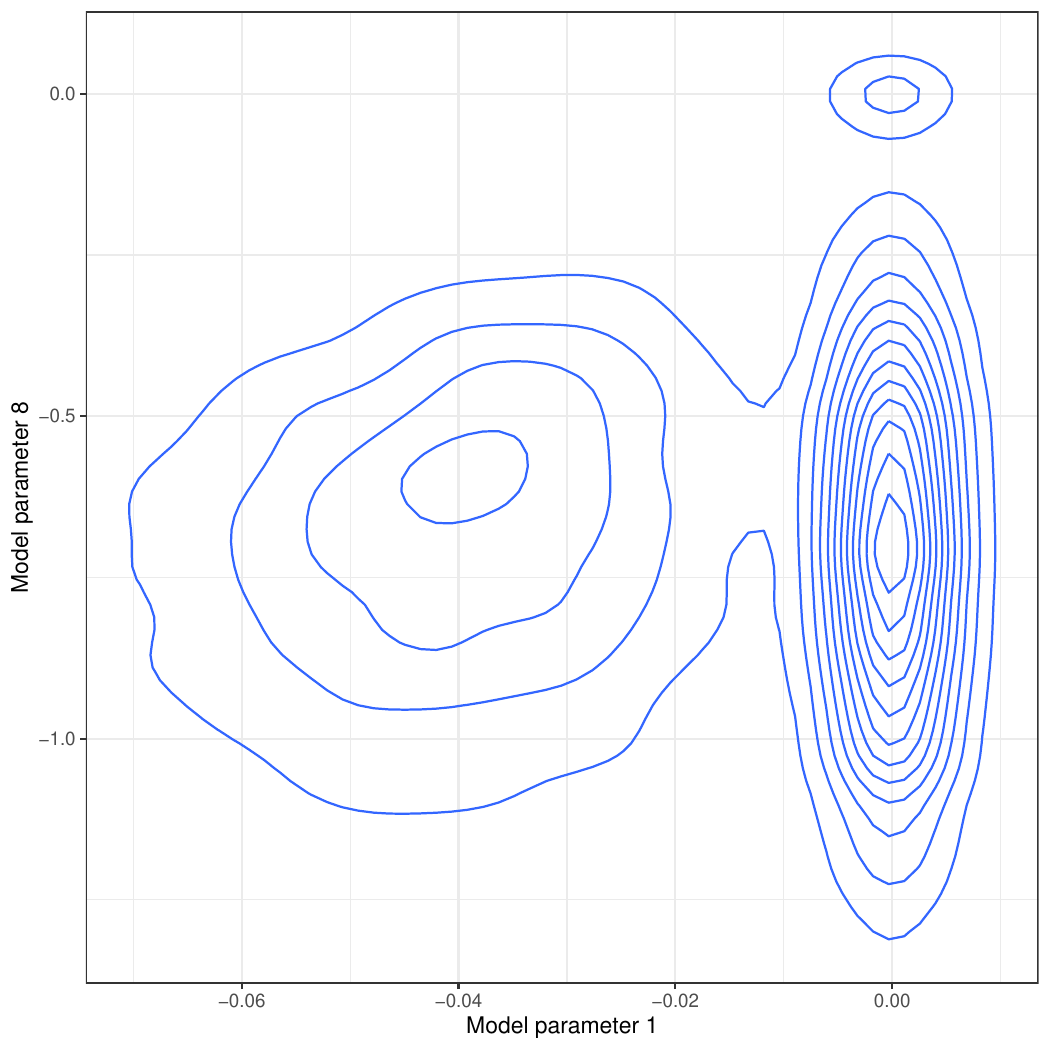}
		\includegraphics[width=0.24\linewidth]{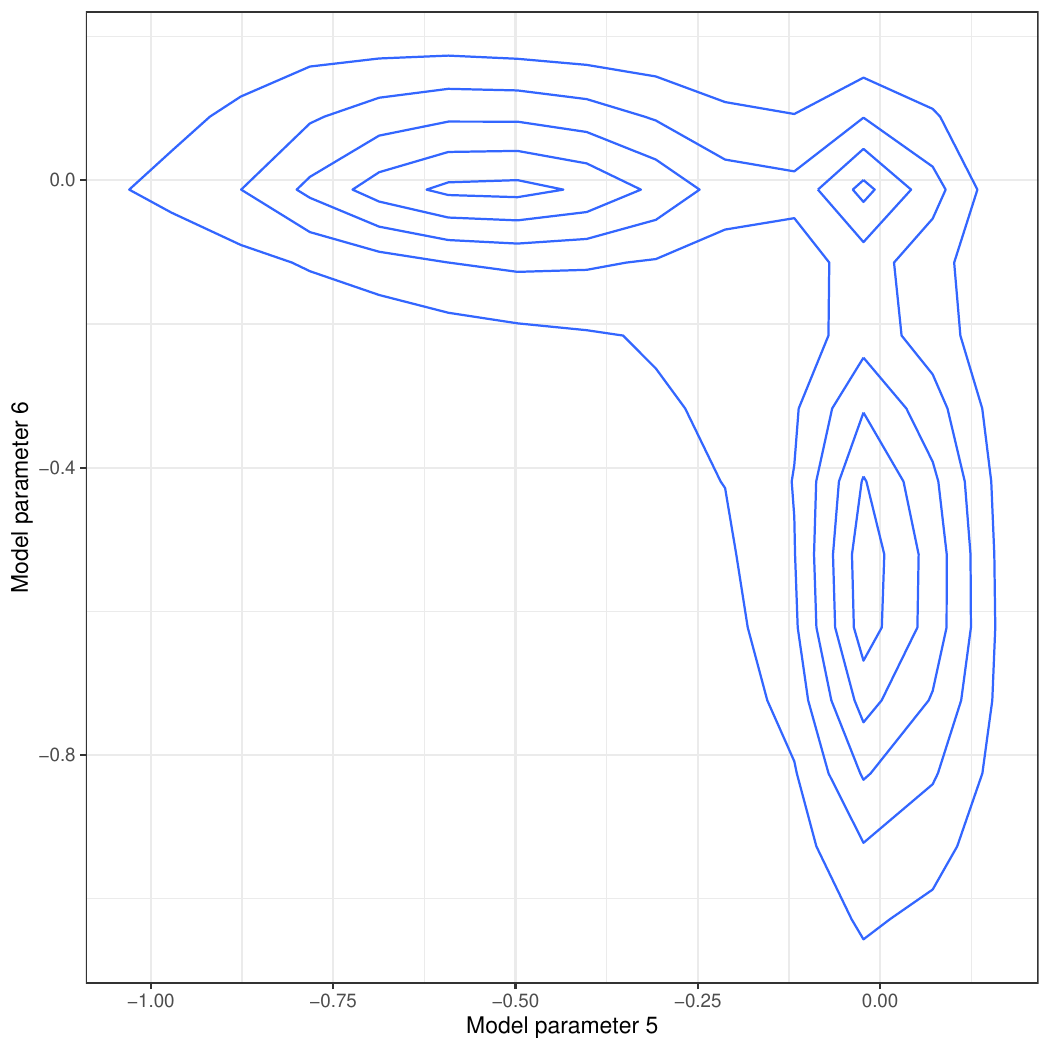}
	\end{center}
	\caption{Examples of multimodality encountered in the Spike-and-Slab model applied to the RMS Titanic passenger dataset. Joint posterior distribution of the following regression parameters  (1) \texttt{passengerAge} and \texttt{passengerClass}; (2) \texttt{passengerAge} and \texttt{SiblingsSpousesAboard}; (3) \texttt{passengerClass} and \texttt{passengerClass2}, the latter corresponding to an artificially duplicated regressor used to investigate the effect of co-linearity on the posterior approximation.}
	\label{fig:multimodality-spike-slab}
\end{figure}

\begin{figure}[H]
	\begin{center}
		\includegraphics[width=0.6\linewidth]{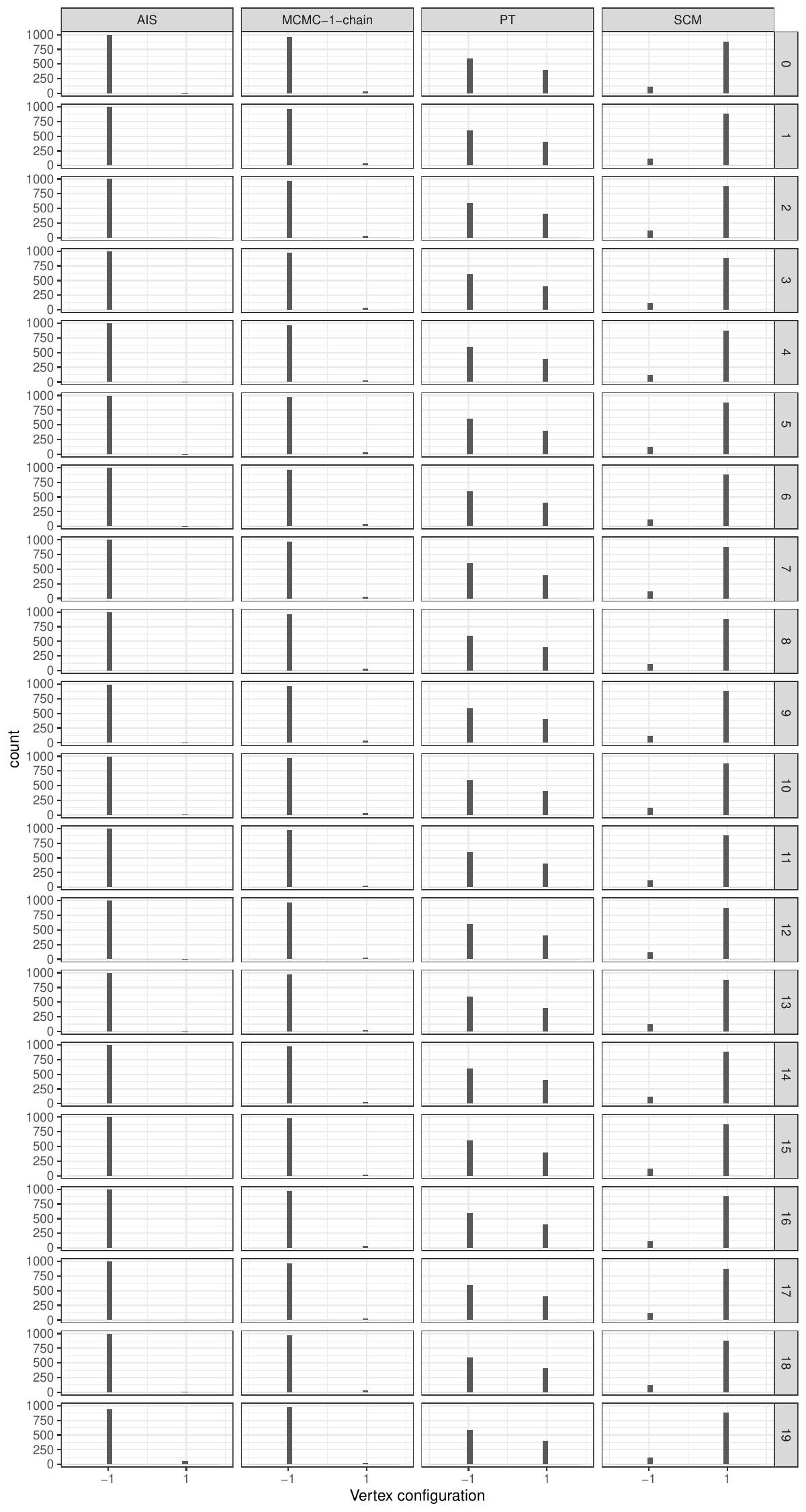}
	\end{center}
	\caption{Example of multimodality in an Ising model. Each facet is a posterior probability mass function, facet rows index the first 20 vertices of the Ising graph, facet columns, four posterior approximation methods. By symmetry, each marginal should place equal mass at spins $-1$ and $1$. The multimodality is only correctly captured by our proposed algorithm (NRPT, third column, wall clock time of 7.250s, $N = 10$). For the other methods, the posterior approximation either misses the other mode completely, namely when using AIS (left column, wall clock time of 5.723s, 1000 particles, relative ESS adaptive annealing schedule) or single-chain MCMC (second column, wall clock time of 1.439s), or, for annealed SMC (fourth column, wall clock time of 8.269s, 1000 particles, relative ESS adaptive annealing schedule), detects the multimodality but not their respective proportions. All wall clock time are reported for a 2.8 GHz Intel Core i7.}
	\label{fig:multimodality-ising}
\end{figure}

\subsection{Section~\ref{sec:mixture-example}}\label{app:mixture-example}
 \begin{figure}[H]
 	\begin{center}
 		\includegraphics[width=0.32\linewidth]{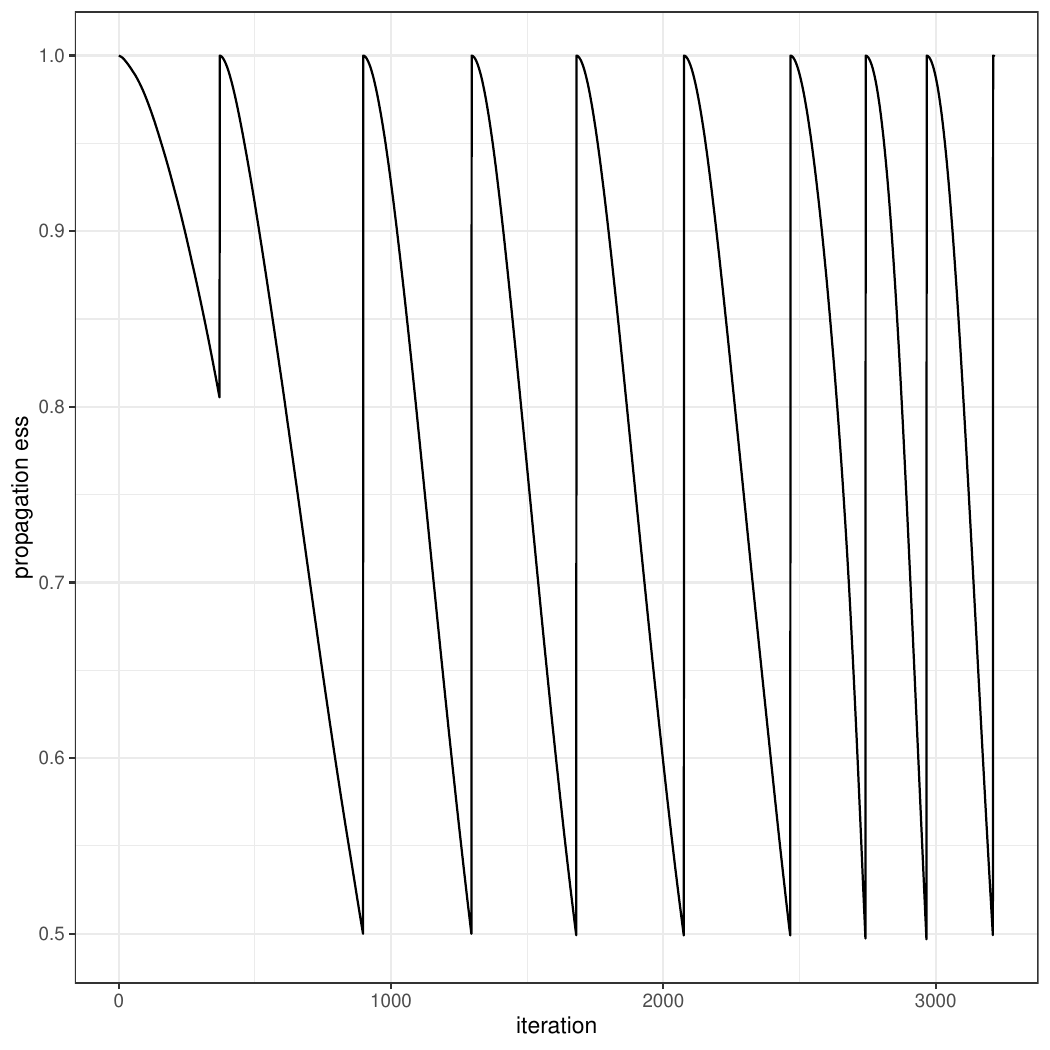} 
 		\includegraphics[width=0.32\linewidth]{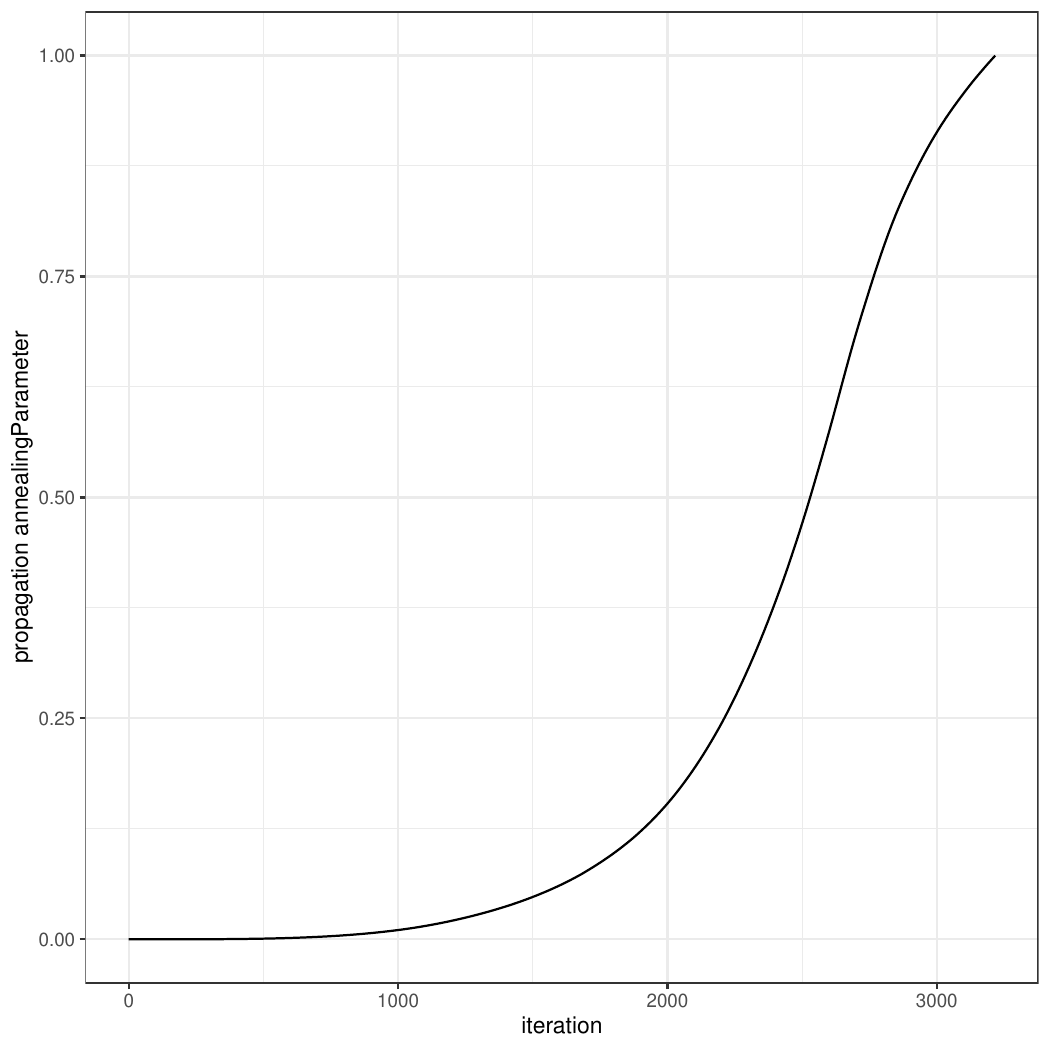}
 		\includegraphics[width=0.32\linewidth]{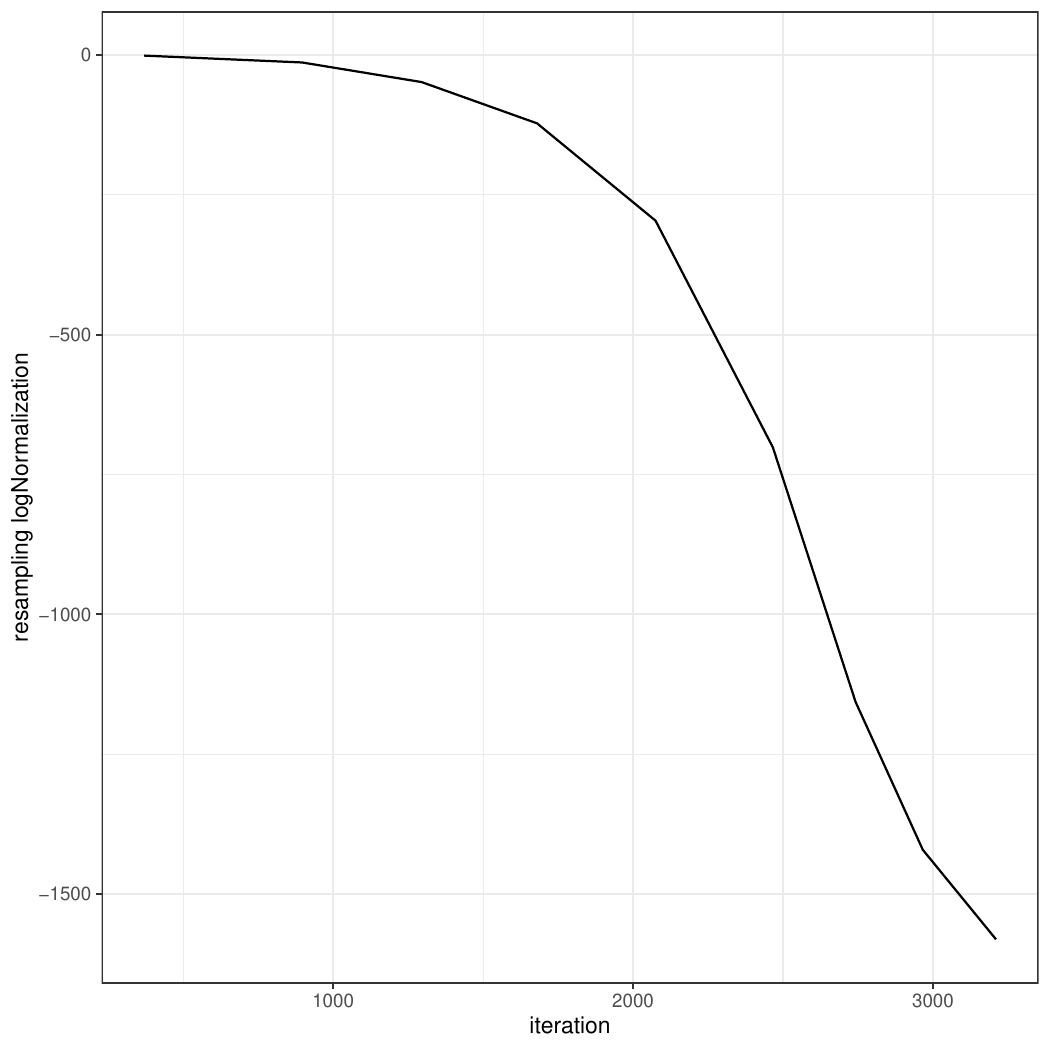} 
 	\end{center}
 	\caption{Diagnostics for the adaptive annealed SMC method (SCM), used as a benchmark on the Bayesian mixture problem. From left to right: (1) ESS as a function of the SMC iteration. Resampling is performed when the ESS drops below $1/2$; (2) annealing parameter as a function of the SMC iteration; (3) log normalization estimates at each resampling step.}
 	\label{fig:scm-diag}
 \end{figure} 
 
 \begin{figure}[H]
 	\begin{center}
 		\includegraphics[width=0.24\linewidth]{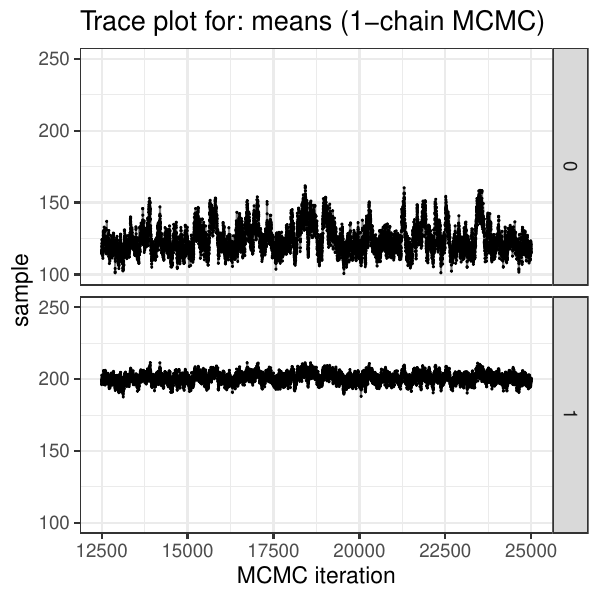} 
 		\includegraphics[width=0.24\linewidth]{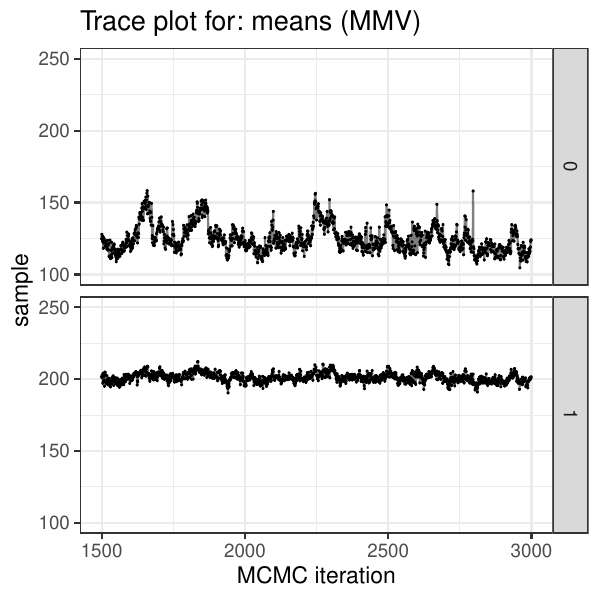}
 		\includegraphics[width=0.24\linewidth]{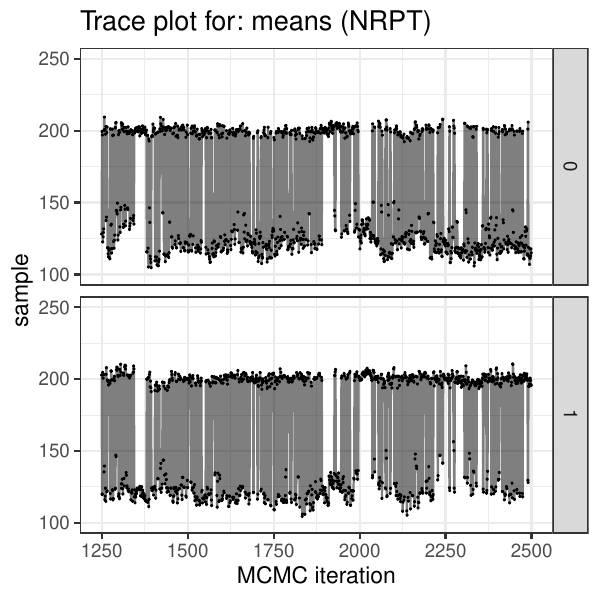} 
 		\includegraphics[width=0.24\linewidth]{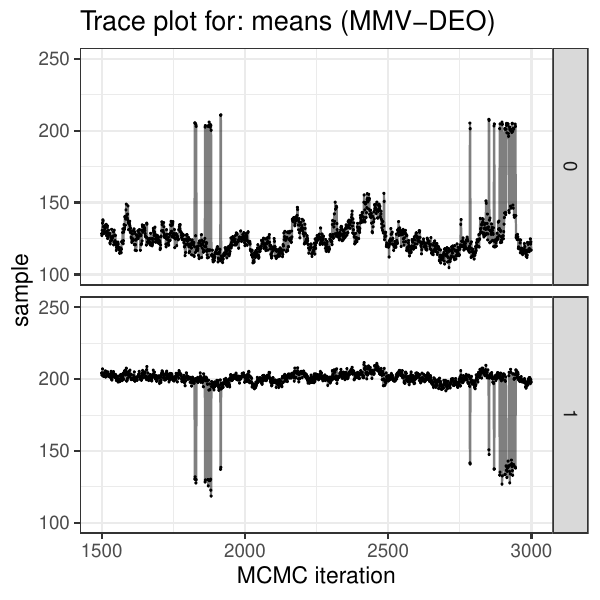}
 	\end{center}
 	\caption{Mixture modelling example: post burn-in trace plots of two model parameters (facet rows) for four MCMC methods considered (facet columns). The two parameters correspond to the location parameters of two exchangeable clusters. Correct MCMC exploration of this unidentifiable model requires label switching, providing a test bed for MCMC over multimodal targets. All methods use a computational budget lower or equal to NRPT's.}
 	\label{fig:trace-plots-mcmcs}
 \end{figure}

 \begin{figure}[H]
 	\begin{center}
 		\includegraphics[width=0.24\linewidth]{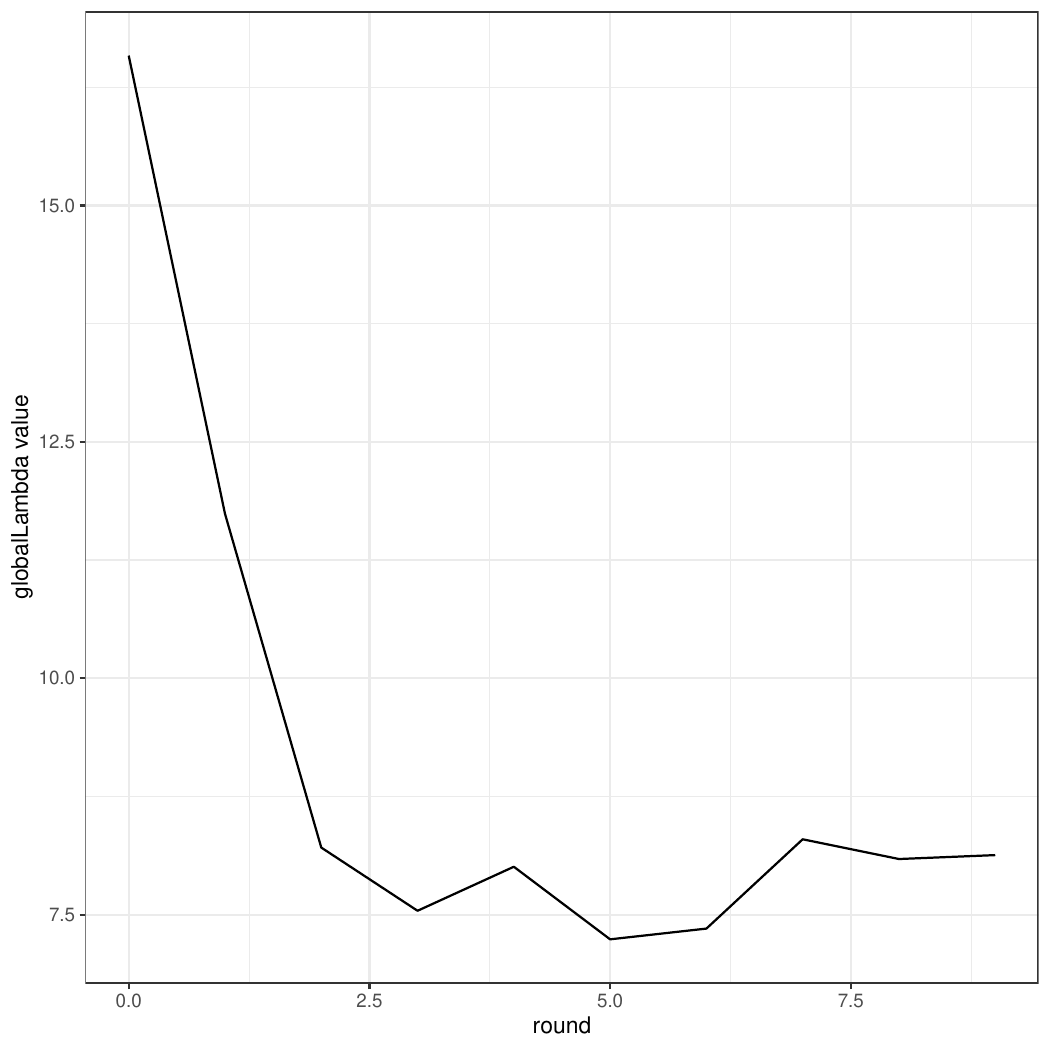}
 		\includegraphics[width=0.24\linewidth]{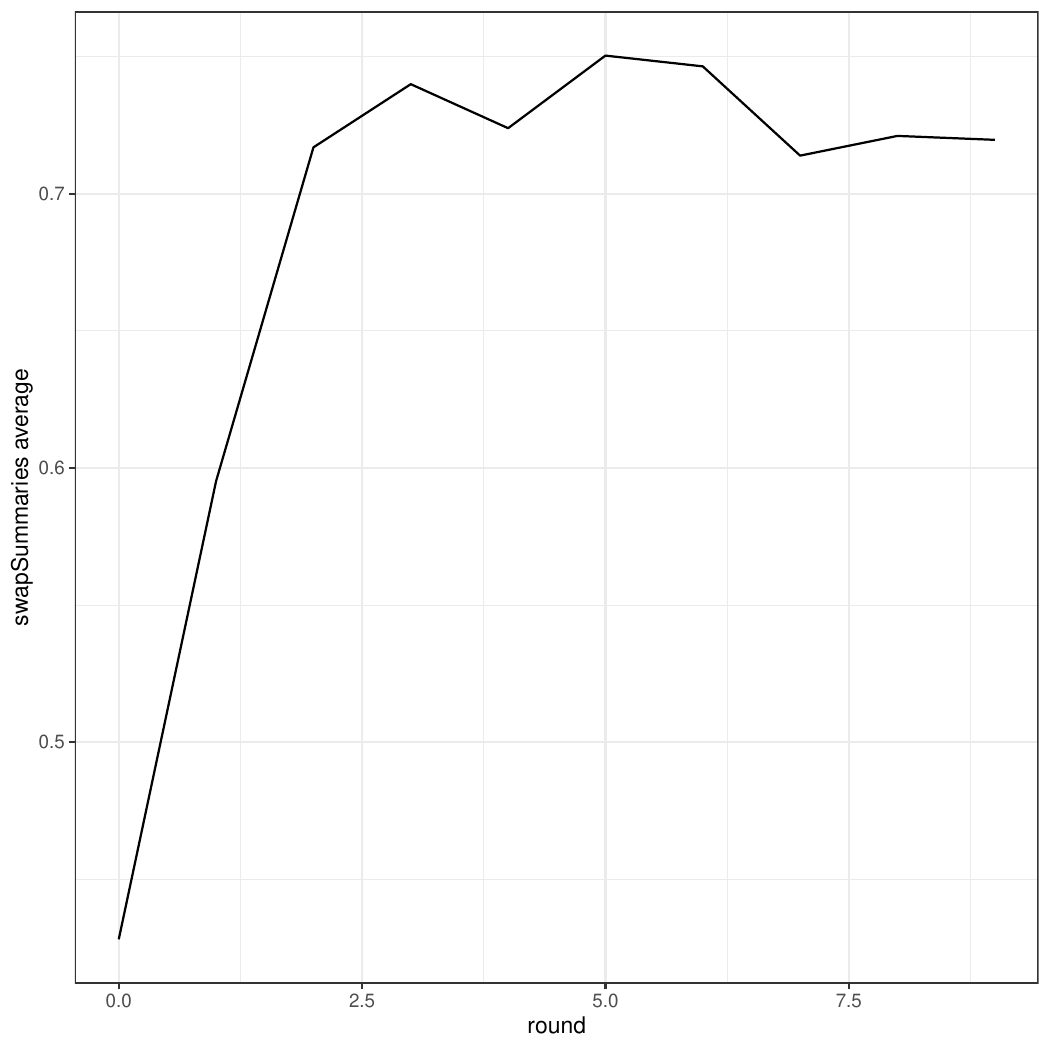}
 		\includegraphics[width=0.24\linewidth]{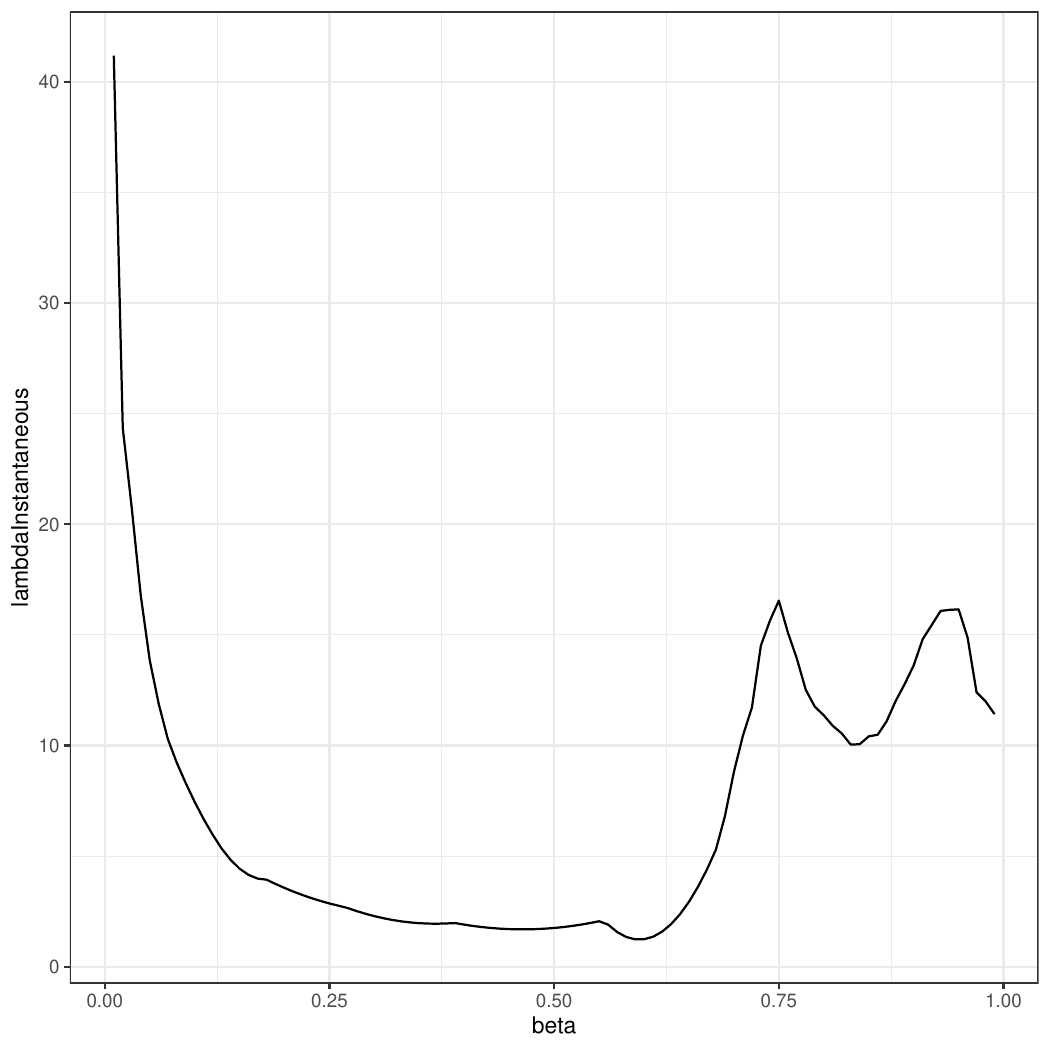}
 		\includegraphics[width=0.24\linewidth]{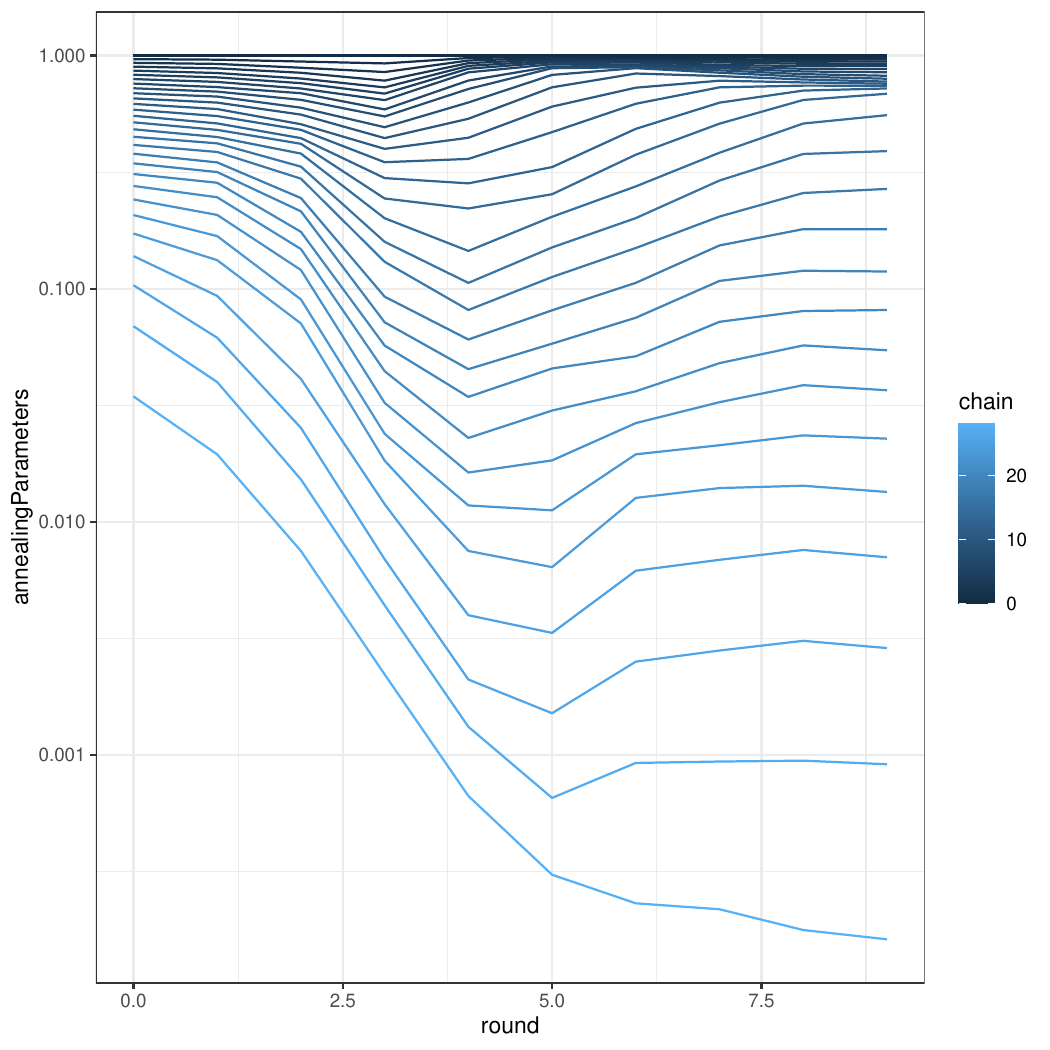}
 	\end{center}
 	\caption{Diagnostics for NRPT on the mixture modelling example. From left to right: (1) Estimate of the global communication barrier $\hat \Lambda$ as a function of the schedule optimization round. (2) Average swap acceptance probability across the 30 chains as a function of the schedule optimization round. The final value is $72\%$. (3) Estimated communication barrier $\lambda$ output by NRPT. The two peaks can be interpreted as transition points where the cluster membership indicator variables go from disorganized (all cluster membership variables are i.i.d.\ at $\beta=0$), to all taking the same value within a cluster. The two clusters having different  number of data points and parameters induce two  distinct transition points. The empirical behaviour observed is analogous to the phase transition found in statistical mechanics models such as the Ising model. Investigation of possible phase transitions in clustering models would be an interesting future direction of investigation, although somewhat orthogonal to this work. (4) Learning curves for the annealing parameters (ordinate axis, log scale) for the $30$ NRPT chains (colours) as a function of the schedule optimization round (abscissa).}
 	\label{fig:nrpt-diag}
 \end{figure}

\subsection{Section~\ref{sec:chromobreak}}\label{app:chromobreak}

\begin{figure}[H]
	\begin{center}
		\includegraphics[width=0.19\linewidth]{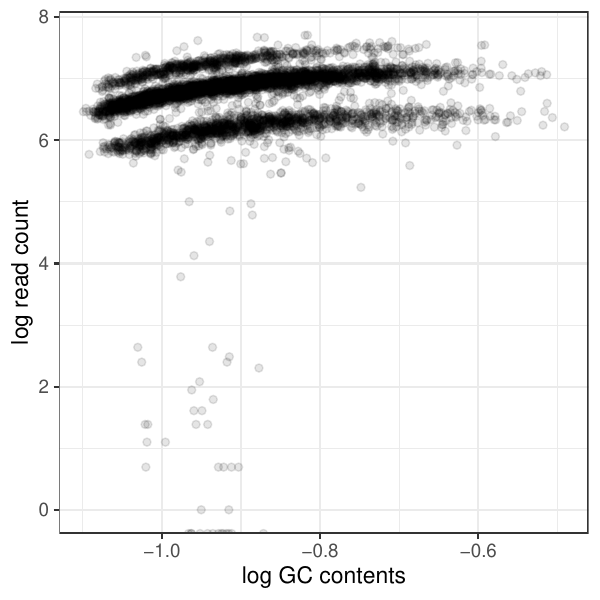} 
		\includegraphics[width=0.19\linewidth]{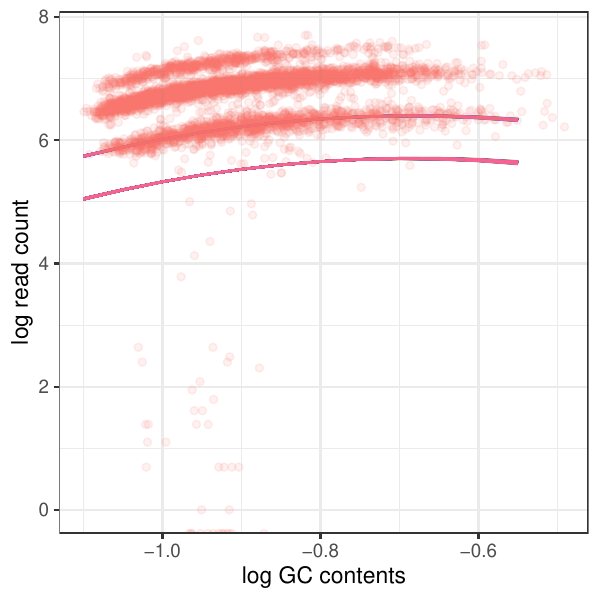} 
		\includegraphics[width=0.19\linewidth]{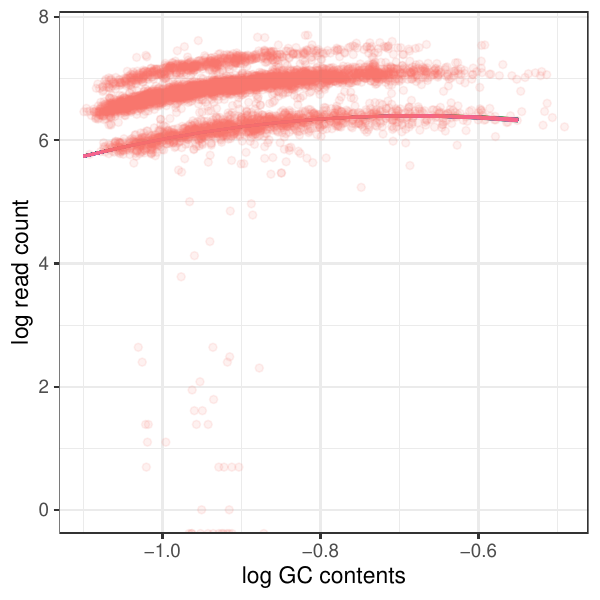} 
		\includegraphics[width=0.17\linewidth]{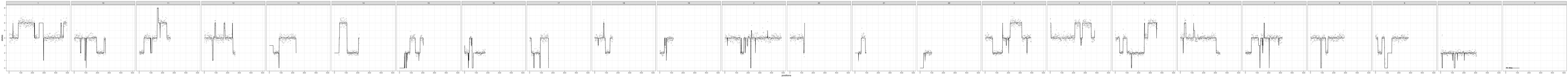}
		\includegraphics[width=0.17\linewidth]{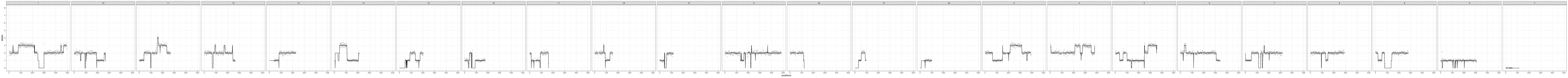}
	\end{center}
	\caption{Copy number inference: inputs and outputs. From left to right: (1) raw data, where each dot represent a genomic bin $i, c$, with its abscissa showing $\log g_{i,c}$ and its ordinate, $\log y_{i,c}$. (2) posterior distribution obtained from NRPT for the random function $f(\theta, \cdot)$ obtained from NRPT. One translucent line $f(\theta_n, \cdot)$ is drawn for each sampled $\theta_n$. Notice the bi-modality of the posterior distribution. (3) The approximation of the same posterior distribution based on a single chain misses one of the modes. (4) Two regions in chromosome 2 supporting a plausible genome duplication (and hence multimodality). We show here a sample index $n$ corresponding to the tetraploid mode in the second sub-panel. The dots show the data after GC-correction based on sample $n$, i.e.\ $\exp(\log y_{i,c} - f(\theta_n, \log g_{i,c}))$. The line shows the imputed copy numbers at iteration $n$. Notice the two small regions with odd copy numbers, around position 25 and just before position 200. (5) Same visualization as  panel (4) but for a sample index $n$ corresponding to the diploid mode. Regions 25 and 200 cannot be explained without assuming genome duplication, however since the regions supporting duplication are small, the model remains uncertain about ploidy and hence multimodal.}
	\label{fig:chromo}
\end{figure}

\begin{figure}[H]
	\begin{center}
		\includegraphics[width=0.09\linewidth]{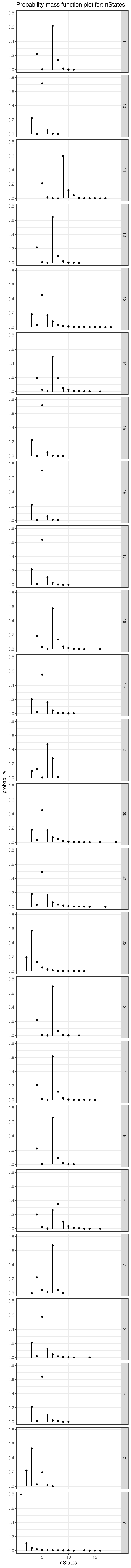}
		\includegraphics[width=0.09\linewidth]{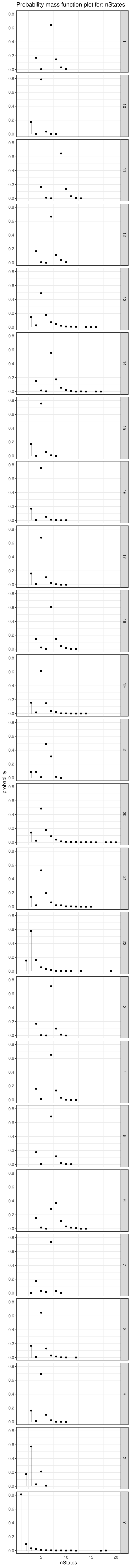}
		\includegraphics[width=0.09\linewidth]{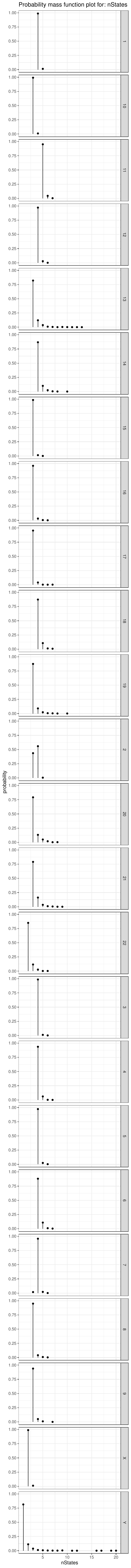}
		\includegraphics[width=0.09\linewidth]{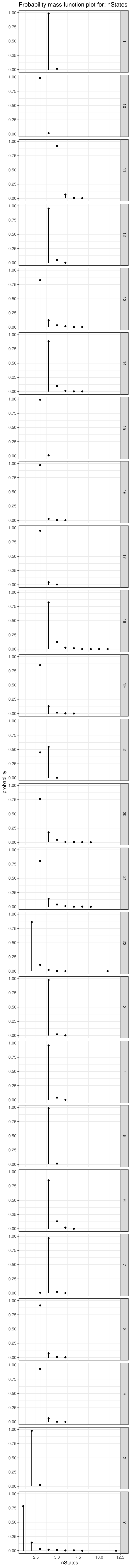}
		
	\end{center}
	\caption{Multimodality of the random variables $m_c$ controlling the size of the copy number space for each chromosome $c$. Each facet displays an approximate posterior probability mass function. Facet rows are indexed by chromosomes. Facet columns are indexed by posterior inference methods, which are all MCMC schemes ran for 5000 iterations. (1, 2) The multimodality arising from a weakly supported genome duplication is captured by NRPT, based on $N=50$ and $N=25$ respectively. (3) Results for an MCMC run based on the same exploration kernel misses the multimodality. (4) The stochastic adaptation method MMV fails to learn a schedule capable of performing round trips within the 5000 iterations and therefore misses the multimodality as well. See Section~\ref{sec:chromobreak} for details.}
	\label{fig:multimodality-chromo}
\end{figure}

\begin{figure}[H]
	\begin{center}
		\includegraphics[width=0.85\linewidth]{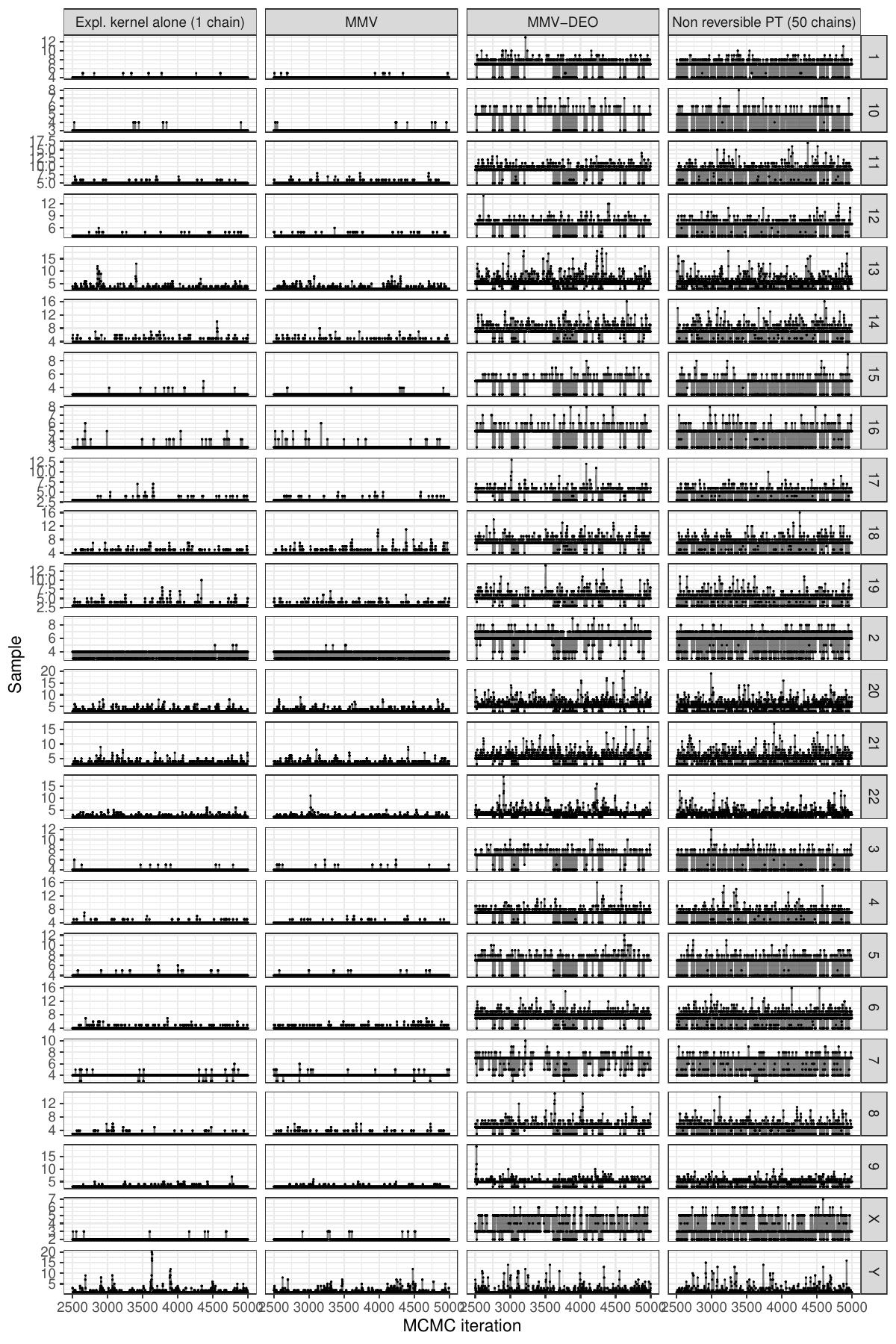}
	\end{center}
	\caption{Enlarged version of Figure~\ref{fig:nstates-traces-small}.}
	\label{fig:nstates-traces}
\end{figure}

\begin{figure}[H]
	\begin{center}
		\includegraphics[width=0.24\linewidth]{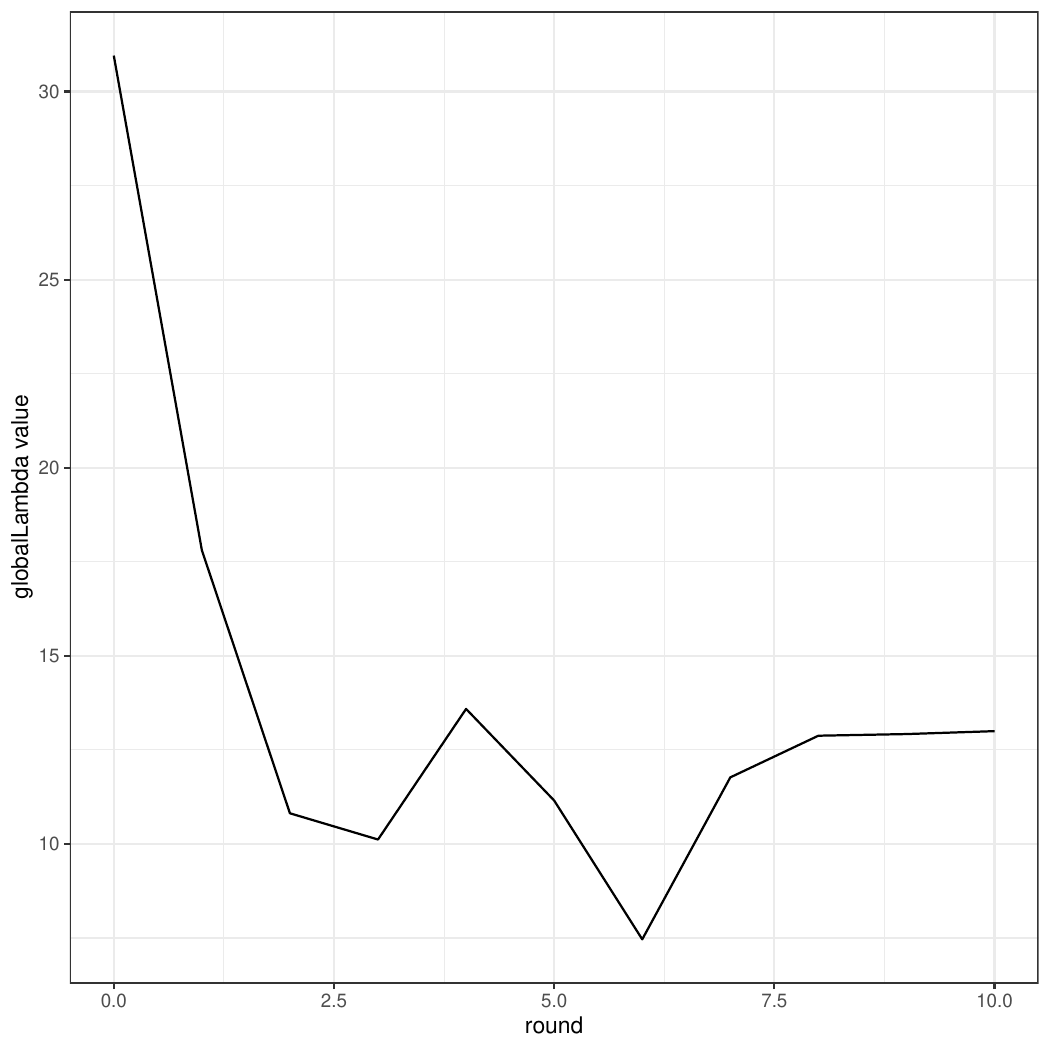}
		\includegraphics[width=0.24\linewidth]{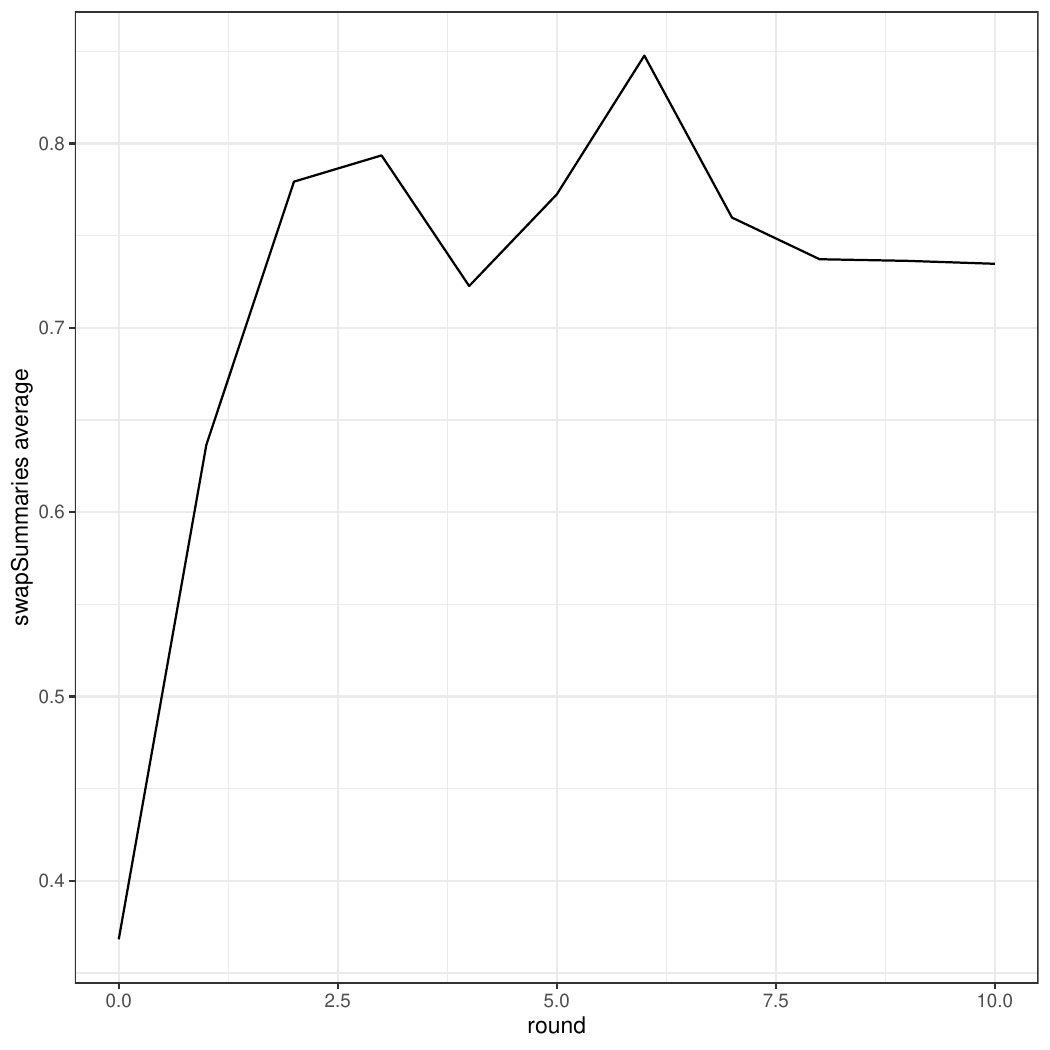}
		\includegraphics[width=0.24\linewidth]{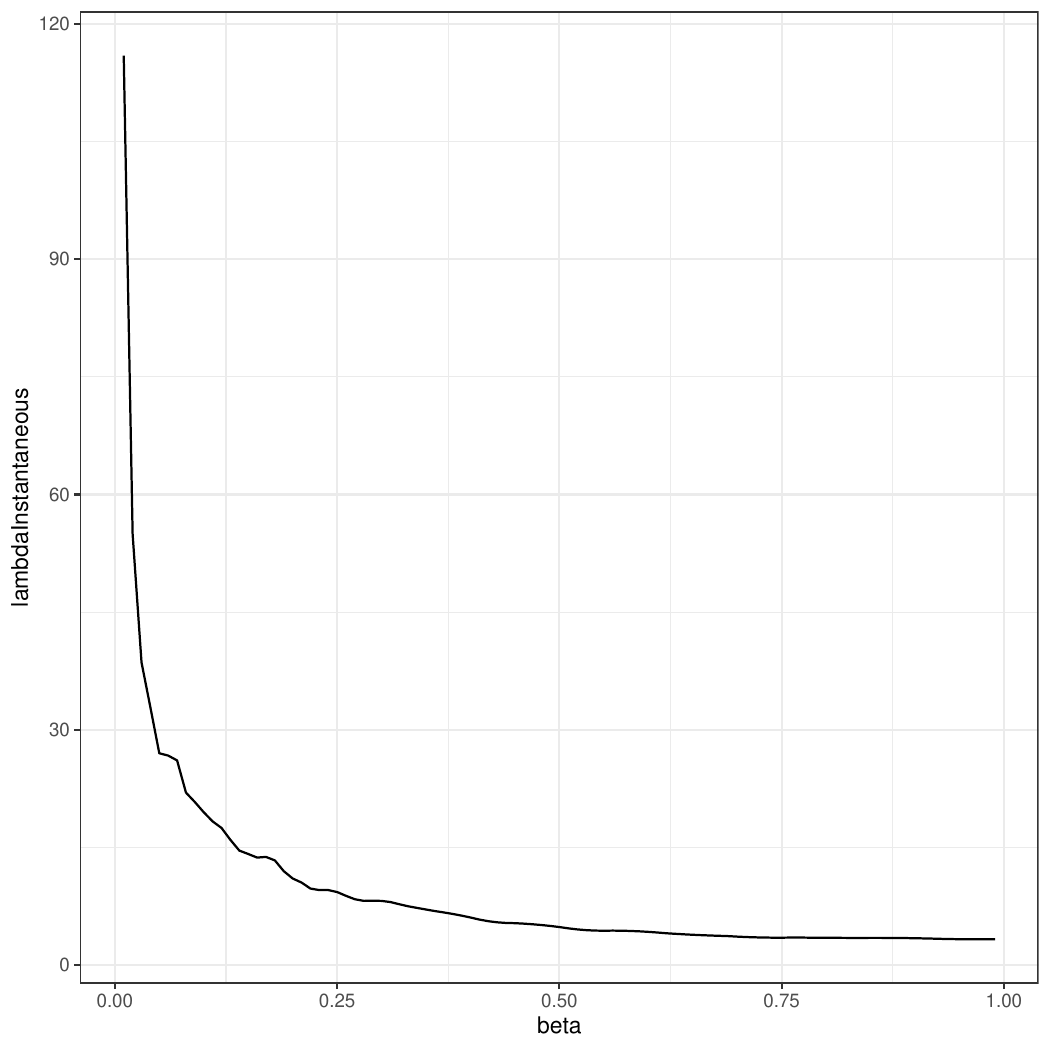}
		\includegraphics[width=0.24\linewidth]{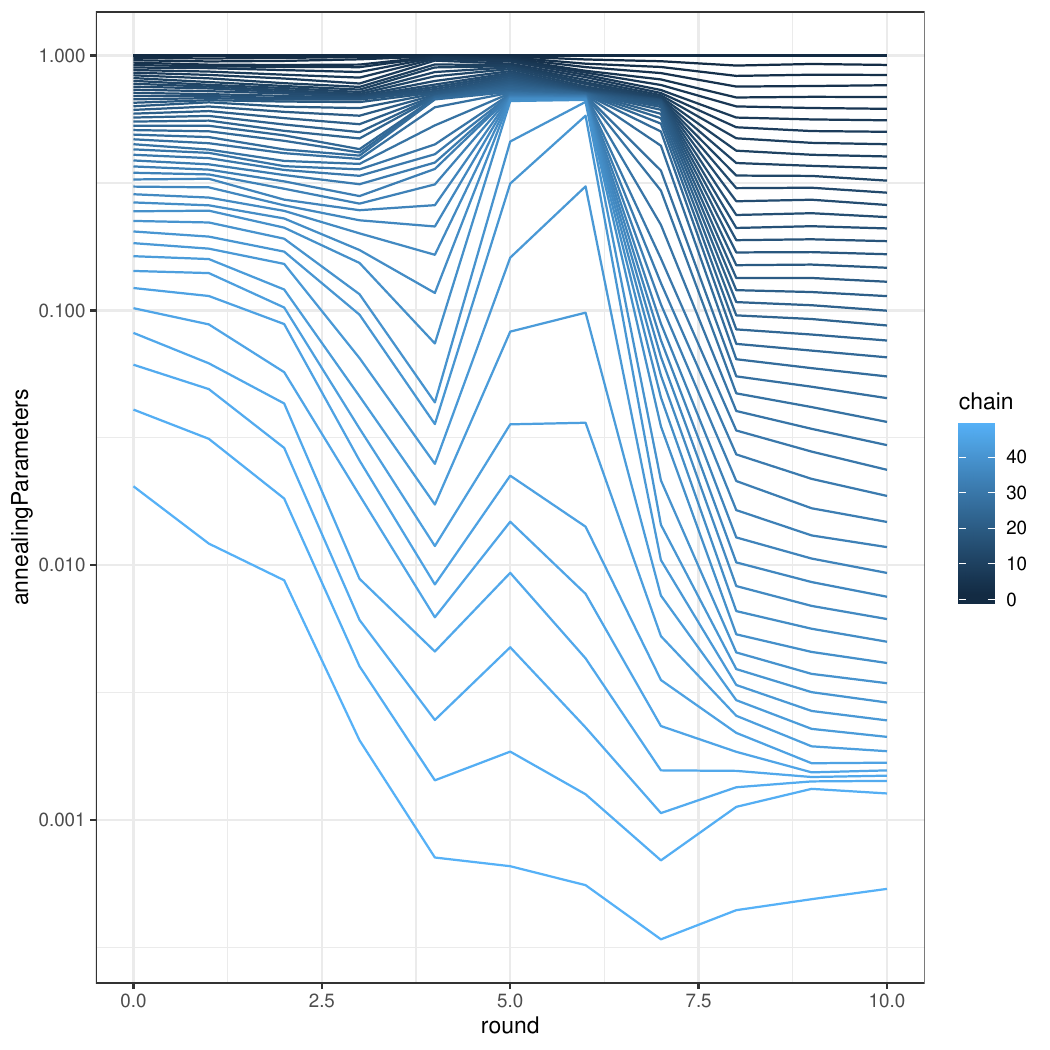}
	\end{center}
	\caption{NRPT diagnostics on the copy number inference example. From left to right: (1) Estimate of the global communication barrier $\hat \Lambda$ as a function of the scehdule optimization round. (2) Average swap acceptance probability across the 50 chains as a function of the schedule optimization round. (3) Estimated communication barrier $\lambda$ output by NRPT.  (4) Learning curve for the annealing parameters (ordinate axis, log scale) for the $50$ NRPT chains (colours) as a function of the schedule optimization round (abscissa).}
	\label{fig:chromo-diag}
\end{figure}
\end{appendices}

\end{document}